\documentclass[12pt,a4paper]{article}
\usepackage[english]{babel}
\usepackage{float}
\usepackage{amsmath, amssymb, amsthm, epsfig, color, mathtools, bbm, dsfont, mathrsfs, bm}
\usepackage[latin1]{inputenc}
\usepackage[T1]{fontenc}
\usepackage{indentfirst}
\usepackage[shortlabels]{enumitem}
\usepackage[title]{appendix}
\usepackage{lmodern}       
\usepackage{comment}  
\usepackage{textcomp, upgreek}
\usepackage{setspace}
\usepackage{soul}
\usepackage{accents}
\usepackage{hyperref}
\usepackage{geometry}
\usepackage{xcolor}

\usepackage{xr}   
\usepackage{graphicx} 
\usepackage[export]{adjustbox}
\usepackage{subcaption} 
\usepackage{mathtools} 
\usepackage{siunitx}
\usepackage{bm} 
\usepackage{comment}

\usepackage{addfont}
\addfont{OT1}{rsfs10}{\rsfs}

\usepackage{tikz, tkz-euclide}
\usetikzlibrary{patterns}
\usetikzlibrary{arrows.meta}
\usetikzlibrary{calc}
\usetikzlibrary{decorations.markings}
\usepackage{fancybox}

\usepackage{graphicx}
\graphicspath{ {./images/} }
\newcommand{\Z}{\mathbb{Z}}

\newcommand{\lab}{\mathrm{lab}}

\newcommand{\diam}{\mathrm{diam}}
\newcommand{\dis}{\mathrm{dist}}

\newcommand{\I}{\mathrm{I}}
\newcommand{\Sp}{\mathrm{sp}}

\newcommand{\R}{\mathcal{R}}
\newcommand{\fint}{\partial_{\mathrm{in}}}
\newcommand{\fext}{\partial_{\mathrm{ex}}}
\newcommand{\C}{\mathscr{C}}

\newcommand{\calP}{\mathcal{P}}

\def\d{\mathop{\textrm{\rm d}}\nolimits}                  
\def\Int{\mathop{\textrm{\rm Int}}\nolimits}                
\def\Ext{\mathop{\textrm{\rm Ext}}\nolimits}                  

\newcommand{\be}{\begin{equation}}
\newcommand{\ee}{\end{equation}}

\numberwithin{equation}{section}

  \newcounter{dummy} \numberwithin{dummy}{section}
  \theoremstyle{plain}
  \newtheorem*{theorem*}        {Theorem}
	\newtheorem*{conjecture*}   {Conjecture}
  \newtheorem{theorem}[dummy]          {Theorem}
  \newtheorem{lemma}[dummy]              {Lemma}
  \newtheorem*{lemma*}          {Lemma}
    
  \newtheorem{corollary}[dummy]           {Corollary}
  \newtheorem{proposition}[dummy]       {Proposition}
   
  \newtheorem{remark}[dummy]           {Remark}
  
  \theoremstyle{remark}
  \theoremstyle{definition}
   \newtheorem{definition}[dummy]          {Definition}
\makeatletter

\newcommand\longleftrightarrowfill@{%
  \arrowfill@\leftarrow\relbar\rightarrow}
\makeatother

\definecolor{Red}{cmyk}{0,1,1,0}

\definecolor{Blue}{cmyk}{1,1,0,0}

\definecolor{DarkBlue}{rgb}{0.1,0.1,0.5}
\definecolor{Red}{rgb}{0.9,0.0,0.1}
\definecolor{DarkGreen}{rgb}{0.10,0.50,0.10}
\definecolor{DarkRed}{rgb}{0.50,0.10,0.10}
\definecolor{bleu}{RGB}{0,140,189}%

\newgeometry{vmargin={15mm}, hmargin={12mm,17mm}}

\DeclareMathOperator{\s}{\mathrm{sp}}

\DeclarePairedDelimiter\ceil{\lceil}{\rceil} 
\DeclarePairedDelimiter\floor{\lfloor}{\rfloor} 

\begin{document}

\begin{center}
{\LARGE Phase Transitions in Multidimensional Long-Range Random Field Ising Models}
\vskip.5cm
Lucas Affonso$^{1}$, Rodrigo Bissacot$^{1,2}$, Jo{\~a}o Maia$^{1,3}$
\vskip.3cm
\begin{footnotesize}
$^{1}$ Institute of Mathematics and Statistics (IME-USP), University of S\~{a}o Paulo, Brazil\\
$^{2}$ Faculty of Mathematics and Computer Science, Nicolaus Copernicus University, Poland\\
$^{3}$ Beijing International Center for Mathematical Research (BICMR), Peking University, China\\
\end{footnotesize}
\vskip.1cm
\begin{scriptsize}
emails: lucas.affonso.pereira@gmail.com, rodrigo.bissacot@gmail.com, maia.joaovt@gmail.com
\end{scriptsize}

\end{center}

\begin{abstract}
We extend a recent argument by Ding and Zhuang from nearest-neighbor to long-range interactions and prove the phase transition in a class of ferromagnetic random field Ising mo-dels. Our proof combines a generalization of Fr\"ohlich-Spencer contours to the multidimensional setting, proposed by two of us, with the coarse-graining procedure introduced by Fisher, Fr\"ohlich, and Spencer. Our result shows that the Ding-Zhuang strategy is also useful for interactions $J_{xy}=|x-y|^{- \alpha}$ when $\alpha > d$ in dimension $d\geq 3$ if we have a suitable system of contours, yielding an alternative approach that does not use the Renormalization Group Method (RGM), since Bricmont and Kupiainen suggested that the RGM should also work on this generality. We can consider i.i.d. random fields with Gaussian or Bernoulli distributions.
\end{abstract}

\section{Introduction}
The problem of the presence or absence of phase transition is central in statistical mechanics. To prove the existence of a phase transition, the standard idea is to define a notion of contour and use \textit{Peierls' argument} \cite{Peierls.1936}. In the Ising model \cite{Ising_25}, particles of the system interact only with their nearest neighbors. On ferromagnetic long-range Ising models \cite{Dyson.69, Kac_Thompson_69}, there is interaction between each pair of spins in the lattice. The Hamiltonian of the model is given formally by
\begin{equation*}
    H(\sigma) = - \sum_{x,y\in \Z^d}J_{xy}\sigma_x\sigma_y,
\end{equation*}
where $J_{xy}=J|x-y|^{-\alpha}$, $J>0$ and $\alpha > d$. It is well-known that phase transition in dimension 2 for Ising models with nearest-neighbors implies phase transition for long-range interactions when $d\geq 2$, as a consequence of correlation inequalities. For the one-dimensional lattice, it is known that short-range models do not present phase transition \cite{georgii.gibbs.measures}. In the long-range case, a different behavior was conjectured depending on the exponent $\alpha$ (see \cite{Kac_Thompson_69}), but the problem was challenging.

In dimension $d=1$, phase transition was proved first in 1969 by Dyson \cite{Dyson.69}, for $\alpha \in (1,2)$, by proving phase transition in an auxiliary model, known nowadays as the \emph{Dyson model} or hierarchical model. Dyson's approach fails exactly on the critical exponent $\alpha=2$. It was already known that for $\alpha>2$ uniqueness holds \cite{georgii.gibbs.measures}. In 1982, Fr{\"o}hlich and Spencer \cite{Frohlich.Spencer.82} introduced a notion of one-dimensional contours and then applied Peierls' argument to show phase transition for the critical value $\alpha = 2$. These contours were inspired by the multiscale techniques previously introduced to study the Berezinskii-Kosterlitz-Thouless transition in two-dimensional continuous spin systems \cite{FS81}. Later, Cassandro, Ferrari, Merola and Presutti  \cite{Cassandro.05} extended the contour argument previously available for $\alpha=2$ to exponents $\alpha\in (3-\frac{\ln 3}{\ln 2}, 2]$, with the additional restriction that the nearest-neighbor interaction is strong, i.e.,  ${J(1)\gg 1}$; this restriction was removed for a subclass of interactions in \cite{Bissacot_Endo_Enter_Kimura_18}. Recently, in \cite{Affonso_Bissacot_Corsini_Welsch_2024}, the Fr{\"o}hlich-Spencer argument was extended to the entire region $1<\alpha<2$ for any $J(1)>0$. Further results were obtained using contour arguments, such as the decay of correlations, cluster expansions, and phase transition with random interactions; some references with these results are \cite{Cassandro.Merola.Picco.17, Cassandro.Merola.Picco.Rozikov.14, Imbrie.82, Imbrie.Newman.88, Johansson.91}. 

In the multidimensional setting ($d\geq 2$), Ginibre, Grossmann, and Ruelle \cite{Ginibre.Grossmann.Ruelle.66} proved the phase transition for $\alpha > d+1$ using an enhanced version of Peierls' argument and the usual contours. Park used a different notion of contour for long-range systems in \cite{Park.88.I, Park.88.II}, extending the Pirogov-Sinai theory available for short-range interactions assuming $\alpha > 3d+1$, and he can also consider Potts models and models without symmetry with his methods. Some results in the literature suggest that truly long-range effects appear only when $d < \alpha \leq d+1$, see \cite{Biskup_Chayes_Kivelson_07}. Recently, inspired by the ideas from Fr{\"o}hlich and Spencer in \cite{FS81, Frohlich.Spencer.82}, Affonso, Bissacot, Endo, and Handa \cite{Affonso.2021} introduced a multiscale multidimensional contour and proved phase transition by a contour argument in the whole region $\alpha > d$. They can consider long-range Ising models with deterministic decaying fields, first introduced in the context of nearest-neighbor interactions in \cite{Bissacot_Cioletti_10}. For such models, the lack of analyticity of the free energy does not imply a phase transition since these models have the same free energy as the models with zero field, and it is expected that slowly decaying fields imply uniqueness, see \cite{Bissacot_Cass_Cio_Pres_15}. In this setting, a contour argument is useful for proofs of phase transitions as well as for uniqueness; some papers with models with deterministic decaying fields are \cite{Aoun_Ott_Velenik_23, Bissacot_Cass_Cio_Pres_15, Bissacot_Cioletti_10, Bissacot_Endo_Enter_Kimura_18, Cioletti_Vila_2016}. For a more detailed version of \cite{Affonso.2021} we recommend the Ph.D. thesis of one of the authors \cite{Affonso_23}.

The Random Field Ising model (RFIM) \cite{Imry.Ma.75} is the nearest-neighbor Ising model with an additional external field given by a family of i.i.d. Gaussian random variables $(h_x)_{x\in\Z^d}$ with mean 0 and variance 1. Formally, the Hamiltonian of the model is given by
\begin{equation*}
    H(\sigma) = - \sum_{\substack{x,y\in \Z^d \\|x-y|=1}}J\sigma_x\sigma_y  - \varepsilon\sum_{x\in\Z^d}h_x\sigma_x,
\end{equation*}
where $J>0$, $\varepsilon>0$, and $d \geq 1$. A detailed account of the history of the phase transition problem for this model and detailed proofs are presented by Bovier in \cite{Bovier.06}. Here, we give a brief overview.

During the 1980s, the question of the specific dimension where the phase transition for the RFIM should happen attracted much attention and was a topic of heated debate. Two convincing arguments divided the physics community. One of them, due to Imry and Ma \cite{Imry.Ma.75}, was a non-rigorous application of the Peierls' argument together with the use of the isoperimetric inequality. The key idea of Peierls' argument is to define a notion of contour and calculate the energy cost of \textit{erasing} each contour, i.e., the energy cost of flipping all spins inside the contour. When there is no external field, the energy necessary to flip the spins in a region $A\subset \Z^d$ is of the order of the boundary $|\partial A|$. When we add an external field, we get an extra cost depending on this field. Imry and Ma argued that this cost should be approximately $\sqrt{|A|}$. By the isoperimetric inequality, $\sqrt{|A|}\leq |\partial A|^{\frac{d}{2(d-1)}}$, which is strictly smaller than $|\partial A|$ for all regions only when $d\geq 3$, so this should be the region where phase transition occurs. The other argument, due to Parisi and Sourlas \cite{Parisi.Sourlas.79}, based on dimensional reduction \cite{Aharony_Imry_Ma_76} and supersymmetry arguments, predicted that the $d$-dimensional RFIM would behave like the $d-2$-dimensional nearest-neighbor Ising model, therefore presenting a phase transition only when $d\geq 4$. 

The question was settled by two celebrated papers showing that Imry and Ma's prediction was correct. First, in 1988, Bricmont and Kupiainen \cite{Bricmont.Kupiainen.88} showed that there is phase transition almost surely in $d\geq3$, for low temperatures and $\varepsilon$ small enough. Their proof uses a rigorous renormalization group analysis, and it is considered involved. Still, they suggested that the result works for any model with a suitable contour representation and a centered sub-Gaussian external field. Later on, in 1990, Aizenman and Wehr \cite{Aizenman.Wehr.90} proved uniqueness for $d\leq 2$. For detailed proofs of these results, we refer the reader to \cite{Bovier.06} (see also \cite{Berretti.85, Camia.18, Frohlich.Imbre.84,  Klein.Masooman.97} for more uniqueness results). 

Recently, Ding and Zhuang \cite{Ding2021}, provided a simpler proof of the phase transition, not using RGM. In addition, Ding, Liu, and Xia \cite{Ding_Liu_Xia_24} proved that if $\beta_c(d)$ is the critical inverse of the temperature of the Ising model with no field, for all $\beta>\beta_c(d)$ there exists a critical value $\varepsilon_0(d, \beta)$ such that the RFIM with $\varepsilon \leq \varepsilon_0$ presents phase transition. 

In the present paper, we are considering a long-range Ising model with a random field, whose Hamiltonian is given formally by
\begin{equation*}
    H(\sigma) = - \sum_{x,y\in \Z^d}J_{xy}\sigma_x\sigma_y - \varepsilon\sum_{x\in\Z^d}h_x\sigma_x,
\end{equation*}
where $J_{xy}=J|x-y|^{-\alpha}$, $J, \varepsilon>0$, $\alpha > d$, $d\geq 3$, and $(h_x)_{x\in\Z^d}$ being a family of i.i.d. Gaussian random variable with mean 0 and variance 1. The only rigorous result on phase transition in the long-range setting is for the one-dimensional long-range Ising model with a random field, by Cassandro, Orlandi, and Picco \cite{Cassandro.Picco.09}. They used the contours of \cite{Cassandro.05} to show the phase transition for the model when $\alpha\in (3-\frac{\ln 3}{\ln 2}, \frac{3}{2})$, under the assumption $J(1) \gg 1$. We stress that, as remarked by Aizenman, Greenblatt, and Lebowitz \cite{Aizenman_Greenblatt_Lebowitz_2012}, although their argument does not work for the whole region of the exponent $\alpha$, the phase transition holds for values close to the critical value $\alpha=3/2$, since by the Aizenman-Wehr theorem we know that there is uniqueness for $\alpha\geq 3/2$.
\begin{remark}
 After the publication of this paper on Arxiv, in \cite{Ding_Huang_Maia_25}, one of the authors, together with Ding and Huang complemented our results, showing the existence of phase transition for low-dimensions $d=1,2$ throughout the entire region $d<\alpha<3d/2$, and also for the two-dimensional critical value $\alpha = 3$. Their argument in $d=2$ uses our contours and follows our arguments somewhat closely for $2<\alpha<3$. Now we have the complete phase diagram for the RFIM. Moreover, using the contours introduced in the present paper, the phase transition for the Potts model with decaying and random fields was proved in \cite{Affonso_Bissacot_Faria_Welsch_25}, as well as the convergence of the cluster expansion and the polynomial decay of the truncated correlation functions in low temperatures in \cite{Affonso_Bissacot_Maia_Welsch_25}.

\end{remark}

The argument from Ding and Zhuang in \cite{Ding2021}, for $d\geq3$, involves controlling the probability of a bad event, which is related to controlling the quantity $$\sup_{\substack{0\in A\subset\Z^d \\ A \text{ connected }}}\frac{\sum_{x\in A}h_x}{|\partial A|},$$ known as the greedy animal lattice normalized by the boundary. The greedy animal lattice normalized by the size, instead of the boundary, was extensively studied for general distributions of $(h_x)_{x\in\Z^d}$, see \cite{Cox_Gandolfi_Griffin_Kesten_93, Gandolfi_Kesten_94, Hammond_06, Martin_02}. When we normalize by the boundary, an argument by Fisher, Fr\"{o}hlich and Spencer \cite{FFS84} shows that the expected value of the greedy animal lattice is finite. In dimension $d=2$, the expected value is not finite, see \cite{Ding_Wirth_23}. The supremum is taken over connected regions containing the origin since the interiors of the usual Peierls contours are of this form.

For the long-range model, the interior of the contours is not necessarily connected. In fact, long-range contours may have considerably large diameters with respect to their size, so their interiors can be very sparse. Our definition of the contours is strongly inspired by the $(M,a,r)$-partition in \cite{Affonso.2021}. They are constructed using a multiscale procedure that ensures that the contours have no cluster with small density. With them, we generalize the arguments by Fisher-Fr\"{o}hlich-Spencer \cite{FFS84} and prove that the expected value of the greedy animal lattice is finite, even considering regions not necessarily connected. Then, we prove the phase transition for $d\geq 3$. Our main result can be stated as
\begin{theorem*}
Given $d\geq 3$, $\alpha>d$, there exists $\beta_c\coloneqq\beta(d, \alpha)$ and $\varepsilon_c\coloneqq\varepsilon(d, \alpha)$ such that, for $\beta> \beta_c$ and $\varepsilon\leq \varepsilon_c$, the extremal Gibbs measures $\mu_{\beta, \varepsilon}^+$ and $\mu_{\beta, \varepsilon}^-$ are distinct, that is, $\mu_{\beta, \varepsilon}^+ \neq \mu_{\beta, \varepsilon}^-$ $\mathbb{P}$-almost surely. Therefore, the long-range random field Ising model presents a phase transition.
\end{theorem*}

\textit{Ideas of the proof:}  We first introduce a suitable notion of contour, for which we can control both the energy cost of erasing a contour (Proposition \ref{Prop: Cost_erasing_contour}) and the number of contours of a fixed size that surround the origin (Corollary \ref{Cor: Bound_on_C_0_n}). Our contours, as in the Pirogov-Sinai theory, are composed of a support that represents the incorrect points of a configuration and the plus and minus interior, which are the regions inside the contours. We denote $\mathcal{C}_0$ the set of all contours surrounding the origin and $\mathcal{C}_0(n)$ the set of contours in $\mathcal{C}_0$ with $n$ points in the support. We also denote $\I_-(n)$, the set of all subsets of $\Z^d$ that are the minus interior of a contour in $\mathcal{C}_0(n)$. 

By the Ding-Zhuang method, we can show that the phase transition follows from controlling the probability of the bad event
$$\mathcal{E}^c\coloneqq \left\{\sup_{\substack{\gamma\in\mathcal{C}_0}} \frac{\Delta_{\I_-(\gamma)}(h)}{c_2|\gamma|} > \frac{1}{4}\right\},$$
where $|\gamma|$ is the size of the support of a contour $\gamma$ and $(\Delta_A)_{A\Subset \Z^d}$ is a family of functions that, by Lemma \ref{Lemma: Concentration.for.Delta.General}, have the same tail of $\sum_{x\in A}h_x$, and the distribution of $\Delta_A(h) - \Delta_{A^\prime}(h)$ is the same as $\Delta_{A\Delta A^\prime}(h)$, for all $A, A^\prime\in\Z^d$ finite, see \cite{Ding2021}. 

In the nearest-neighbor case, the contours are $d-1$ connected objects, so all the interiors $\I_-(\gamma)$ are connected, and the two properties in Lemma \ref{Lemma: Concentration.for.Delta.General} together with the coarse-graining procedure introduced by Fisher, Fr\"ohlich, and Spencer in \cite{FFS84} is enough to control $\mathbb{P}(\mathcal{E}^c)$. In the long-range setting, the arguments are more involved. Given a family of scales $(2^{r\ell})_{\ell\geq 0}$, with $r$ being a suitable constant, we can partition $\Z^d$ into disjoint fitting cubes with sides $2^{r\ell}$, see Figure \ref{Fig: Cubes}. Each such cube is called an $r\ell$-cube, and all cubes throughout our analysis will be of this form unless stated otherwise. The strategy of the coarse-graining argument is to, at each scale, approximate each interior $\I_-(\gamma)$ by a simpler region $B_\ell(\gamma)$, formed by the union of disjoint cubes with side length $2^{r\ell}$, see Figure \ref{Fig: Figura7}. The argument follows once you have two estimations: on the error of this approximation and on the size of the set $B_\ell (\mathcal{C}_0(n)) \coloneqq \{B\subset\Z^d: B=B_\ell(\gamma), \text{ for some }\gamma \in\mathcal{C}_0(n)\}$ containing all regions that are approximations of contours in $\mathcal{C}_0(n)$.
\vspace{-2cm}
\begin{figure}[ht]
     \centering
     \input{Figures/Figura.6}
        \caption{The figure on the left represents the minus interior of a contour $\gamma$, while the central figure represents its approximation. The picture on the right depicts the error of the approximation, consisting of both the cross-hatched regions not covered by the gray area and the gray areas not intersected by the cross-hatched region.}
    \label{Fig: Figura7}
\end{figure}

In Corollary \ref{Cor: Bound_diam_B_ell}, we show that $|B_\ell(\gamma)\Delta \I_-(\gamma)|\leq c2^{r\ell}|\gamma|$, so the error in the approximation is not too large. This bound follows \cite{FFS84} closely since we approximate the interiors in the same way. Then, we need to estimate $|B_\ell (\mathcal{C}_0(n))|$, which is done by a fairly distinct argument from the short-range case. The difficulty of the proof comes from the fact that our contours may be disconnected. As the regions $B_\ell(\gamma)$ are the union of disjoint cubes, they are, up to a constant, determined by $\fint \mathfrak{C}_{\ell}(\gamma)$, the collection of cubes needed to cover $B_{\ell}(\gamma)$ that share a face with a cube that does not intersect $B_{\ell}(\gamma)$, see Figure \ref{Fig: Figura8}. 

\begin{figure}[h]
    \centering
    \input{Figures/Figura.7}
    \caption{The region $B_\ell(\gamma)$ is the approximation of the interior in Figure \ref{Fig: Figura7}, and $\fint \mathfrak{C}_{r\ell}(\gamma)$ is the collection containing all the doted cubes.}
    \label{Fig: Figura8}
\end{figure}

Following the argument of Fisher-Fr\"ohlich-Spencer we show that $|\fint \mathfrak{C}_{r\ell}(\gamma)|\leq M_{|\gamma|,\ell}\coloneqq c {|\gamma|}{2^{-r\ell(d-1)}}$ (see Proposition \ref{Proposition1}), so we can bound $|B_{\ell}(\mathcal{C}_0(n))|$ by counting all possible choices of at most $M_{n,\ell}$ non-intersecting cubes with side length $2^{r\ell}$ in $\Z^d$, with a suitable restriction. When the contours are connected, this restriction is that all cubes must be close to a surface with size $n$, and the proof follows once you can use that the minimal path connecting all the cubes has length at most $cn$. This is not true for our contours, so we need a different strategy. 

Let $\C_{rL}(\gamma)$ be the smallest collection of cubes, in the $rL$ scale, needed to cover $\gamma$. The property we will use is that, for all $L\geq \ell$, every cube in $\fint \mathfrak{C}_\ell(\gamma)$ is covered or is next to a cube in  $\C_{rL}(\gamma)$, see Figure \ref{Fig: Figura9}. As every cube with side length $2^{rL}$ contains $2^{rd(L-\ell)}$ cubes with side length $2^{r\ell}$, any fixed collection $\C_{rL}(\gamma)$ covers $2^{rd(L-\ell)}|\C_{rL}(\gamma)|$ cubes in the $r\ell$ scale.

\begin{figure}[ht]
    \centering
    \input{Figures/Figura.8}
    \caption{The contour $\gamma$ is the one the originates the interior $\I_-(\gamma)$ in Figure \ref{Fig: Figura7}. The dotted cubes are the cubes in $\fint \mathfrak{C}_\ell(\gamma)$ and the larger cubes are the ones needed to cover the contour $\gamma$.}
    \label{Fig: Figura9}
\end{figure}

Roughly, we can bound $|B_{\ell}(\mathcal{C}_0(n))|$ by counting all the possible choices of at most $M_{n,\ell}$ cubes in the $2^{r\ell}$ scale that are covered by a collection in $\C_{rL}(\mathcal{C}_0(n))\coloneqq \{\C_{rL} : \C_{rL} = \C_{rL}(\gamma) \text{ for some }\gamma\in\mathcal{C}_0(n)\}$, that is, we can bound as follows:

\begin{equation*}
    |B_{\ell}(\mathcal{C}_0(n))|\leq \sum_{\C_{rL}\in \C_{rL}(\mathcal{C}_0(n))} \sum_{M=1}^{M_{n,\ell}}\binom{2^{rd(L-\ell)}|\C_{rL}|}{M}.
\end{equation*}

The construction of our contours allow us to upper bound $|\C_{rL}(\gamma)|$  and $|\C_{rL}(\mathcal{C}_0(n))|$, see Proposition \ref{Prop. Bound.on.C_rl(gamma)} and Proposition \ref{Prop: Bound_on_rl_coverings}, and we are able to get the same bound for $|B_{\ell}(\mathcal{C}_0(n))|$ as in \cite{FFS84} by making an appropriate choice of the large scale $L = L(\ell)$ depending on the smaller one. From this, the control on the probability of the bad event follows as an application of Dudley's entropy bound.

This paper is divided as follows. In Section 2, we define the model and the contours, and suitable generalizations to the constructions in \cite{Affonso.2021} are introduced. In Section 3, we define the bad event of the external field and prove that it occurs with a small probability. In this section, the generalizations of the coarse-graining procedure are presented. Finally, in Section 4, we present the proof of the phase transition.

This paper is contained in the Ph.D. thesis of J. Maia \cite{Maia_24}.

\section{Preliminaries}  

    \subsection{The model}
    The set of configurations of the long-range Ising model is, as usual, $\Omega \coloneqq \{-1,1\}^{\Z^d}$. However, each spin interacts with all others, not only its nearest neighbors, so the interaction $\{J_{xy}\}_{x,y\in\Z^d}$ is defined as

\begin{equation}\label{Long-Range Interaction}
    J_{xy} = \begin{cases}
                   \frac{J}{|x-y|^\alpha} &\text{ if }x\neq y,\\
                   0                &\text{otherwise,} 
              \end{cases}
\end{equation}
where $J >0$, $\alpha>d$ and the distance $|x-y|$ is given by the $\ell_1$-norm. We write $\Lambda\Subset \Z^d$ to denote a finite subset of $\Z^d$
. Fixed such $\Lambda$, the \textit{local configurations} is given by $\Omega_\Lambda\coloneqq \{-1,1\}^\Lambda$. Moreover, given ${\eta\in\Omega}$, the set of local configurations with $\eta$ boundary condition is ${\Omega_\Lambda^\eta\coloneqq \{\sigma\in\Omega : \sigma_x=\eta_x, \text{ }\forall x\in\Lambda^c\}}$. The \textit{local Hamiltonian of the random field long-range Ising model} in $\Lambda\Subset\Z^d$ with $\eta$-boundary condition is $H_{\Lambda; \varepsilon h}^{\eta}:  \Omega_\Lambda^\eta \to \mathbb{R}$, given by 

\begin{equation}
    H_{\Lambda; \varepsilon h}^{\eta}(\sigma) \coloneqq -\sum_{x,y\in\Lambda} J_{xy}\sigma_x\sigma_y - \sum_{x\in \Lambda, y\in\Lambda^c} J_{xy}\sigma_x\eta_y - \sum_{x\in\Lambda} \varepsilon h_x\sigma_x,
\end{equation}
where the external field is a family $\{h_x\}_{x\in\Z^d}$ of i.i.d. random variables in $(\widetilde{\Omega}, \mathcal{A}, \mathbb{P})$, and every $h_x$ has a standard normal distribution\footnote{ Our results also hold for more general distributions of $h_x$, see Remarks \ref{Rmk: Bernoulli_external_field} and \ref{Rmk: More.general.h_x}. }. The parameter $\varepsilon >0$ controls the variance of the external field. Given $\Lambda\Subset\Z^d$, consider $\mathscr{F}_\Lambda$ the $\sigma$-algebra generated by the cylinders sets supported in $\Lambda$ and $\mathscr{F}$ the $\sigma$-algebra generated by finite union of cylinders. One of the main objects of study in classical statistical mechanics are the \textit{finite volume Gibbs measures}, which are probability measures in $(\Omega, \mathscr{F})$, given by 
    \begin{equation}
        \mu_{\Lambda;\beta, \varepsilon h}^\eta(\sigma) \coloneqq \mathbbm{1}_{\Omega_\Lambda^\eta}(\sigma)\frac{e^{-\beta H_{\Lambda, \varepsilon h}^{\eta}(\sigma)}}{Z_{\Lambda; \beta, \varepsilon}^{\eta}(h)},
    \end{equation}
where $\beta>0$ is the inverse temperature and $Z_{\Lambda; \beta, \varepsilon}^{\eta}$ is the \textit{partition function}, defined as 

\begin{equation}
    Z_{\Lambda; \beta, \varepsilon}^{\eta}(h)\coloneqq \sum_{\sigma\in\Omega_\Lambda^\eta} e^{-\beta H_{\Lambda, \varepsilon h}^{\eta}(\sigma)}.
\end{equation}
Since the external field is random, the Gibbs measures are random variables. To make the dependence of $\mu_{\Lambda;\beta, \varepsilon h}^\eta$ on $\widetilde{\Omega}$ explicit, we write $\mu_{\Lambda;\beta, \varepsilon h}^\eta[\omega]$, with $\omega$ being a general element of $\widetilde{\Omega}$. Two particularly important boundary conditions are given by the configurations $\eta_{+} \equiv +1$($\eta_{x}= +1, \forall x \in \mathbb{Z}^{d})$ and analogously $\eta_{-} \equiv -1$, and configurations are called $+$ and $-$ boundary conditions, respectively. For these boundary conditions, we can $\mathbb{P}$-almost surely define the infinite volume measures by taking the weak*-limit
\begin{equation}
    \mu_{\beta,\varepsilon h}^{\pm}[\omega] \coloneqq \lim_{n\to\infty} \mu_{\Lambda_n;\beta, \varepsilon h}^{\pm}[\omega],
\end{equation}
where $(\Lambda_n)_{n\in\mathbb{N}}$ is any sequence invading $\Z^d$, that is, for any subset $\Lambda\Subset\mathbb{Z}^d$, there exists $N=N(\Lambda)>0$ such that $\Lambda\subset\Lambda_n$ for every $n>N$. To have more than one Gibbs measure, it is enough to show that $\mu_{\beta,\varepsilon h}^{+}\neq  \mu_{\beta,\varepsilon h}^{-}$, $\mathbb{P}$-almost surely, see \cite[Theorem 7.2.2]{Bovier.06}.

    \subsection{The contours} 
    
    Contours were first defined in the seminal paper of R. Peierls \cite{Peierls.1936}, where he introduced these geometrical objects to prove phase transition in the Ising model for $d\geq 2$. This technique is known nowadays as the \textit{Peierls' argument}. One of the most successful extensions of this argument was made by S. Pirogov and Y. Sinai \cite{Pirogov.Sinai.75}, and extended by Zahradn{\'i}k \cite{Zahradnik.84}. This is known as the \textit{Pirogov-Sinai} theory, which can be used in models with short-range interactions and finite state spaces, even without symmetries. The Pirogov-Sinai Theory was one of the achievements cited when Yakov Sinai received the Abel Prize \cite{Sinai_Abel_Prize}.

For long-range models, using the usual Peierls' contours with plaquettes of dimension $d-1$, Ginibre, Grossman, and Ruelle, in \cite{Ginibre.Grossmann.Ruelle.66}, proved phase transition for $\alpha > d+1$.  Park, in \cite{Park.88.I,Park.88.II}, considered systems with two-body interactions satisfying $|J_{xy}|\leq |x-y|^{-\alpha}$ for $\alpha > 3d+1$, and extended the Pirogov-Sinai theory for this class of models.  Fr{\"o}hlich and Spencer, in \cite{Frohlich.Spencer.82}, proposed a different contour definition for the one-dimensional long-range Ising models. Roughly speaking, collections of intervals are the new contours but arranged in a particular way. When they are sufficiently far apart, the collections of intervals are deemed as different contours, while collections of intervals close enough are considered a single contour. Note that this definition drastically contrasts with the notion of contour in the multidimensional setting, since now they are not necessarily connected objects of the lattice. This fact implies that the control of the number of contours for a fixed size could be much more challenging. 

Inspired by such contours, Affonso, Bissacot, Endo, and Handa proposed a definition of contour extending the contours of Fr{\"o}hlich and Spencer to any dimension $d\geq 2$, see \cite{Affonso.2021}. With these contours, they were able to use Peierls' argument to show phase transition in the whole region $\alpha>d$, with $d\geq 2$. We modify the contour definition of \cite{Affonso.2021} using a similar partition through multiscale methods. 

We choose to use the new definition of contours for two main reasons: the definition is simpler, and it produces less sparse contours. This latter property can be expressed in several different ways. For example, one can show that a contour with size $n$ has diameter at most $n^C$, for a suitable constant $C>0$ (see Proposition 4.2.10 in \cite{Maia_24}). With the previous definition, one could only guarantee an exponential upper bound in $n$, which is not enough to replicate the results of \textcolor{red}{\cite{Affonso_Bissacot_Maia_Welsch_25}}, see e.g. \cite[Proposition 3.6]{Affonso_Bissacot_Maia_Welsch_25}. Relevant to our application, we can improve the upper bound on the number of cubes needed to cover a contour, from a polynomial bound to an exponential bound, see Propositions \ref{Prop. Bound.on.C_rl(gamma)_Lucas} and \ref{Prop. Bound.on.C_rl(gamma)}. With some modification of our arguments, we can replicate our results even adopting the contours of \cite{Affonso.2021} (in Remark \ref{Rmk: Adaptation_for_Mar_partition} we describe the key adaptations that must be made), but the same is not true for the two-dimensional arguments of \cite{Ding_Huang_Maia_25}. In this section, we describe our contours.


\begin{definition}\label{def1}
	Given $\sigma \in \Omega$, a point $x \in \Z^d$ is called \emph{+ (or - resp.)} \emph{correct} if $\sigma_y = +1$, (or $-1$, resp.) for all points $y$ such that $|x-y|\leq 1$. The \emph{boundary} of $\sigma$, denoted by $\partial \sigma$, is the set of all points in $\Z^d$ that are neither $+$ nor $-$ correct.
\end{definition}

The boundary of a configuration is not finite in general, it can even be the whole lattice $\Z^d$. To avoid this problem, we will restrict our attention to configurations with finite boundaries. Such configurations, by definition of incorrectness, satisfy $\sigma \in \Omega^+_\Lambda$ or $\sigma \in \Omega^-_\Lambda$ for some $\Lambda\Subset \Z^d$. We also defined, for each $\Lambda\Subset \Z^d$, $\Lambda^{(0)}$ as the unique unbounded connected component of $\Lambda^c$. The \textit{volume} of $\Lambda$ is defined as $V(\Lambda)\coloneqq \Z^d\setminus \Lambda^{(0)}$. The \textit{interior} of $\Lambda$ is $\I(\Lambda)\coloneqq \Lambda^c\setminus \Lambda^{(0)}$.

The usual definition of contours in Pirogov-Sinai theory considers only the connected subsets of the boundary $\partial \sigma$. We have to proceed differently for long-range models since every point in the lattice interacts with all the others. The definition below is strongly inspired in \cite{Affonso.2021} and allows contours to be disconnected (as in one-dimensional long-range models); in return, we can control the interaction between two contours and also the probabilities of bad events that we will introduce in Section 3.
    
    \begin{definition}\label{Def: delta-partiton}
    Let $M>0$ and $a,\delta >d$. For each $A\Subset\Z^d$, a set $\Gamma(A) \coloneqq \{\overline{\gamma} : \overline{\gamma} \subset A\}$ is called a $(M,a,\delta)$-\emph{partition} when the following two conditions are satisfied.
	\begin{enumerate}[label=\textbf{(\Alph*)}, series=l_after] 
		\item They form a partition of $A$, i.e.,  $\bigcup_{\overline{\gamma} \in \Gamma(A)}\overline{\gamma}=A$ and $\overline{\gamma} \cap \overline{\gamma}' = \emptyset$ for distinct elements of $\Gamma(A)$.  
		
		\item For all $\overline{\gamma}, \overline{\gamma}^\prime \in \Gamma(A)$ 
			\be\label{B_distance_2}
			\d(\overline{\gamma},\overline{\gamma}') > M\min\left \{|V(\overline{\gamma})|,|V(\overline{\gamma}')|\right\}^\frac{a}{\delta}.
			\ee
	\end{enumerate}
\end{definition}

\begin{remark}\label{Rmk: choice_of_a_and_delta}
    Our construction and the control of the energy works for any $d<\delta<\frac{a(\alpha - d)}{2}$ and  $a>\frac{2(d+1)}{(\alpha-d) \wedge 1}$. To simplify the calculations, we will take $\delta = d+1$ and  $a \coloneqq a(\alpha,d) = \frac{3(d+1)}{(\alpha-d) \wedge 1}$ from now on, so $\frac{a}{\delta} = \frac{3}{(\alpha-d) \wedge 1}$ and a $(M,a,\delta)$-partition will be called $(M,a)$-partition.
\end{remark}

In Figure \ref{Fig: Figura10} we give an example of a region $A\Subset \Z^d$ that is only one contour using the $(M,a,r)$-partition of \cite{Affonso.2021}, but can be partitioned into multiple components to form a $(M,a)$-partition. 

\begin{figure}[H]
    \centering
    \tikzset{every picture/.style={line width=0.75pt}} 

\begin{tikzpicture}[x=0.75pt,y=0.75pt,yscale=-1,xscale=1]

\draw  [fill={rgb, 255:red, 155; green, 155; blue, 155 }  ,fill opacity=0.63 ] (386.36,19.67) -- (459.53,19.67) -- (459.53,92.84) -- (386.36,92.84) -- cycle ;
\draw  [fill={rgb, 255:red, 255; green, 255; blue, 255 }  ,fill opacity=1 ] (392.34,25.65) -- (453.55,25.65) -- (453.55,86.86) -- (392.34,86.86) -- cycle ;
\draw  [dash pattern={on 4.5pt off 4.5pt}] (319.59,59.27) -- (386.01,59.27) -- (386.01,59.22) -- (319.59,59.22) -- cycle ;
\draw  [fill={rgb, 255:red, 155; green, 155; blue, 155 }  ,fill opacity=0.63 ] (526.82,19.67) -- (600,19.67) -- (600,92.84) -- (526.82,92.84) -- cycle ;
\draw  [fill={rgb, 255:red, 255; green, 255; blue, 255 }  ,fill opacity=1 ] (532.8,25.65) -- (594.02,25.65) -- (594.02,86.86) -- (532.8,86.86) -- cycle ;
\draw  [dash pattern={on 4.5pt off 4.5pt}] (460.05,59.27) -- (526.48,59.27) -- (526.48,59.22) -- (460.05,59.22) -- cycle ;
\draw  [fill={rgb, 255:red, 155; green, 155; blue, 155 }  ,fill opacity=0.63 ] (105,19.67) -- (178.18,19.67) -- (178.18,92.84) -- (105,92.84) -- cycle ;
\draw  [fill={rgb, 255:red, 255; green, 255; blue, 255 }  ,fill opacity=1 ] (110.98,25.65) -- (172.2,25.65) -- (172.2,86.86) -- (110.98,86.86) -- cycle ;
\draw  [fill={rgb, 255:red, 155; green, 155; blue, 155 }  ,fill opacity=0.63 ] (245.47,19.67) -- (318.64,19.67) -- (318.64,92.84) -- (245.47,92.84) -- cycle ;
\draw  [fill={rgb, 255:red, 255; green, 255; blue, 255 }  ,fill opacity=1 ] (251.45,25.65) -- (312.66,25.65) -- (312.66,86.86) -- (251.45,86.86) -- cycle ;
\draw  [dash pattern={on 4.5pt off 4.5pt}] (178.7,59.27) -- (245.13,59.27) -- (245.13,59.22) -- (178.7,59.22) -- cycle ;
\draw  [fill={rgb, 255:red, 155; green, 155; blue, 155 }  ,fill opacity=0.63 ] (190.49,176.67) -- (236.39,176.67) -- (236.39,222.57) -- (190.49,222.57) -- cycle ;
\draw  [fill={rgb, 255:red, 255; green, 255; blue, 255 }  ,fill opacity=1 ] (194.24,180.42) -- (232.64,180.42) -- (232.64,218.82) -- (194.24,218.82) -- cycle ;
\draw  [fill={rgb, 255:red, 155; green, 155; blue, 155 }  ,fill opacity=0.63 ] (278.6,176.67) -- (324.5,176.67) -- (324.5,222.57) -- (278.6,222.57) -- cycle ;
\draw  [fill={rgb, 255:red, 255; green, 255; blue, 255 }  ,fill opacity=1 ] (282.35,180.42) -- (320.75,180.42) -- (320.75,218.82) -- (282.35,218.82) -- cycle ;
\draw  [fill={rgb, 255:red, 155; green, 155; blue, 155 }  ,fill opacity=0.63 ] (14,176.67) -- (59.9,176.67) -- (59.9,222.57) -- (14,222.57) -- cycle ;
\draw  [fill={rgb, 255:red, 255; green, 255; blue, 255 }  ,fill opacity=1 ] (17.75,180.42) -- (56.15,180.42) -- (56.15,218.82) -- (17.75,218.82) -- cycle ;
\draw  [fill={rgb, 255:red, 155; green, 155; blue, 155 }  ,fill opacity=0.63 ] (102.11,176.67) -- (148.01,176.67) -- (148.01,222.57) -- (102.11,222.57) -- cycle ;
\draw  [fill={rgb, 255:red, 255; green, 255; blue, 255 }  ,fill opacity=1 ] (105.86,180.42) -- (144.26,180.42) -- (144.26,218.82) -- (105.86,218.82) -- cycle ;
\draw  [fill={rgb, 255:red, 155; green, 155; blue, 155 }  ,fill opacity=0.63 ] (549.49,176.77) -- (595.39,176.77) -- (595.39,222.67) -- (549.49,222.67) -- cycle ;
\draw  [fill={rgb, 255:red, 255; green, 255; blue, 255 }  ,fill opacity=1 ] (553.24,180.52) -- (591.64,180.52) -- (591.64,218.92) -- (553.24,218.92) -- cycle ;
\draw  [fill={rgb, 255:red, 155; green, 155; blue, 155 }  ,fill opacity=0.63 ] (637.6,176.77) -- (683.5,176.77) -- (683.5,222.67) -- (637.6,222.67) -- cycle ;
\draw  [fill={rgb, 255:red, 255; green, 255; blue, 255 }  ,fill opacity=1 ] (641.35,180.52) -- (679.75,180.52) -- (679.75,218.92) -- (641.35,218.92) -- cycle ;
\draw  [fill={rgb, 255:red, 155; green, 155; blue, 155 }  ,fill opacity=0.63 ] (373,176.77) -- (418.9,176.77) -- (418.9,222.67) -- (373,222.67) -- cycle ;
\draw  [fill={rgb, 255:red, 255; green, 255; blue, 255 }  ,fill opacity=1 ] (376.75,180.52) -- (415.15,180.52) -- (415.15,218.92) -- (376.75,218.92) -- cycle ;
\draw  [fill={rgb, 255:red, 155; green, 155; blue, 155 }  ,fill opacity=0.63 ] (461.11,176.77) -- (507.01,176.77) -- (507.01,222.67) -- (461.11,222.67) -- cycle ;
\draw  [fill={rgb, 255:red, 255; green, 255; blue, 255 }  ,fill opacity=1 ] (464.86,180.52) -- (503.26,180.52) -- (503.26,218.92) -- (464.86,218.92) -- cycle ;
\draw   (349,162.84) -- (350,162.84) -- (350,241) -- (349,241) -- cycle ;
\draw    (200,105) -- (170.19,145.24) ;
\draw [shift={(169,146.84)}, rotate = 306.53] [color={rgb, 255:red, 0; green, 0; blue, 0 }  ][line width=0.75]    (10.93,-3.29) .. controls (6.95,-1.4) and (3.31,-0.3) .. (0,0) .. controls (3.31,0.3) and (6.95,1.4) .. (10.93,3.29)   ;
\draw    (501,101) -- (530.79,140.25) ;
\draw [shift={(532,141.84)}, rotate = 232.8] [color={rgb, 255:red, 0; green, 0; blue, 0 }  ][line width=0.75]    (10.93,-3.29) .. controls (6.95,-1.4) and (3.31,-0.3) .. (0,0) .. controls (3.31,0.3) and (6.95,1.4) .. (10.93,3.29)   ;

\draw (333.04,40.2) node [anchor=north west][inner sep=0.75pt]    {$M2^{ra\ell }$};
\draw (412.57,94.1) node [anchor=north west][inner sep=0.75pt]    {$2^{r\ell }$};
\draw (473.5,40.2) node [anchor=north west][inner sep=0.75pt]    {$M2^{ra\ell }$};
\draw (553.03,94.1) node [anchor=north west][inner sep=0.75pt]    {$2^{r\ell }$};
\draw (131.21,94.1) node [anchor=north west][inner sep=0.75pt]    {$2^{r\ell }$};
\draw (192.15,40.2) node [anchor=north west][inner sep=0.75pt]    {$M2^{ra\ell }$};
\draw (271.68,94.1) node [anchor=north west][inner sep=0.75pt]    {$2^{r\ell }$};
\draw (94,231.4) node [anchor=north west][inner sep=0.75pt]    {$( M,a) \ -\ \text{partition}$};
\draw (74,258.4) node [anchor=north west][inner sep=0.75pt]    {$\Gamma ( A) \ =\ \{A_{1} ,\ A_{2} ,\ A_{3} ,\ A_{4}\}$};
\draw (439,233.06) node [anchor=north west][inner sep=0.75pt]    {$( M,a,\ r) \ -\ \text{partition}$};
\draw (501,259.06) node [anchor=north west][inner sep=0.75pt]    {$\{A\}$};
\draw (27,152.4) node [anchor=north west][inner sep=0.75pt]    {$A_{1}$};
\draw (115,153.4) node [anchor=north west][inner sep=0.75pt]    {$A_{2}$};
\draw (206,154.4) node [anchor=north west][inner sep=0.75pt]    {$A_{3}$};
\draw (294,154.4) node [anchor=north west][inner sep=0.75pt]    {$A_{4}$};
\draw (345,80.4) node [anchor=north west][inner sep=0.75pt]    {$A$};
\draw (520,159.4) node [anchor=north west][inner sep=0.75pt]    {$A$};

\end{tikzpicture}
    \caption{We wish to partition the gray area $A$. Each cube has side $2^{r\ell}$ and the distance between each other is $M2^{ra\ell}$. With $r=2$, $2^{r}-1 = 3$ and therefore no partition into smaller parts is a $(M,a,r)$-partition. However, the partition into connected components $\{A_1, A_2, A_3, A_4\}$ is an $(M,a)$-partition, since $V(A_i)^{\frac{a}{\delta}} = 2^{ra\frac{d}{\delta}\ell}<2^{ra\ell}$ whenever $\delta>d$.}
    \label{Fig: Figura10}
\end{figure}

The existence of a $(M,a)$-partition for any $A\Subset\Z^d$ does not depend on the choice of $M,a>0$. However, to guarantee the existence of phase transition, we have to choose particular values for these parameters, see Remark \ref{Rmk: choice_of_a_and_delta}. Later on, in Proposition \ref{Prop: Cost_erasing_contour}, $M$ will be taken large enough. 

We write $\Gamma(\sigma) \coloneqq \Gamma(\partial\sigma)$ for a $(M,a)$-partition of $\partial\sigma$. In general, there is more than one $(M,a)$-partition for each region $A\in\Z^d$. Given two partitions $\Gamma$ and $\Gamma^\prime$ of a set $A$, we say that $\Gamma$ \emph{is finer than} $\Gamma'$, and denote $\Gamma\preceq\Gamma^{\prime}$, if for every $\overline{\gamma} \in \Gamma$ there is $\overline{\gamma}' \in \Gamma'$ with $\overline{\gamma} \subseteq \overline{\gamma}'$. Given any two $(M,a)$-partitions $\Gamma(A)$ and $\Gamma^\prime(A)$, 
    $\Gamma\cap\Gamma^\prime\coloneqq \{\overline{\gamma}\cap\overline{\gamma}^\prime : \overline{\gamma} \in \Gamma(A), \ \overline{\gamma}\in\Gamma^\prime(A), \ \overline{\gamma}\cap\overline{\gamma}^\prime\neq \emptyset\}$
 is a $(M,a)$-partition finer than both $\Gamma(A)$ and $\Gamma^\prime(A)$. As the number of partitions is finite, we can intersect all partitions to get a finest $(M,a)$-partition. 

 From now on, when taking a $(M,a)$-partition $\Gamma(A)$, we will always assume it is the finest. It is easy to see that the finest $(M,a)$-partition $\Gamma(A)$ satisfies the following property:

\begin{itemize}
    \item[\textbf{(A1)}] For any $\overline{\gamma},\overline{\gamma}^\prime\in \Gamma(A)$, $\overline{\gamma}'$ is contained in only one connected component of $(\overline{\gamma})^c$.
\end{itemize}

Property \textbf{(A1)} is essential to define labels as in \cite{Affonso.2021}. See Figure \ref{Fig. Exemple_A1} for an example of partition not satisfying \textbf{(A1)}.

 
\tikzset{
pattern size/.store in=\mcSize, 
pattern size = 5pt,
pattern thickness/.store in=\mcThickness, 
pattern thickness = 0.3pt,
pattern radius/.store in=\mcRadius, 
pattern radius = 1pt}
\makeatletter
\pgfutil@ifundefined{pgf@pattern@name@_2ysuecq82}{
\makeatletter
\pgfdeclarepatternformonly[\mcRadius,\mcThickness,\mcSize]{_2ysuecq82}
{\pgfpoint{-0.5*\mcSize}{-0.5*\mcSize}}
{\pgfpoint{0.5*\mcSize}{0.5*\mcSize}}
{\pgfpoint{\mcSize}{\mcSize}}
{
\pgfsetcolor{\tikz@pattern@color}
\pgfsetlinewidth{\mcThickness}
\pgfpathcircle\pgfpointorigin{\mcRadius}
\pgfusepath{stroke}
}}
\makeatother

 
\tikzset{
pattern size/.store in=\mcSize, 
pattern size = 5pt,
pattern thickness/.store in=\mcThickness, 
pattern thickness = 0.3pt,
pattern radius/.store in=\mcRadius, 
pattern radius = 1pt}
\makeatletter
\pgfutil@ifundefined{pgf@pattern@name@_p46h6u1xc}{
\makeatletter
\pgfdeclarepatternformonly[\mcRadius,\mcThickness,\mcSize]{_p46h6u1xc}
{\pgfpoint{-0.5*\mcSize}{-0.5*\mcSize}}
{\pgfpoint{0.5*\mcSize}{0.5*\mcSize}}
{\pgfpoint{\mcSize}{\mcSize}}
{
\pgfsetcolor{\tikz@pattern@color}
\pgfsetlinewidth{\mcThickness}
\pgfpathcircle\pgfpointorigin{\mcRadius}
\pgfusepath{stroke}
}}
\makeatother

\begin{figure}[H]
	\centering
	
    \tikzset{every picture/.style={line width=0.75pt}} 

\begin{tikzpicture}[x=0.75pt,y=0.75pt,yscale=-0.75,xscale=0.75]

\draw  [fill={rgb, 255:red, 234; green, 234; blue, 234 }  ,fill opacity=0.94 ] (186,64) .. controls (206,54) and (372,68.5) .. (455,27.5) .. controls (538,-13.5) and (567,184) .. (559,204.5) .. controls (551,225) and (566.42,303.92) .. (484,314.5) .. controls (401.58,325.08) and (189.03,307.2) .. (169,302.5) .. controls (148.97,297.8) and (57,206.5) .. (55,176.5) .. controls (53,146.5) and (166,74) .. (186,64) -- cycle ;
\draw  [fill={rgb, 255:red, 255; green, 255; blue, 255 }  ,fill opacity=1 ] (356,212.5) .. controls (361.71,209.65) and (360.41,200.35) .. (356.46,188.03) .. controls (346.6,157.21) and (320.24,107.48) .. (346,92.5) .. controls (382.06,71.54) and (480,53.5) .. (494,83.5) .. controls (508,113.5) and (481,182.5) .. (501,212.5) .. controls (521,242.5) and (235,300.5) .. (215,270.5) .. controls (195,240.5) and (336,222.5) .. (356,212.5) -- cycle ;
\draw  [fill={rgb, 255:red, 255; green, 255; blue, 255 }  ,fill opacity=1 ] (154,105) .. controls (174,95) and (220,69) .. (254,82.5) .. controls (288,96) and (301.65,130.51) .. (301,138.5) .. controls (300.35,146.49) and (301,140) .. (271,185.5) .. controls (241,231) and (200.19,240.89) .. (169,251.5) .. controls (137.81,262.11) and (118.4,216.1) .. (108,200.5) .. controls (97.6,184.9) and (134,115) .. (154,105) -- cycle ;
\draw  [pattern=_2ysuecq82,pattern size=5.25pt,pattern thickness=0.75pt,pattern radius=0.75pt, pattern color={rgb, 255:red, 133; green, 128; blue, 128}] (173,130.5) .. controls (193,120.5) and (207,101) .. (230,130.5) .. controls (253,160) and (239,163.5) .. (193,193.5) .. controls (147,223.5) and (167,198.5) .. (147,168.5) .. controls (127,138.5) and (153,140.5) .. (173,130.5) -- cycle ;
\draw  [pattern=_p46h6u1xc,pattern size=5.25pt,pattern thickness=0.75pt,pattern radius=0.75pt, pattern color={rgb, 255:red, 133; green, 128; blue, 128}] (436,120.5) .. controls (456,80.5) and (465,81.5) .. (462,109.5) .. controls (459,137.5) and (470.24,185.85) .. (442,203.5) .. controls (413.76,221.15) and (412.21,218.61) .. (407,204.5) .. controls (401.79,190.39) and (400,224.5) .. (381,221.5) .. controls (362,218.5) and (386,181.5) .. (392,143.5) .. controls (398,105.5) and (416,160.5) .. (436,120.5) -- cycle ;

\draw (442,261.4) node [anchor=north west][inner sep=0.75pt]  [font=\LARGE]  {$\overline{\gamma} $};
\draw (189,138.4) node [anchor=north west][inner sep=0.75pt]  [font=\LARGE]  {$\overline{\gamma} '$};
\draw (417,152.4) node [anchor=north west][inner sep=0.75pt]  [font=\LARGE]  {$\overline{\gamma} '$};

\end{tikzpicture}
\caption{An example of how Condition \textbf{(A1)} works: considering $\overline{\gamma}'$ the dotted region and $\overline{\gamma}$ the grey region, one can readily see that $\overline{\gamma}'$ intersects two different connected components of $(\overline{\gamma})^c$. To turn this into a partition satisfying condition \textbf{(A1)}, one should separate $\overline{\gamma}'$ in two different sets of $\Gamma(A)$.}
\label{Fig. Exemple_A1}
\end{figure}

Counting the number of contours surrounding zero using the finest $(M,a)$-partition may be troublesome since the definition provides very little information on these objects. To extract good properties of these contours, we establish a multiscale procedure, depending on a parameter $r$, that creates a $(M,a)$-partition of any given set. To define this procedure, we introduce some notation. 

 For any $x\in\Z^d$ and $m\geq 0$,
\begin{equation}
    C_{m}(x) \coloneqq \left(\prod_{i=1}^d{\left[2^{m}x_i , \ 2^{m}(x_i+1) \right)}\right)\cap \Z^d,
\end{equation}
is the cube of $\mathbb{Z}^d$ centered at $2^{m}x + 2^{m-1} - \frac{1}{2}$ with side length $2^{m} -1$. Any such cube is called an $m$-cube. As all cubes in this paper are of this form, with centers $2^{m}x + 2^{m-1} - \frac{1}{2}$ and $x \in \mathbb{Z}^d$, we will often omit the point $x$ in what follows, writing $C_m$ for an $m$-cube instead of $C_m(x)$. An arbitrary collection of $m$-cubes will be denoted $\mathscr{C}_m$ and $B_{\mathscr{C}_m}\coloneqq \cup_{C\in\mathscr{C}_m}C$ is the region covered by $\mathscr{C}_m$.
We denote by $\mathscr{C}_m(\Lambda)$ the covering of $\Lambda\Subset\Z^d$ with the smallest possible number of $m$-cubes.

\input{Figures/Figura.0}

For each $n\geq 0$, define the graph $G_n(\Lambda) = (V_n(\Lambda), E_n(\Lambda))$ with vertex set $V_n(\Lambda) = \mathscr{C}_n(\Lambda)$ and $E_n(\Lambda) = \{ (C_n, C_n^\prime) : d(C_n,C_n^\prime) \leq M2^{an}\}$. Let $\mathscr{G}_n(\Lambda)$ be the connected components of $G_n(\Lambda)$. Given ${G = (V,E) \in \mathscr{G}_n(\Lambda)}$, we denote $\Lambda^G \coloneqq \Lambda \cap B_V$ the area of $\Lambda$ covered by $G$. In the next proposition, we introduce a procedure to construct possibly non-trivial $(M,a)$-partitions.
    
\begin{proposition}\label{Prop:Construction_(M,a,delta)_partition}
    For any $r> 0$ and $A\Subset\Z^d$, there is a possibly non-trivial $(M,a)$-partition $\Gamma^r(A)$.
\end{proposition}

\begin{proof}
  Given $r> 0$ and $A\Subset\Z^d$, $\Gamma^r(A)$ is the partition of $A$ created by the following procedure. In the first step we consider $A_1 \coloneqq A$ and we take the connected components of $G\in \mathscr{G}_r(A_1)$ such that $A_1^G$ have small density, that is, consider
\begin{equation*}
   \mathscr{P}_1 \coloneqq \{G \in  \mathscr{G}_r(A_1) : |V(A_1^G)|\leq {2^{r(d+1)}}\}.
\end{equation*}
    Then, the subsets to be removed in the first step are $\Gamma_1^r(A) \coloneqq \{A_1^G : G\in \mathscr{P}_1\}$ and the set left to partition is $A_2 \coloneqq A_1\setminus \bigcup\limits_{\gamma \in \Gamma_1^r(A)}\gamma$.
 We can repeat this procedure inductively by taking 
\begin{equation*}
   \mathscr{P}_n \coloneqq \{G \in  \mathscr{G}_{rn}(A_{n}) : |V(A_n^G)|\leq {2^{rn{(d+1)}}}\},
\end{equation*}
then define $\Gamma_n^r(A) \coloneqq \{A_n^G : G \in \mathscr{P}_n\}$ and $A_{n+1} \coloneqq A_n \setminus \bigcup\limits_{ \gamma\in \Gamma_n^r(A)}\gamma$. As the cubes invade the lattice, this procedure stops, in the sense that for some $N$ large enough, $\mathscr{P}_n=\emptyset$ for all $n\geq N$. We then define $\Gamma^r(A) \coloneqq \cup_{n\geq 0} \Gamma_n^r(A)$.
 By this construction, $\Gamma^r(A)$ is clearly a partition of $A$, so condition \textbf{(A)} follows. To show condition \textbf{(B)}, take $\overline{\gamma},\overline{\gamma}^\prime\in \Gamma^r(A)$. Let $m\geq n\geq 1$ be such that $\overline{\gamma}\in\Gamma_n^r(A)$ and $\overline{\gamma}^\prime\in\Gamma_m^r(A)$. Then, 
\begin{equation*}
\d(\overline{\gamma},\overline{\gamma}^\prime)\geq M2^{rna}\geq M\left(2^{rn(d+1)}\right)^{\frac{a}{d+1}} \geq M|V(\overline{\gamma})|^{\frac{a}{d+1}}.
\end{equation*}
If $m=n$, the same inequality holds for $|V(\overline{\gamma}^\prime)|$ and condition \textbf{(B)} holds. When $m>n$, $\overline{\gamma}^\prime$ was not removed at step $n$, so $|V(\overline{\gamma}^\prime)|> 2^{rn(d+1)} \geq |V(\overline{\gamma})|$, so $|V(\overline{\gamma})| = \min\{|V(\overline{\gamma})|,|V(\overline{\gamma}^\prime)|\}$ and again we get condition \textbf{(B)}. 
\end{proof}

The construction in Proposition \ref{Prop:Construction_(M,a,delta)_partition} works for any $r>0$, but we need to take $r$ large enough for the computations in Section 3 to work. So we fix $r\coloneqq 4\lceil\log_2(a+1) \rceil + d +1$, where $\lceil x \rceil$ is the smallest integer greater than or equal to $x$. This $r$ is taken larger than the one in \cite{Affonso.2021} to simplify some calculations. All our computations should work with the previous choice of $r$, with some adaptation. Next, we define the label of a contour. 
 
\begin{definition}
    For $\Lambda\subset\Z^d$, the \textit{edge boundary} of $\Lambda$ is $\partial\Lambda \coloneqq \{\{x,y\} \subset \Z^d: |x-y|=1, x \in \Lambda, y \in \Lambda^c\}$. The \textit{inner boundary} of $\Lambda$ is $\fint\Lambda\coloneqq\{x\in\Lambda : \exists y\in \Lambda^c \text{ such that }|x-y|=1\}$ and the \textit{external boundary} is $\fext\Lambda\coloneqq\{x\in\Lambda^c : \exists y\in \Lambda \text{ such that }|x-y|=1\}$.
\end{definition}

\begin{remark}
    The usual isoperimetric inequality states that $2d|\Lambda|^{\frac{d-1}{d}}\leq |\partial \Lambda|$. The inner boundary and the edge are related by $|\fint \Lambda|\leq |\partial \Lambda|\leq 2d|\fint \Lambda|$, so we can write the inequality as $|\Lambda|^{\frac{d-1}{d}}\leq |\fint \Lambda|$.
\end{remark}

To define the label of a contour, the naive definition would be to take the sign of the inner boundary of the set $\overline{\gamma}$. However, this cannot be done since this inner boundary may have different signs, see Figure \ref{fig: Figura2}.

\input{Figures/Figura.2}

For any $\Lambda\Subset\Z^d$, its connected components are denoted $\Lambda^{(1)}, \dots, \Lambda^{(n)}$. Given $\overline{\gamma} \in\Gamma(\sigma)$, a connected component $\overline{\gamma}^{(k)}$ is \textit{external} if $V(\overline{\gamma}^{(j)})\subset V(\overline{\gamma}^{(k)})$, for all other connected components $\overline{\gamma}^{(j)}$ satisfying $V(\overline{\gamma}^{(j)})\cap V(\overline{\gamma}^{(k)}) \neq \emptyset$. Denoting 
\begin{equation*}
    \overline{\gamma}_\mathrm{ext} = \hspace{-0.5cm}\bigcup_{\substack{k\geq 1 \\ \overline{\gamma}^{(k)} \text{ is external}}}\hspace{-0.5cm}\overline{\gamma}^{(k)},
\end{equation*} 
it is shown in \cite[Lemma 3.8]{Affonso.2021} that the sign of $\sigma$ is constant in $\fint V(\overline{\gamma}_{\mathrm{ext}})$. The \textit{label} of $\overline{\gamma}$ is the function $\lab_{\overline{\gamma}} :\{(\overline{\gamma})^{(0)}, \I(\overline{\gamma})^{(1)}\dots, \I(\overline{\gamma})^{(n)}\} \rightarrow \{-1,+1\}$ defined as: $\lab_{\overline{\gamma}}(\I(\overline{\gamma})^{(k)})$ is the sign of the configuration $\sigma$ in $\fint V(\I(\overline{\gamma})^{(k)})$, for $k\geq 1$, and $\lab_{\overline{\gamma}}((\overline{\gamma})^{(0)})$ is the sign of $\sigma$ in $\fint V(\overline{\gamma}_\mathrm{ext})$.  We then define the contours.

\begin{definition}
Given a configuration $\sigma$ with finite boundary, its \emph{contours} $\gamma$ are pairs $(\overline{\gamma},\lab_{\overline{\gamma}})$,  where $\overline{\gamma} \in \Gamma(\sigma)$ and $\lab_{\overline{\gamma}}$ is the label of $\overline{\gamma}$ as defined above. The \emph{support of the contour}  $\gamma$ is defined as $\Sp(\gamma)\coloneqq \overline{\gamma}$ and its \emph{size} is given by $|\gamma| \coloneqq |\Sp(\gamma)|$.
\end{definition}

Each contour $\gamma$ has an \textit{interior}, given by $\I(\gamma) \coloneqq \I(\Sp(\gamma))$, and a \textit{volume}, given by $V(\gamma)\coloneqq V(\Sp(\gamma))$. 
We also split the interior according to its labels as

\begin{equation*}
    \I_\pm(\gamma) = \hspace{-0.7cm}\bigcup_{\substack{k \geq 1, \\ \lab_{\overline{\gamma}}(\I(\gamma)^{(k)})=\pm 1}}\hspace{-0.7cm}\I(\gamma)^{(k)}.
\end{equation*}
Different from Pirogov-Sinai theory, where the interiors of contours are a union of simply connected sets, the interior $\I(\gamma)$ is at most the union of connected sets, that is, they may have holes. 

An arbitrary collection of contours $\Gamma = \{\gamma_1,\dots,\gamma_n\}$ may not form a $(M,a)$-partition. Even so, their labels may not be compatible. When there exists a configuration $\sigma$ with contours precisely $\Gamma$, we say that $\Gamma$ is \textit{compatible}. Notice that there is no bijection between compatible collections of contours and configurations since more than one configuration can have the same boundary and label. 

\input{Figures/Figura.3}

A contour $\gamma\in\Gamma$ is \textit{external} if its external connected components are not contained in any other $V(\gamma')$, for $\gamma' \in \Gamma\setminus\{\gamma\}$. Taking $\I_\pm(\Gamma) \coloneqq \cup_{\gamma\in\Gamma}\I_\pm(\gamma)$ and $V(\Gamma)\coloneqq\cup_{\gamma\in\Gamma}V(\gamma)$, for each $\Lambda\Subset\Z^d$ we consider the sets\\
\begin{equation*}
\mathcal{E}^\pm_\Lambda \coloneqq\{\Gamma= \{\gamma_1, \ldots, \gamma_n\}: \Gamma \text{ is compatible,} \gamma_i \text{ is external}, \lab_{\gamma_i}((\gamma_i)^{(0)})=\pm1, V(\Gamma) \subset \Lambda\}, \vspace{0.2cm}
\end{equation*}
of all external compatible families of contours with external label $\pm$ contained in $\Lambda$.  When we write $\gamma \in \mathcal{E}^\pm_\Lambda$ we mean $\{\gamma\} \in \mathcal{E}^\pm_\Lambda$. Most of the time the set $\Lambda$ will play no role, so we will often omit the subscript. 

The first step for a Peierls-type argument to hold is to control the number of contours with a fixed size. Consider $\mathcal{C}_0(n) \coloneqq \{\gamma \in \mathcal{E}^+_\Lambda: 0 \in V(\gamma), |\gamma|=n\}$, the set of contours with fixed size with the origin in its volume, and  $\mathcal{C}_0 \coloneqq \cup_{n\geq 1}\mathcal{C}_0(n)$. 
We will later show in Corollary \ref{Cor: Bound_on_C_0_n} that the size of the set $\mathcal{C}_0(n)$ is exponentially bounded depending on $n$. 

The second key step of a Peierls-type argument is to control the energy cost of erasing a contour. Given $\Gamma\in \mathcal{E}^+$, the configurations compatible with $\Gamma$ are $\Omega(\Gamma)\coloneqq \{\sigma\in\Omega^+_\Lambda : \Gamma\subset \Gamma(\sigma)\}$. The map $\tau_\Gamma:\Omega(\Gamma) \rightarrow \Omega_{\Lambda}^+$ defined as 
\be
\tau_\Gamma(\sigma)_x = 
\begin{cases}
	\;\;\;\sigma_x &\text{ if } x \in \I_+(\Gamma)\cup V(\Gamma)^c, \\
	-\sigma_x &\text{ if } x \in \I_-(\Gamma),\\
	+1 &\text{ if } x \in \Sp(\Gamma),
\end{cases}
\ee
erases a family of compatible contours, since the spin-flip preserves incorrect points but transforms $-$-correct points into $+$-correct points. Define, for $B\Subset\Z^d$, the interaction
\begin{equation*}
    F_B\coloneqq \sum_{\substack{x\in B \\ y\in  B^c}}J_{xy}.
\end{equation*}

Given $B\subset\Z^d$ and $\sigma\in\Omega$ with $\partial\sigma$ finite, let $\Gamma_{\Int}(\sigma, B)$ be the contours $\gamma^\prime$ with $\Sp(\gamma^\prime) \in \Gamma(\sigma)$ and enclosed by $B$, that is, $\Sp(\gamma^\prime)\subset B$. Define also $\Gamma_{\Ext}(\sigma, B)$ as the contours $\gamma^\prime$ with $\Sp(\gamma^\prime) \in \Gamma(\sigma)$ outside $B$, that is, $\Sp(\gamma^\prime)\subset B^c$.

\begin{lemma}\label{Lemma: Aux_1}
Given $\sigma\in\Omega$ with $\partial \sigma$ finite and a contour $\gamma$ with $\Sp(\gamma)\in \Gamma(\sigma)$ , there is a constant $\kappa^{(1)}_\alpha\coloneqq \kappa^{(1)}_\alpha(\alpha, d)$, such that, for  $B = \Sp(\gamma)$ or $B=\I_-(\gamma)$ we have 

\begin{equation}\label{Eq: Lemma_aux_2}
    \sum_{\substack{x\in B \\ y\in V(\Gamma_{\Ext}(\sigma, B)\setminus\{\gamma\})}} J_{xy} \leq \kappa^{(1)}_\alpha \left[ \frac{|B|}{M^{\alpha - d}}|V(\gamma)|^{\frac{a}{d+1}(d-\alpha)} + \frac{F_B}{M} \right].
\end{equation}

\end{lemma}

\begin{proof}

Fixed $\sigma$ and $B$, we drop them from the notation, so $\Gamma_{\Ext} \coloneqq \Gamma_{\Ext}(\sigma, B)$.  Splitting $\Gamma_{\Ext}\setminus\{\gamma\}$ into $\Upsilon_1 \coloneqq \{\gamma^\prime \in \Gamma_{\Ext}\setminus\{\gamma\} : |V(\gamma^\prime)| \geq |V(\gamma)|\}$ and $\Upsilon_2 = \Gamma_{\Ext} \setminus (\Upsilon_1\cup \gamma)$ we get 
\begin{equation*}
     \sum_{\substack{x\in B \\ y\in V(\Gamma_{\Ext}\setminus\{\gamma\})}} J_{xy} \leq  \sum_{\substack{x\in B \\ y\in V(\Upsilon_1)}} J_{xy} +  \sum_{\substack{x\in B \\ y\in V(\Upsilon_2)}} J_{xy}.
\end{equation*}
For any  $\gamma^\prime \in \Upsilon_1$, $\d(\gamma,\gamma^\prime)> M |V(\gamma)|^{\frac{a}{d+1}}$, hence 
\begin{equation}\label{Eq: Lemma_aux_2_1}
    \sum_{\substack{x\in B \\ y\in V(\Upsilon_1)}} J_{xy} \leq \sum_{\substack{x\in B \\ y : |y-x| > R}} J_{xy} = |B|\sum_{y : |y|>R}J_{0y},
\end{equation}
with $R\coloneqq M |V(\gamma)|^{\frac{a}{d+1}}$. Defining $s_d(n) \coloneqq |\{x\in \Z^d : |x|=n \}|$, it is known that $s_d(n)\leq 2^{2d - 1}e^{d-1}n^{d-1}$, see for example \cite[Lemma 4.2]{Affonso.2021}. Using the integral bound of the sum, we can show that
\begin{equation}\label{Eq: Interaction_outside_ball}
 \begin{split}
 \sum_{y : |y|>R}J_{0y} &= J \sum_{n>R} \frac{s_d(n)}{n^\alpha} \leq \frac{J2^{d - 1 + \alpha}e^{d-1}}{(\alpha - d)} {R}^{d-\alpha}.
 \end{split}
\end{equation}
Together with \eqref{Eq: Lemma_aux_2_1}, this yields
    \begin{align}\label{Eq: Lemma_aux_2.Part1}
    \sum_{\substack{x\in B \\ y\in V(\Upsilon_1)}} J_{xy} \leq \frac{J2^{d-1 + \alpha }e^{d-1}}{(\alpha - d)}\frac{|B|}{M^{\alpha - d}}|V(\gamma)|^{\frac{a}{d+1}(d-\alpha)}.
\end{align}
To bound the other term, split $\Upsilon_2$ into layers $\Upsilon_{2,m} \coloneqq \{ \gamma^\prime \in \Upsilon_2 : |V(\gamma^\prime)|=m\}$, for $1\leq m\leq |V(\gamma)|-1$. Denoting  $y_{\gamma^\prime,x} \in \Sp(\gamma^\prime)$ the point satisfying $\d(x, \Sp(\gamma^\prime)) = \d(x, y_{\gamma^\prime, x})$, we can bound 
\begin{align*}
   \sum_{\substack{x\in B \\ y\in V(\Upsilon_{2,m})}} J_{xy} \leq m \sum_{\substack{x\in B \\ \gamma^\prime \in \Upsilon_{2,m}}} J_{x,y_{\gamma^\prime, x}}.
\end{align*}
Since $\left|x - y_{\gamma^\prime, x}\right| >Mm^{\frac{a}{d+1}}$, and $\left|y_{\gamma^\prime, x} -  y_{\gamma^{\prime\prime}, x}\right| \geq \d(\gamma^\prime, \gamma^{\prime\prime})>Mm^{\frac{a}{d+1}}$ for any $\gamma^\prime, \gamma^{\prime\prime}\in\Upsilon_{2,m}$, the balls with radius $\frac{M}{3}m^{\frac{a}{d+1}}$ centered in $y_{\gamma^\prime, x}$, for all $\gamma^\prime\in \Upsilon_{2,m}$ are disjoint and are contained in $B^c$. Hence, we can bound

\begin{equation*}
    \sum_{\substack{x\in B \\ \gamma^\prime \in\Upsilon_{2,m}}} J_{x,y_{\gamma^\prime, x}} \leq \frac{3}{Mm^{\frac{a}{d+1}}}F_B.
\end{equation*}
That gives us
\begin{align}\label{Eq: Lemma_aux_2.Part2}
     \sum_{\substack{x\in B \\ y\in V(\Upsilon_{2})}} J_{xy} \leq \sum_{m=1}^{|V(\gamma)|-1} \frac{3}{Mm^{\frac{a}{d+1}-1}}F_B\leq \frac{3\zeta({\frac{a}{d+1}-1})}{M}F_B,
\end{align}
what concludes the proof for $\kappa_\alpha^{(1)}\coloneqq \frac{J2^{d-1 + \alpha}e^{d-1}}{(\alpha - d)} + {3\zeta({\frac{a}{d+1}-1})}$.

\end{proof}

\begin{corollary}\label{Cor: Corrolary_of_Lemma_aux}
For any configuration $\sigma\in\Omega$ and $\gamma\in\Gamma(\sigma)$, 

\begin{equation*}
  \sum_{\substack{x\in \Sp(\gamma) \\ y\in V(\Gamma(\sigma)\setminus\{\gamma\})}} J_{xy} \leq\frac{\kappa^{(2)}_\alpha}{M^{(\alpha - d)\wedge 1}}F_{\Sp(\gamma)},  \\
\end{equation*}
 
\begin{equation*}
    \sum_{\substack{x\in \I_-(\gamma) \\ y\in V(\Gamma_{\Ext}(\sigma, \I_-(\gamma))\setminus\{\gamma\})}} J_{xy}  \leq \frac{\kappa^{(2)}_\alpha}{M^{(\alpha - d)\wedge 1}}F_{\I_-(\gamma)},
\end{equation*}
and 
\begin{equation*}
    \sum_{\substack{x\in \I_-(\gamma)^c \\ y\in V(\Gamma_{\Int}(\sigma, \I_-(\gamma)))}} J_{xy} \leq  \kappa^{(2)}_\alpha\frac{F_{\I_-(\gamma)}}{M},
\end{equation*}
with $\kappa^{(2)}_\alpha \coloneqq \kappa_\alpha^{(1)}[J^{-1} + 1]$.
\end{corollary}

\begin{proof}
    The first inequality is a direct application of the Lemma \ref{Lemma: Aux_1} for $B=\Sp(\gamma)$, once we note that $\Gamma_{\Ext}(\sigma, B) = \Gamma(\sigma)\setminus\{\gamma\}$ and, our choice of $a$,  ${|\gamma|}{|V(\gamma)|^{\frac{a}{d+1}(d-\alpha)} } \leq 1$. Finally, we put $1 + F_{\Sp(\gamma)} \leq (J^{-1} + 1) F_{\Sp(\gamma)}$. The second inequality is likewise a direct application of Lemma \ref{Lemma: Aux_1} for $B=\I_-(\gamma)$, since  ${|\I_-(\gamma)||V(\gamma)|^{- \frac{a}{d+1}(\alpha-d)} }\leq 1$ and, similarly, $1 + F_{\I_-(\gamma)}\leq(J^{-1} + 1)F_{\I_-(\gamma)}$.
For the last inequality, we cannot apply Lemma \ref{Lemma: Aux_1} directly. However, the proof works in the similar steps when we take $B = \I_-(\gamma)^c$. Moreover, notice that $V(\Gamma_{\Int}(\sigma, \I_-(\gamma))) = V(\Gamma_{\Ext}(\sigma, \I_-(\gamma)^c))$ and, for all $\gamma^\prime\in \Gamma_{\Int}(\sigma, \I_-(\gamma))$, $|V(\gamma^\prime)| < |V(\gamma)|$.  In the notation of the proof of Lemma \ref{Lemma: Aux_1}, this means that $\Upsilon_{2} = \Gamma_{\Ext}(\sigma, \I_-(\gamma)^c)$, so equation \eqref{Eq: Lemma_aux_2.Part2} yields
\begin{align*}
      \sum_{\substack{x\in \I_-(\gamma)^c \\ y\in V(\Gamma_{\Ext}(\sigma, \I_-(\gamma)^c))}} J_{xy} \leq \zeta\Big({\frac{a}{d+1}-1}\Big)\frac{F_{\I_-(\gamma)^c}}{M}.
\end{align*}
Since $F_{\I_-(\gamma)^c} = F_{\I_-(\gamma)}$ and $\zeta({\frac{a}{d+1}-1})\leq \kappa^{(1)}_\alpha \leq  \kappa^{(2)}_\alpha$, we get the desired bound.
\end{proof}

We are ready to prove the main proposition of this section:

\begin{proposition}\label{Prop: Cost_erasing_contour}
	For $M$ large enough, there exists a constant $c_2\coloneqq c_2(\alpha,d)>0$, such that for  any $\Lambda \Subset \Z^d$, $\gamma\in \mathcal{E}^+_\Lambda$, and $\sigma \in \Omega(\gamma)$ it holds that
	\be
	H_{\Lambda; 0}^+(\sigma)- H_{\Lambda;0}^+(\tau_{\gamma}(\sigma))\geq c_2\left(|\gamma| + F_{\I_-(\gamma)} + F_{\Sp(\gamma)} \right).
	\ee
\end{proposition}

\begin{proof}
   To simplify the notation, we denote $\tau\coloneqq \tau_\gamma(\sigma)$. Taking $B(\gamma) \coloneqq \I_+(\gamma)\cup V(\gamma)^c$ and redoing the steps of the proof of \cite[Proposition 4.5]{Affonso.2021}, we can write 
\begin{equation}\label{Eq: Difference_of_Hamiltonians_1}
   \begin{split} 
    H_\Lambda^+(\sigma) - H_\Lambda^+(\tau) &= \sum_{\substack{x \in \Sp(\gamma) \\ y \in \Z^d}} J_{xy}\mathbbm{1}_{ \{ \sigma_x \neq \sigma_y \}} +       \sum_{\substack{x \in \Sp(\gamma) \\ y \in \Sp(\gamma)^c}} J_{xy}\mathbbm{1}_{ \{ \sigma_x \neq \sigma_y \}} - 2\sum_{\substack{x \in \I_-(\gamma) \\ y \in B(\gamma)}} J_{xy}\sigma_x\sigma_y \\
    & - 2\sum_{\substack{x \in \Sp(\gamma) \\ y \in B(\gamma)}} J_{xy}\mathbbm{1}_{ \{ \sigma_y = -1 \}} - 2\sum_{\substack{x \in \Sp(\gamma) \\ y \in \I_-(\gamma)}} J_{xy}\mathbbm{1}_{ \{ \sigma_y = +1\}}.
\end{split}
\end{equation}

We start by analyzing the last two negative terms. It holds that 
\begin{equation*}
    \sum_{\substack{x \in \Sp(\gamma) \\ y \in B(\gamma)}} J_{xy}\mathbbm{1}_{ \{ \sigma_y = -1 \}} + \sum_{\substack{x \in \Sp(\gamma) \\ y \in \I_-(\gamma)}} J_{xy}\mathbbm{1}_{ \{ \sigma_y = +1\}} \leq \sum_{\substack{x \in \Sp(\gamma) \\ y \in V(\Gamma(\sigma)\setminus\{\gamma\})}} J_{xy},
\end{equation*}
since the characteristic function can only be non-zero in the volume of other contours. By Corollary \ref{Cor: Corrolary_of_Lemma_aux},
\begin{equation}\label{Eq: Aux_1_Diferenca_de_Hamiltonianos}
    \sum_{\substack{x \in \Sp(\gamma) \\ y \in V(\Gamma(\sigma)\setminus\{\gamma\})}} J_{xy}  \leq {\kappa^{(2)}_\alpha}\frac{F_{\Sp(\gamma)}}{M^{(\alpha - d)\wedge 1}}.
\end{equation}

For the remaining negative term, taking $\Gamma$ the contours associated to $\sigma$,  $\Gamma^\prime \coloneqq \Gamma_{\Int}(\sigma, \I_-(\gamma))$ the contours inside $\I_-(\gamma)$, and $\Gamma^{\prime\prime} \coloneqq \Gamma_{\Ext}(\sigma, \I_-(\gamma))\setminus \{\gamma\}$, as $\Gamma^\prime\cup\Gamma^{\prime\prime} =  \Gamma\setminus {\gamma}$ we can write
\begin{equation}\label{Eq: Interaction_between_interior_and_B}
\begin{split}
     \sum_{\substack{x \in \I_-(\gamma) \\ y \in B(\gamma)}} J_{xy}\sigma_x\sigma_y = \sum_{\substack{x \in V(\Gamma^{\prime}) \\ y \in  V(\Gamma^{\prime\prime}) }} J_{xy} + \sum_{\substack{x \in \I_-(\gamma)\setminus V(\Gamma^{\prime}) \\ y \in  V(\Gamma^{\prime\prime}) }} 2J_{xy}\mathbbm{1}_{\{\sigma_y=-1\}} + \sum_{\substack{x \in V(\Gamma^{\prime}) \\ y \in  B(\gamma)\setminus V(\Gamma^{\prime\prime}) }} 2J_{xy}\mathbbm{1}_{\{\sigma_x=+1\}} \\
       - \sum_{\substack{x \in \I_-(\gamma)\setminus V(\Gamma^{\prime}) \\ y \in  V(\Gamma^{\prime\prime}) }}  J_{xy}  - \sum_{\substack{x \in V(\Gamma^{\prime}) \\ y \in  V(\Gamma^{\prime\prime}) }} 2J_{xy}\mathbbm{1}_{\{\sigma_x\neq \sigma_y\}} -  \sum_{\substack{x \in \I_-(\gamma) \\ y \in  B(\gamma)\setminus V(\Gamma^{\prime\prime}) }}J_{xy}.
\end{split}
\end{equation}

We can bound the first two terms by
\begin{equation}\label{Eq: Aux_1_Interaction_between_interior_and_B}
 \sum_{\substack{x \in V(\Gamma^{\prime}) \\ y \in  V(\Gamma^{\prime\prime}) }} J_{xy} +  \sum_{\substack{x \in \I_-(\gamma)\setminus V(\Gamma^{\prime}) \\ y \in  V(\Gamma^{\prime\prime}) }} 2J_{xy}\mathbbm{1}_{\{\sigma_y=-1\}}  \leq  2\sum_{\substack{x \in \I_-(\gamma) \\ y \in V(\Gamma^{\prime\prime}) }} J_{xy} \leq 2\kappa^{(2)}_\alpha\frac{F_{\I_-(\gamma)}}{M^{(\alpha - d)\wedge 1}}.
 \end{equation}
In the second inequality, we are applying Corollary \ref{Cor: Corrolary_of_Lemma_aux}. For the next term, since $B(\gamma)\setminus V(\Gamma^{\prime\prime})\subset \I_-(\gamma)^c$, we can bound
\begin{equation}\label{Eq: Aux_2_Interaction_between_interior_and_B}
    \sum_{\substack{x \in V(\Gamma^{\prime}) \\ y \in  B(\gamma)\setminus V(\Gamma^{\prime\prime}) }} 2J_{xy}\mathbbm{1}_{\{\sigma_y=+1\}}\leq \sum_{\substack{x \in V(\Gamma^{\prime}) \\ y \in  \I_-(\gamma)^c}} 2J_{xy}\leq 2\kappa_\alpha^{(2)}\frac{F_{\I_-(\gamma)}}{M}.
\end{equation}
 In the last inequality we are again applying Corollary \ref{Cor: Corrolary_of_Lemma_aux}.

For the negative terms in \eqref{Eq: Interaction_between_interior_and_B}, we bound the term containing $\mathbbm{1}_{\{\sigma_x\neq\sigma_x\}}$ by $0$ and multiply the remaining terms by $\frac{1}{(2d+1)2^{\alpha+2}}$, getting 
\begin{equation*}
      \sum_{\substack{x \in \I_-(\gamma)\setminus V(\Gamma^{\prime}) \\ y \in  V(\Gamma^{\prime\prime}) }}  J_{xy} +  \sum_{\substack{x \in V(\Gamma^{\prime}) \\ y \in  V(\Gamma^{\prime\prime}) }} 2J_{xy}\mathbbm{1}_{\{\sigma_x\neq \sigma_y\}} +  \sum_{\substack{x \in \I_-(\gamma) \\ y \in  B(\gamma)\setminus V(\Gamma^{\prime\prime}) }}J_{xy} \geq \frac{1}{(2d+1)2^{\alpha + 2}}\left[ \sum_{\substack{x \in \I_-(\gamma)\setminus V(\Gamma^{\prime}) \\ y \in  V(\Gamma^{\prime\prime}) }}  J_{xy}  +  \sum_{\substack{x \in \I_-(\gamma) \\ y \in  B(\gamma)\setminus V(\Gamma^{\prime\prime}) }}J_{xy} \right].
\end{equation*}
Using the second inequality of \eqref{Eq: Aux_1_Interaction_between_interior_and_B}, we have 
\begin{equation}\label{Eq: Aux_4_Interaction_between_interior_and_B}
\begin{split}
  \sum_{\substack{x \in \I_-(\gamma)\setminus V(\Gamma^{\prime}) \\ y \in  V(\Gamma^{\prime\prime}) }} J_{xy} + \sum_{\substack{x \in \I_-(\gamma) \\ y \in  B(\gamma)\setminus V(\Gamma^{\prime\prime}) }} J_{xy} 
    &= F_{\I_-(\gamma)} - \sum_{\substack{x \in V(\Gamma^{\prime}) \\ y \in  V(\Gamma^{\prime\prime})}} J_{xy} -  \sum_{\substack{x \in \I_-(\gamma)\\ y \in  \Sp(\gamma)}} J_{xy}  \\
    &\geq \Big(1 -  \frac{2\kappa^{(2)}_\alpha}{M^{(\alpha - d)\wedge 1}}\Big)F_{\I_-(\gamma)} - F_{\Sp(\gamma)}.
    \end{split}
\end{equation}
 Plugging inequalities \eqref{Eq: Aux_1_Interaction_between_interior_and_B}, \eqref{Eq: Aux_2_Interaction_between_interior_and_B} and \eqref{Eq: Aux_4_Interaction_between_interior_and_B} back in \eqref{Eq: Interaction_between_interior_and_B} we get
\begin{align}\label{Eq: Aux_2_Diferenca_de_Hamiltonianos}
    \sum_{\substack{x \in \I_-(\gamma) \\ y \in B(\gamma)}} J_{xy}\sigma_x\sigma_y \leq\left( \frac{6\kappa^{(2)}_\alpha}{M^{(\alpha - d)\wedge 1}}-\frac{1}{(2d+1)2^{\alpha+2}}\right)F_{\I_-(\gamma)} + \frac{F_{\Sp(\gamma)}}{(2d+1)2^{\alpha+2}}
\end{align}

For the positive terms in \eqref{Eq: Difference_of_Hamiltonians_1}, we use the triangular inequality to get 
\begin{equation*}
    J_{xy} \geq \frac{1}{(2d+1)2^\alpha}\sum_{|x-x^\prime|\leq 1}J_{x^\prime y},
\end{equation*}
and therefore 
\begin{equation}\label{Eq: Aux_3_Diferenca_de_Hamiltonianos}
    \sum_{\substack{x \in \Sp(\gamma) \\ y \in \Z^d}} J_{xy}\mathbbm{1}_{ \{ \sigma_x \neq \sigma_y \}} +       \sum_{\substack{x \in \Sp(\gamma) \\ y \in \Sp(\gamma)^c}} J_{xy}\mathbbm{1}_{ \{ \sigma_x \neq \sigma_y \}} \geq \frac{1}{(2d+1)2^\alpha}\left(Jc_\alpha |\gamma| + F_{\Sp(\gamma)}\right),
\end{equation}
with $c_\alpha = \sum_{\substack{y\in\Z^d\setminus\{0\}}}|y|^{-\alpha}$. Plugging \eqref{Eq: Aux_1_Diferenca_de_Hamiltonianos}, \eqref{Eq: Aux_2_Diferenca_de_Hamiltonianos} and \eqref{Eq: Aux_3_Diferenca_de_Hamiltonianos} back in \eqref{Eq: Difference_of_Hamiltonians_1} we get
\begin{align*}
     H_\Lambda^+(\sigma) - H_\Lambda^+(\tau) \geq \frac{Jc_\alpha }{(2d+1)2^\alpha} |\gamma| + \left(\frac{1}{(2d+1)2^{\alpha + 1}} -  \frac{12\kappa^{(2)}_\alpha}{M^{(\alpha - d)\wedge 1}}\right)F_{\I_-(\gamma)} + \left( \frac{1 }{(2d+1)2^{\alpha+1}} -  \frac{2\kappa^{(2)}_\alpha}{M^{(\alpha - d)\wedge 1}}\right)F_{\Sp(\gamma)},
\end{align*}
what proves the proposition for $M^{(\alpha - d)\wedge 1}>24\kappa^{(2)}_\alpha2^{\alpha + 1}(2d+1)$. 

\end{proof}

\section{Ding and Zhuang approach}
The main idea used in Ding and Zhuang's proof of phase transition in \cite{Ding2021} is to make the Peierls' argument on the joint space of the configurations and the external fields and, when erasing a contour, perform in the external field the same flips you do in the configuration. Doing this, the part on the Hamiltonian that depends on the external field does not change, but the partition function does. The complication of this method is to control such differences.

In the short-range case, the spins that need to be flipped to erase a contour are precisely the ones in the interior of it. This is not the case for the long-range model, so we make a slight modification in the argument, and instead of performing the same flips in both spaces, we flip the external field only on $\I_-(\gamma)$. Doing this, not only does the partition function change, but we also get an extra cost when comparing the original energy with the energy after performing such a transformation. This extra term depends only on the external field in $\Sp{(\gamma)}$.  

In this section, we define the measure in the joint space and show that, with high probability, the change of partition function resulting from such flipping is upper-bounded by the size of the support $|\gamma|$, with high probability.

    \subsection{Joint measure and bad events}  
    Given $\Lambda\subset\Z^d$, a contour associated with a configuration in $\Omega_\Lambda^+$ is not always inside $\Lambda$. To avoid this, we consider the event $\Theta_{\Lambda} \coloneqq \{ \sigma : \sigma_x \text{ is } +\text{-correct for all }x\in\fint \Lambda \}$ and the conditional measure 
\begin{equation}
    \nu_{\Lambda; \beta, \varepsilon h}^{+}(\mathcal{A}) \coloneqq \mu_{\Lambda; \beta,\varepsilon h}^+(\mathcal{A} | \Theta_{\Lambda})
\end{equation}
for any $\mathcal{A}\subset \Omega$ measurable. The Markov property guarantees that $\nu_{\Lambda; \beta, \varepsilon h}^{+}$ is also a local Gibbs measure, with the advantage that all contours associated with it are inside $\Lambda$. Define the joint measure for $(\sigma, h)$ as

\begin{equation*}
    \mathbb{Q}_{\Lambda; \beta, \varepsilon}^+(\sigma \in \mathcal{A}, h\in \mathcal{B}) \coloneqq \int_{\mathcal{B}} \nu_{\Lambda;\beta, \varepsilon h}^+(\mathcal{A}) d\mathbb{P}(h),
\end{equation*}
for $\mathcal{A}\subset\Omega$ measurable and $\mathcal{B}\subset \mathbb{R}^{\Lambda}$ a Borel set. This measure $\mathbb{Q}_{\Lambda;\beta,\varepsilon}$ has density
\begin{equation*}
    g_{\Lambda; \beta, \varepsilon}^+(\sigma, h) \coloneqq \prod_{x\in\Lambda}\frac{1}{\sqrt{2\pi}}e^{-\frac{1}{2}h_x^2} \times \nu_{\Lambda;\beta, \varepsilon h}^+(\sigma).
\end{equation*}

The operation $\tau_{\gamma}$ used to remove a contour $\gamma\in\Gamma(\sigma)$ can be written as a particular case of the following one: given $A\subset\Z^d$, take $\tau_A:\mathbb{R}^{\Z^d} \xrightarrow{} \mathbb{R}^{\Z^d}$ as 
\begin{equation}
    (\tau_A(\sigma))_x \coloneqq \begin{cases}
                        -\sigma_x &\text{if }x\in A,\\
                        \sigma_x   &\text{otherwise},
                      \end{cases}
\end{equation}
for every $x\in\Z^d$. Defining $\s(\gamma, \sigma)^\pm\coloneqq \{ x\in \s(\gamma): \sigma_x = \pm 1\}$, the transformation that erases a contour $\gamma$ is $\tau_\gamma(\sigma) = \tau_{\I_-(\gamma)\cup \s^-(\gamma,\sigma)}(\sigma)$.

The main idea used in the proof of phase transition in \cite{Ding2021} is to make the Peierls' argument on the measure $\mathbb{Q}_{\Lambda;\beta,\varepsilon}$, and perform in the external field the same flips one does in the configuration when erasing a contour. Formally, in \cite{Ding2021} they compare the density $g_{\Lambda; \beta, \varepsilon}^+(\sigma, h)$ with the density after erasing a contour $\gamma\in\Gamma(\sigma)$, and performing the same flips on the external field. For the short-range model, the spins that need to be flipped to erase a contour are precisely the ones in the interior of it. This is not the case for the long-range setting, so we compare $g_{\Lambda; \beta, \varepsilon}^+(\sigma, h)$ with the density after erasing $\gamma$ and flipping the external field only in $\I_-(\gamma)$, getting

\begin{align}\label{Eq: quotient.of.gs}
    \frac{g_{\Lambda; \beta, \varepsilon}^+(\sigma, h)}{g_{\Lambda; \beta, \varepsilon}^+(\tau_{\gamma}(\sigma),\tau_{\I_-(\gamma)}(h))} 
    &= \exp{\{\beta H_{\Lambda, \varepsilon \tau_{\I_-(\gamma)}(h)}^{+}(\tau_{\gamma}(\sigma)) - \beta H_{\Lambda, \varepsilon h}^{+}(\sigma)8\}}\frac{Z_{\Lambda; \beta, \varepsilon}^{+}(\tau_{\I_-(\gamma)}(h))}{Z_{\Lambda; \beta, \varepsilon}^{+}(h)}  \nonumber \\ 
    &\leq \exp{\{- \beta c_2 |\gamma| -2\beta\sum_{x\in \Sp^-(\gamma,\sigma)}\varepsilon h_x\}}\frac{Z_{\Lambda; \beta, \varepsilon}^{+}(\tau_{\I_-(\gamma)}(h))}{Z_{\Lambda; \beta, \varepsilon}^{+}(h)}.
\end{align}
where the constant $c_2$ is the one given by Proposition \ref{Prop: Cost_erasing_contour}.

The sum of the external field in $\Sp^-(\gamma,\sigma)$ can be shown to be of order $|\Sp^-(\gamma,\sigma)|$, and do not influence the Peierls' argument. However, the quotient of the partition functions can be bigger than the exponential term. Denoting
\begin{equation}
\Delta_A(h) \coloneqq -\frac{1}{\beta}\log{\frac{Z_{\Lambda; \beta, \varepsilon}^{+}(h)}{Z_{\Lambda; \beta, \varepsilon}^{+}(\tau_{A}(h))}},
\end{equation}
 for every $A\subset \Z^d$, the bad event are

$$ \mathcal{E}_0^c \coloneqq \left\{ \sup_{ \gamma\in\mathcal{C}_0, \ \sigma\in\Omega(\gamma)} \frac{2\left|\sum_{x\in \Sp^-(\gamma,\sigma)}\varepsilon h_x\right|}{c_2|\gamma|} > \frac{1}{4}\right\}$$
and
$$\mathcal{E}_1^c\coloneqq \left\{\sup_{\substack{\gamma\in\mathcal{C}_0}} \frac{|\Delta_{\I_-(\gamma)}(h)|}{c_2|\gamma|} > \frac{1}{4}\right\}.$$

To control the probability of these bad events,  we need a concentration result for Gaussian random variables. The following one is due to M. Ledoux and M. Talagrand, and a proof can be found in \cite{Ledoux.Talagrand.91}. 

\begin{theorem}\label{Theo: Gaussian.concentration}
    Let $f:\mathbb{R}^M \xrightarrow[]{} \mathbb{R}$ be a uniform Lipschitz continuous function with constant $C_{Lip}$, that is, for any $X,Y\in\mathbb{R}^M$, $$|f(X) - f(Y)| \leq C_{Lip} || X - Y ||_2 .$$ 
    
    Then, if $X_1,\dots, X_M$ are i.i.d. Gaussian random variables with variance 1,
    \begin{equation}\label{Eq: Tail.concentration}
        \mathbb{P}\left(|f(X_1,\dots, X_M) - \mathbb{E}(f(X_1, \dots, X_M))|\geq z\right) \leq 2\exp{\left\{\frac{-z^2}{2C_{Lip}^2}\right\}}.
    \end{equation}
\end{theorem}

\begin{remark}\label{Rmk: MVT.Lipschitz}
    If $f$ is differentiable and $||\nabla f(\cdot)||_2$ is bounded, the mean value theorem guarantees that $\sup_{Z\in\mathbb{R}^M}||\nabla f(Z)||_2$ is a uniform Lipschitz constant for $f$. 
\end{remark}
\begin{remark}\label{Rmk: Bernoulli_external_field}
    If $f$ has a compact support and convex level sets, an equation similar to \eqref{Eq: Tail.concentration} holds, with some adjustments on the constants and replacing the mean by the median, see \cite[Theorem 7.1.3]{Bovier.06}. Therefore, our results hold when $h_i$ has a Bernoulli distribution $\mathbb{P}(h_i=+1) =\mathbb{P}(h_i=-1)= \frac{1}{2}$. 
\end{remark}

Given $A\subset\Z^d$, $h_A\coloneqq (h_x)_{x\in A}$ denotes the restriction of the external field to the subset $A$. The next Lemma was proved in \cite{Ding2021}, and is a direct consequence of the previous theorem, replacing $f$ by functions of the form $\Delta_A$, with $A\Subset\Z^d$. As the proofs for the short and long-range are the same, we omit it. 

\begin{lemma}\label{Lemma: Concentration.for.Delta.General}
    For any $A, A^\prime \Subset \mathbb{Z}^d$ and $\lambda>0$, we have 
\begin{equation}\label{Eq: Tail.of.Delta_A}
    \mathbb{P}\left(|\Delta_A(h)| \geq \lambda \vert h_{A^c}\right) \leq2e^{\frac{-\lambda^2}{8\varepsilon^2 |A|}},
\end{equation}
and 
\begin{equation}\label{Eq: Tail.of.the.diff.of.Deltas}
     \mathbb{P}(|\Delta_{A}(h) - \Delta_{A^\prime}(h)|>\lambda|h_{{A \cup A^\prime}^c}) \leq  2e^{-\frac{{\lambda^2}}{{8\varepsilon^2|A \Delta A^\prime|}}},
\end{equation}
where $A\Delta A^\prime$ is the symmetric difference. Similarly, 

\begin{equation}\label{Eq: Tail.of.sum.in.A}
    \mathbb{P}\left(\left|\sum_{i\in A}\varepsilon h_i\right| \geq \lambda \vert h_{A^c}\right) \leq 2e^{\frac{-\lambda^2}{2\varepsilon^2 |A|}}.
\end{equation}
\end{lemma}

\begin{remark}\label{Rmk: More.general.h_x}
    Lemma \ref{Lemma: Concentration.for.Delta.General} holds whenever $h=(h_x)_{x\in\Z^d}$ satisfy equation \eqref{Eq: Tail.concentration}. As a consequence, our results can be stated for more general external fields. 
\end{remark}

The first main application of Equation \eqref{Eq: Tail.of.sum.in.A} is to control the probability of $\mathcal{E}_0^c$.
\begin{proposition}\label{Prop: Bound.bad.event.0} 
    There exists $C_0\coloneqq C_0(\alpha, d)$ such that $\mathbb{P}(\mathcal{E}_0^c)\leq e^{-{C_0}/{\varepsilon^2}}$.
\end{proposition}
\begin{proof}
This is a direct consequence of Lemma \ref{Lemma: Concentration.for.Delta.General} and Corollary \ref{Cor: Bound_on_C_0_n}, which guarantees that ${|\mathcal{C}_0(n)| \leq e^{c_1 n}}$, thus
\begin{equation}\label{Eq: Bound.bad.event.0.eq.1}
    \begin{split}
    \mathbb{P}(\mathcal{E}_0^c) &\leq \sum_{n\geq 1}\mathbb{P}\left(\sup_{\gamma\in\mathcal{C}_0(n), \ \sigma\in\Omega(\gamma)} 2\varepsilon\left|\sum_{i\in \Sp^-(\gamma,\sigma)}h_i\right| > \frac{c_2}{4}n\right)  \\
    &\leq \sum_{n\geq 1}\sum_{\gamma\in \mathcal{C}_0(n)}\left| \Omega(\gamma) \right| \sup_{\substack{\gamma\in \mathcal{C}_0(n) \\ \sigma \in \Omega(\gamma)}}\mathbb{P}\left(\varepsilon\left|\sum_{i\in \Sp^-(\gamma,\sigma)}h_i\right| > \frac{c_2}{8}n\right) \\
    &\leq \sum_{n\geq 1}|\mathcal{C}_0(n)|2^n\sup_{\substack{\gamma\in \mathcal{C}_0(n) \\ \sigma \in \Omega(\gamma)}}\mathbb{P}\left(\varepsilon\left|\sum_{i\in \Sp^-(\gamma,\sigma)}h_i\right| > \frac{c_2}{8}n\right) \\
    &\leq \sum_{n\geq 1} e^{c_1 n}2^n\sup_{\substack{\gamma\in \mathcal{C}_0(n) \\ \sigma \in \Omega(\gamma)}}\mathbb{P}\left(\varepsilon\left|\sum_{i\in \Sp^-(\gamma,\sigma)}h_i\right| > \frac{c_2}{8}n\right),
\end{split}
\end{equation}
where in the third inequality we used that the number of configurations $\sigma$ compatible with a contour $\gamma$ is at most $2^{|\gamma|}$. Lemma \ref{Lemma: Concentration.for.Delta.General} yields that, for any $\gamma\in\mathcal{C}_0{(n)}$,
\begin{equation}\label{Eq: Bound.bad.event.0.eq.2}
    \mathbb{P}\left(\varepsilon\left|\sum_{i\in \Sp^-(\gamma,\sigma)}h_i\right| > \frac{c_2}{8}n\right) \leq 2 \exp{\frac{-c_2^2n^2}{64\varepsilon^2|\Sp^-(\gamma,\sigma)|}} \leq 
    2\exp{\frac{-c_2^2n}{64\varepsilon^2}}. 
\end{equation}
In the last equation, we just used that $|\Sp^-(\gamma,\sigma)|\leq |\gamma|$. Equation \eqref{Eq: Bound.bad.event.0.eq.1} together with \eqref{Eq: Bound.bad.event.0.eq.2} yields
\begin{align*}
     \mathbb{P}(\mathcal{E}_0^c) &\leq \sum_{n\geq 1} 2e^{(c_1+\log{2})n - \frac{c_2^2}{64\varepsilon^2}n} \leq \sum_{n\geq 1} 2e^{\frac{-2C_0}{\varepsilon^2}n},
\end{align*}
what concludes the proof for an appropriate choice of $C_0$ and $\varepsilon$ small enough.
\end{proof}

To control the probability of $\mathcal{E}_1^c$ we use a coarse-graining argument presented in \cite{FFS84} together with some ideas of entropy bounds used in \cite[Section 3.2]{Affonso.2021}, where it is proved that the number of contours containing $0$ in its volume grows at most exponentially with the size of the contour. As pointed out by \cite{Ding2021}, the proof presented in \cite{FFS84}, despite being self-contained, is a non-trivial application of Dudley's entropy bound. Next, we adapt the proof presented in \cite{FFS84} using this entropy bound. 
    \subsection{Probability Results} 

    To control the probability of $\mathcal{E}_1^c$, we use some results on majorizing measures. For an extensive overview, we refer to \cite{Talagrand_14}. Consider $(T,\d)$ a finite metric space and a process $(X_t)_{t\in T}$ such that, for every $\lambda>0$ and $t,s\in T$,
\begin{equation}\label{Eq: Sub_gaussian_def}
    \mathbb{P}\left( |X_t - X_s| \geq \lambda \right) \leq 2\exp{\frac{-\lambda^2}{2\d(s,t)^2}}.
\end{equation}
Assume also that $\mathbb{E}\left(X_t\right) = 0$ for every $t\in T$. One example of such process is $( \Delta_{\I_-(\gamma)})_{\gamma\in\mathcal{C}_0(n)}$, with the distance $\d_2(A,A^\prime) = 2\varepsilon |A\Delta A^\prime|^{\frac{1}{2}}$ over the set $\I_-(n)\coloneqq \{\I_-(\gamma) : \gamma\in\mathcal{C}_0(n)\}$. For $n\in \mathbb{N}$, consider the quantities $N_n = 2^{2^n}$ and $N_0=1$. 

\begin{definition}
    Given a set $T$, a sequence $(\mathcal{A}_n)_{n\geq 0}$ of partitions of $T$ is \textit{admissible} when $|\mathcal{A}_n|\leq N_n$ and $\mathcal{A}_{n+1}\preceq \mathcal{A}_n$ for all $n\geq 0$.
\end{definition}

Given $t\in T$ and an admissible sequence $(\mathcal{A}_n)_{n\geq 0}$, $A_n(t)$ denotes the element of $\mathcal{A}_n$ that contains $t$. 

\begin{definition}
    Given $\theta > 0$ and a metric space $(T,\d)$, we define
    \begin{equation*}
        \gamma_\theta(T,\d) \coloneqq \inf_{(\mathcal{A}_n)_{n\geq 0}}\sup_{t\in T}\sum_{n\geq 0}2^{\frac{n}{\theta}}\diam(A_n(t)),
    \end{equation*}
where the infimum is taken over all admissible sequences. 
\end{definition}

\begin{theorem}[Majorizing Measure Theorem \cite{Talagrand_87}]\label{MMT} There is a universal constant $L>0$ such that
\begin{equation*}
    \frac{1}{L}\gamma_2(T,\d) \leq \mathbb{E}\left( \sup_{t\in T} X_t \right) \leq L\gamma_2(T,\d).
\end{equation*}
\end{theorem}

    Given $\epsilon>0$, let $N(T,\d, \epsilon)$ be the minimal number of balls with radius $\epsilon$ necessary to cover $T$, using the distance $\d$. 
\begin{proposition}[Dudley's Entropy Bound \cite{Dudley67}]
    Let $(X_t)_{t\in T}$ be a family of centered random variables satisfying \eqref{Eq: Sub_gaussian_def} for some distance $\d$. Then there exists a constant $\overline{L}>0$ such that 
    \begin{equation*}
        \mathbb{E}\left[\sup_{t\in T}X_t\right]\leq \overline{L}\int_{0}^\infty \sqrt{\log N(T,\d,\epsilon)}d\epsilon.
    \end{equation*}
\end{proposition}

We also need the following result. 
\begin{theorem}\label{Theo: Theo_2.2.27_Talagrand} Given a metric space $(T,\d)$ and a family $(X_t)_{t\in T}$ of centered random variables satisfying \eqref{Eq: Sub_gaussian_def}, there is a universal constant $L>0$ such that, for any $u>0$,
\begin{equation*}
\mathbb{P}\left( \sup_{t\in T}X_t > L(\gamma_2(T,\d) + u\diam(T)) \right)\leq e^{-{u^2}},
\end{equation*}
where the $\diam(T)$ is the diameter taken with respect to the distance $\d$
\end{theorem} 
A proof can be found in  \cite[Theorem 2.2.27]{Talagrand_14}. Using these results, the bound on the bad event $\mathcal{E}_1^c$ follows from the next proposition.
\begin{proposition}\label{Prop: Bound.gamma_2}
    Given $n\geq 0$, $d\geq 3$ and $\alpha > d$, there is a constant $L_1 \coloneqq L_1(d,\alpha)>0$  such that $$\gamma_2(\I_-(n),\d_2) \leq \varepsilon L_1 n.$$
\end{proposition}
As a direct consequence of this Proposition, we can control the probability of the bad event.
\begin{proposition}\label{Prop: Bound.bad.event.1}     
    There exists $C_1\coloneqq C_1(\alpha, d)$ such that $\mathbb{P}(\mathcal{E}_1^c)\leq e^{-\frac{C_1}{\varepsilon^2}}$ for any $\varepsilon^2<C_1$. 
\end{proposition}

\begin{proof}   
 By the union bound,
\begin{align}\label{Eq: Union_bound_bad_event}
    \mathbb{P}\left({\sup_{\substack{\gamma\in\mathcal{C}_0}} \frac{\Delta_{\I_-(\gamma)}(h)}{c_2|\gamma|} > \frac{1}{4}}\right) \leq \sum_{n=2}^\infty \mathbb{P}\left({\sup_{\substack{\gamma\in\mathcal{C}_0(n)}} \Delta_{\I_-(\gamma)}(h) > \frac{c_2}{4}}|\gamma|\right). 
\end{align}
Let $\gamma,\gamma^\prime\in \mathcal{C}_0(n)$ be two contours satisfying $\diam(\I_-(n)) = \d_2(\I_-(\gamma),\I_-(\gamma^\prime))$, where the diameter is in the $\d_2$ distance. By the isoperimetric inequality,
\begin{equation*}
    \diam(\I_-(n))= 2\varepsilon{|\I_-(\gamma)\Delta \I_-(\gamma^\prime)|}^{\frac{1}{2}} \leq 2\sqrt{2}\varepsilon n^{(\frac{d}{d-1})\frac{1}{2}} = 2\sqrt{2}\varepsilon n^{(\frac{1}{2} + \frac{1}{2(d-1)})}.
\end{equation*}
Together with Proposition \ref{Prop: Bound.gamma_2}, this yields
    \begin{align*}
  \frac{c_2}{4}|\gamma| &= L\left[\varepsilon L_1 n + \varepsilon L_1 \left(\frac{c_2}{4\varepsilon L_1 L} - 1\right)n\right]\\
    &\geq   L\left[\gamma_2(\I_-(n),\d_2) +  \frac{C_1^\prime}{\varepsilon}n^{\frac{1}{2} - \frac{1}{2(d-1)}}\diam(\I_-(n))\right],
\end{align*}

with $C_1^\prime = \frac{c_2}{16\sqrt{2}L}$ and $\varepsilon<\frac{c_2}{8L_1L}$. Applying Theorem \ref{Theo: Theo_2.2.27_Talagrand} with $u = \frac{C_1^\prime}{\varepsilon}n^{\frac{1}{2} - \frac{1}{2(d-1)}}$, we have
\begin{align*}
    \mathbb{P}\left({\sup_{\substack{\gamma\in\mathcal{C}_0(n)}} \Delta_{\I_-(\gamma)}(h) > \frac{c_2}{4}}|\gamma|\right) &=  \mathbb{P}\left({\sup_{\substack{\I\in\I_-(n)}} \Delta_{\I}(h) > \frac{c_2}{4}}n\right) \\
    &\leq \mathbb{P}\left({\sup_{\substack{\I\in\I_-(n)}} \Delta_{\I}(h) >    L\left[\gamma_2(\I_-(n),\d_2) +  \frac{C_1^\prime}{\varepsilon}n^{\frac{1}{2} - \frac{1}{2(d-1)}}\diam(\I_-(n))\right]}\right) \\ 
    &\leq \exp{\left\{ - \frac{C_1^{\prime2}n^{1 - \frac{1}{(d-1)}}}{\varepsilon^2}\right\}}. 
\end{align*}
Using this back in equation \eqref{Eq: Union_bound_bad_event}, we conclude that 
\begin{equation*}
       \mathbb{P}\left({\sup_{\substack{\gamma\in\mathcal{C}_0}} \frac{\Delta_{\I_-(\gamma)}(h)}{c_2|\gamma|} > \frac{1}{4}}\right) \leq \sum_{n=2}^\infty \exp{\left\{ - \frac{C_1^{\prime2}n^{1 - \frac{1}{(d-1)}}}{\varepsilon^2}\right\}} \leq e^{-\frac{C_1}{\varepsilon^2}},
\end{equation*}
for a suitable constant $C_1\coloneqq C_1(\alpha, d)$ smaller than $\frac{{C_1^\prime}^2}{2}$ and $\varepsilon< C_1$. The dependency on $\alpha$ is due to the dependency on $c_2(\alpha, d)$.

\end{proof}

The next two subsections are dedicated to proving Proposition \ref{Prop: Bound.gamma_2}.

    \subsection{Coarse-graining}\label{Subsection: Coarse-graining} 

    We will apply the previous probability estimates for the family $(\Delta_{\I}(h))_{\I\in\I_-(n)}$. We use the coarse-graining idea introduced in \cite{FFS84} to construct the cover by balls in Dudley's entropy bound.  For each $\ell>0$ and each contour ${\gamma\in\mathcal{C}_0(n)}$, we will associate a region $B_\ell(\gamma)$ that approximates the interior $\I_-(\gamma)$ in a scaled lattice, with the scale growing with $\ell$. This is done in a way that two interiors approximated by the same region are in a ball in the distance $\d_2$ with a fixed radius, depending on $\ell$.

An $r\ell$-cube $C_{r\ell}$ is \textit{admissible} if more than a  half of its points are inside $\I_-(\gamma)$. Thus, the set of admissible cubes is
\begin{equation*}
    \mathfrak{C}_\ell(\gamma) \coloneqq \left\{C_{r\ell} : |C_{r\ell}\cap \I_-(\gamma)| \geq \frac{1}{2}|C_{r\ell}|\right\}.
\end{equation*}
With this notion of admissibility, two contours with the same admissible cubes should be close in distance $\d_2$. Consider functions $B_\ell:\mathcal{E}^+_\Lambda \xrightarrow[]{} \mathcal{P}(\Z^d)$, with $\mathcal{P}(\Z^d) \coloneqq \{A:A\Subset\Z^d\}$, that takes contours $\gamma$ to $B_\ell(\gamma) \coloneqq B_{\mathfrak{C}_{\ell}(\gamma)}$, the region covered by the admissible cubes. We will be interested in counting  the image of $B_\ell$ by $\mathcal{C}_0(n)$, that is, $|B_\ell(\mathcal{C}_{0}(n))| = |\{B:B=B_\ell(\gamma)\text{ for some }\gamma \in \mathcal{C}_0(n)\}|$. Notice that $B_\ell(\gamma)$ is uniquely determined by $\partial B_\ell(\gamma)$. Given any collection $\C_{m}$, we define the \textit{edge boundary of } $\C_m$ as 
$$
\partial \C_m \coloneqq \{ \{C_{m}, C^\prime_{m}\} : C_{m} \in \C_m, \ C_m^\prime \notin \C_m \textrm{ and} \  C_m^\prime \text{ shares a face with }C_m\}.
$$ 
We also define the \textit{inner boundary of }$\C_m$ as
$$
\fint \C_m\coloneqq \{ C_{m}\in \C_m : \exists C_m^\prime \notin \C_m \textrm{ such that }  \{C_m,C_m^\prime\}\in\partial \C_m\}.
$$ 
With this definition, it is clear that $\partial B_\ell(\gamma)$ is uniquely determined by $\partial \mathfrak{C}_\ell(\gamma)$. Hence, defining $\partial \mathfrak{C}_{r\ell}(\mathcal{C}_0(n)) \coloneqq \{\partial\C_{r\ell} : \C_{r\ell}=\mathfrak{C}_\ell(\gamma )\text{ for some }\gamma \in \mathcal{C}_0(n)\}$, we have $|B_\ell(\mathcal{C}_0(n))| = |\partial\mathfrak{C}_\ell(\mathcal{C}_0(n))|$. In a similar fashion we define $\fint \mathfrak{C}_{r\ell}(\mathcal{C}_0(n)) \coloneqq \{\fint\C_{r\ell} : \C_{r\ell}=\mathfrak{C}_\ell(\gamma )\text{ for some }\gamma \in \mathcal{C}_0(n)\}$.

The main result in this section is an upper bound on the number of cubes in $\fint\mathfrak{C}_\ell(\gamma)$. The precise statement is given in Proposition \ref{Proposition1}, which was stated originally in \cite{FFS84} for $d=3$ and $\I_-(\gamma)$ simply connected for all contours, but it can be extended to $d\geq 2$ with no restriction on the interiors, see \cite{Bovier.06}. As we could not find a detailed proof anywhere, we provide one here.

Given a rectangle $\mathcal{R} = [1,r_1]\times[1,r_2]\times\dots\times[1,r_d]$, consider $\R_i\coloneqq\{x\in \R : x_i=1\}$ the face of $\R$ that is perpendicular to the direction $e_i$, for $i=1,\dots,d$. The line that connects a point $x\in \R_i$ to a point in the opposite face of $\R_i$ is $\ell_x^i \coloneqq \{ x + ke_i : 1\leq k\leq r_i\}$. Given $A\subset \Z^d$, the projection of $A\cap \R$ into the face $\R_i$ is
\begin{equation*}
    \calP_i(A\cap\R) \coloneqq \{x\in\R_i : \ell_x^i \cap A \neq \emptyset\}.
\end{equation*}
\input{Figures/Figura.4}

In many situations, we will split the projections into \textit{good} and \textit{bad} points. The set of good points is $\calP_i^{G}(A\cap\R)\coloneqq \{x\in \calP_i(A\cap \R) : \ell_x^i \cap (\R\setminus A) \neq \emptyset\}$, that is, there exist a point in $\ell_x^i\cap \R$ that is not in $A$.  The bad points are defined as $\calP^{B}_i(A\cap\R) \coloneqq \calP_i(A\cap\R)\setminus \calP_i^G(A\cap\R)$. See Figures \ref{fig: Figura.4} and \ref{fig: Figura5} for examples. 

\begin{figure}[H] 
	\centering

\tikzset{every picture/.style={line width=0.75pt}} 

\begin{tikzpicture}[x=0.75pt,y=0.75pt,yscale=-1,xscale=1]

\draw   (392.25,208) -- (269.17,208) -- (269.17,12.4) -- (392.25,12.4) -- cycle ;
\draw  [fill={rgb, 255:red, 155; green, 155; blue, 155 }  ,fill opacity=1 ] (391.9,153.85) .. controls (391.19,193.32) and (364.9,155.43) .. (349.17,148.66) .. controls (333.45,141.89) and (270.23,182.59) .. (269.53,161.85) .. controls (268.82,141.1) and (269.17,122.11) .. (269.17,103.12) .. controls (269.17,84.13) and (298.36,97.85) .. (305.22,101.19) .. controls (312.07,104.53) and (365,118.77) .. (347.41,92.4) .. controls (329.83,66.03) and (306.51,41.19) .. (329.83,43.17) .. controls (353.15,45.14) and (340.5,51.68) .. (353.22,58.11) .. controls (365.93,64.54) and (372.44,68.26) .. (373.96,73.59) .. controls (375.49,78.91) and (391.68,71.37) .. (391.9,76.4) .. controls (392.12,81.43) and (392.6,114.38) .. (391.9,153.85) -- cycle ;
\draw  [dash pattern={on 0.84pt off 2.51pt}]  (268.82,62.33) -- (391.9,62.33) ;
\draw [shift={(268.82,62.33)}, rotate = 0] [color={rgb, 255:red, 0; green, 0; blue, 0 }  ][fill={rgb, 255:red, 0; green, 0; blue, 0 }  ][line width=0.75]      (0, 0) circle [x radius= 3.35, y radius= 3.35]   ;
\draw  [dash pattern={on 0.84pt off 2.51pt}]  (269.53,122.11) -- (392.6,122.11) ;
\draw [shift={(269.53,122.11)}, rotate = 0] [color={rgb, 255:red, 0; green, 0; blue, 0 }  ][fill={rgb, 255:red, 0; green, 0; blue, 0 }  ][line width=0.75]      (0, 0) circle [x radius= 3.35, y radius= 3.35]   ;

\draw (357.41,120.14) node [anchor=north west][inner sep=0.75pt]    {$A$};
\draw (253.5,53.62) node [anchor=north west][inner sep=0.75pt]  [font=\small]  {$p$};
\draw (252.68,111.87) node [anchor=north west][inner sep=0.75pt]  [font=\small]  {$p^{\prime }$};

\end{tikzpicture}

	\caption{Considering $A\cap\R$ the gray region, both points $p,p^\prime \in \mathcal{P}_1(A\cap \R)$ are in the projection, but $p$ is a good point and $p^\prime$ is a bad point. The doted lines represent $\ell_p^1$ and $\ell_{p^\prime}^1$.} \label{fig: Figura5}
\end{figure}

Given $x\in \calP_i^{G}(A\cap\R)$, by definition of the projection, there exists a point in $\ell_x^i\cap A$. Therefore, there exists a point $p\in \ell_x^i$ such that $p\in\fext A \cap \R$. As all lines are disjoint, we conclude that 
\begin{equation}\label{Eq: upper.bound.good.points}
     |\calP_i^{G}(A\cap\R)|\leq |\fext A \cap \R|.
\end{equation}
 We now prove two auxiliary lemmas.
 
\begin{lemma}\label{Lemma: Geo.discreta.1}
    Given $d\geq 2$, for any family of positive integers $(r_i)_{i=1}^d$ with $R\leq r_i \leq 2R$ for some $R\geq 2$, $0<\lambda < 1$ and $A\subset\Z^d$, there exists a constant $c\coloneqq c(d, \lambda)$ such that, if 
    \begin{equation}\label{Eq: hypothesis.lemma.1}
         |\calP_i(A\cap \R)| \leq \lambda|\R_i|
    \end{equation}
    for all $i= 1,\dots, d$, then 
    \begin{equation*}
        \sum_{i=1}^d |\calP_{i}(A\cap \R)|\leq c|\fext A\cap \R|,
    \end{equation*}
    where $\R=[1,r_1]\times\dots\times [1,r_d]$.
\end{lemma}

\begin{proof}
The proof will be done by induction on the dimension. For $d=2$, take a rectangle ${\R=[1,r_1]\times[1,r_2]}$. If there is no bad points in $\calP_1(A\cap\R)$, then 
\begin{align}\label{Eq: Bound.1.on.P.1}
    |\calP_1(A\cap\R)| = |\calP_1^G(A\cap \R)| \leq |\fext A \cap \R|.
\end{align}

If there is a bad point $p=(1,p_2)\in \calP_1^B(A\cap\R)$, $\ell_p^1\subset A\cap \R$  by definition of bad point. As $|\calP_1(A\cap \R)| \leq \lambda|\R_1| < |\R_1|$, there is a point $p^\prime = (1,p_2^\prime)\in \R_1\setminus \calP_1(A\cap \R)$ that is in the face $\R_1$ but not in the projection. By definition of the projection, $\ell_{p^\prime}^1\in A^c\cap \R$. Therefore, for any $1\leq k\leq r_1$, $(k,p_2)\in  A\cap \R$ and $(k,p^\prime_2)\in  A^c\cap \R$, we can find a point $p^k=(k, p^k_2) \in \fext A \cap \R$. Since $p^{k_1}\neq p^{k_2}$ for every $k_1\neq k_2$, we have $r_1 \leq |\fext A \cap \R|$, hence
\begin{equation}\label{Eq: Bound.2.on.P.1}
   |P_1(A\cap \R)| \leq  |\R_1| = {r_2}\leq  2R \leq 2r_1 \leq  2|\fext A \cap \R|.
\end{equation}
A completely analogous argument can be done to bound $|P_2(A\cap \R)|$, and we conclude that
\begin{equation*}
    \sum_{i=1}^2|\calP_i(A\cap \R)|\leq 4|\fext A \cap \R|,
\end{equation*}
and take $c(2,\lambda)=4$. Suppose the lemma holds for $d-1$ and fix a rectangle $\R=[1,r_1]\times\dots\times[1,r_d]$. We split $\R$ into layers $L_k = \{x\in\Z^d : x_d = k\}$, for $k=1,\dots, r_d$. We can then partition the projection and write
\begin{equation*}
|\calP_i(A\cap \R)| = \sum_{k=1}^{r_d} |\calP_i(A\cap \R)\cap L_k|,    
\end{equation*}
for any $i\in\{1,\dots, d-1\}$. This yields

\begin{align}\label{Eq: Partition.sum.proj.}
    \sum_{i=1}^d|\calP_i(A\cap \R)|    &=  \sum_{k=1}^{r_d}\sum_{i=1}^{d-1}|\calP_i(A\cap \R)\cap L_k| + |\calP_d(A\cap \R)|.
\end{align}
Notice now that $\calP_i(A\cap \R)\cap L_k = \calP_i(A\cap (\R\cap L_k))$. Defining the rectangle $\R^k \coloneqq \R\cap L_k$, for every point $p\in \calP_j^B(A\cap \R^k)$, $\ell_p^j \subset A\cap \R^k$. Moreover, we can associate every point $x\in \ell_p^j$ in the line with a point $x^\prime\in \calP_d(A\cap\R)$ by taking $x_m^\prime = x_m$ for $m \leq d-1$ and $x_d^\prime = 1$, therefore

\begin{equation*}
    r_j|\calP_j^B(A\cap \R^k)| = \sum_{p\in \calP_j^B(A\cap \R^k)}|\ell_p^j| \leq |\calP_d(A\cap\R)|.
\end{equation*}

\begin{figure}[H] 
	\centering
	\includegraphics[scale=0.13]{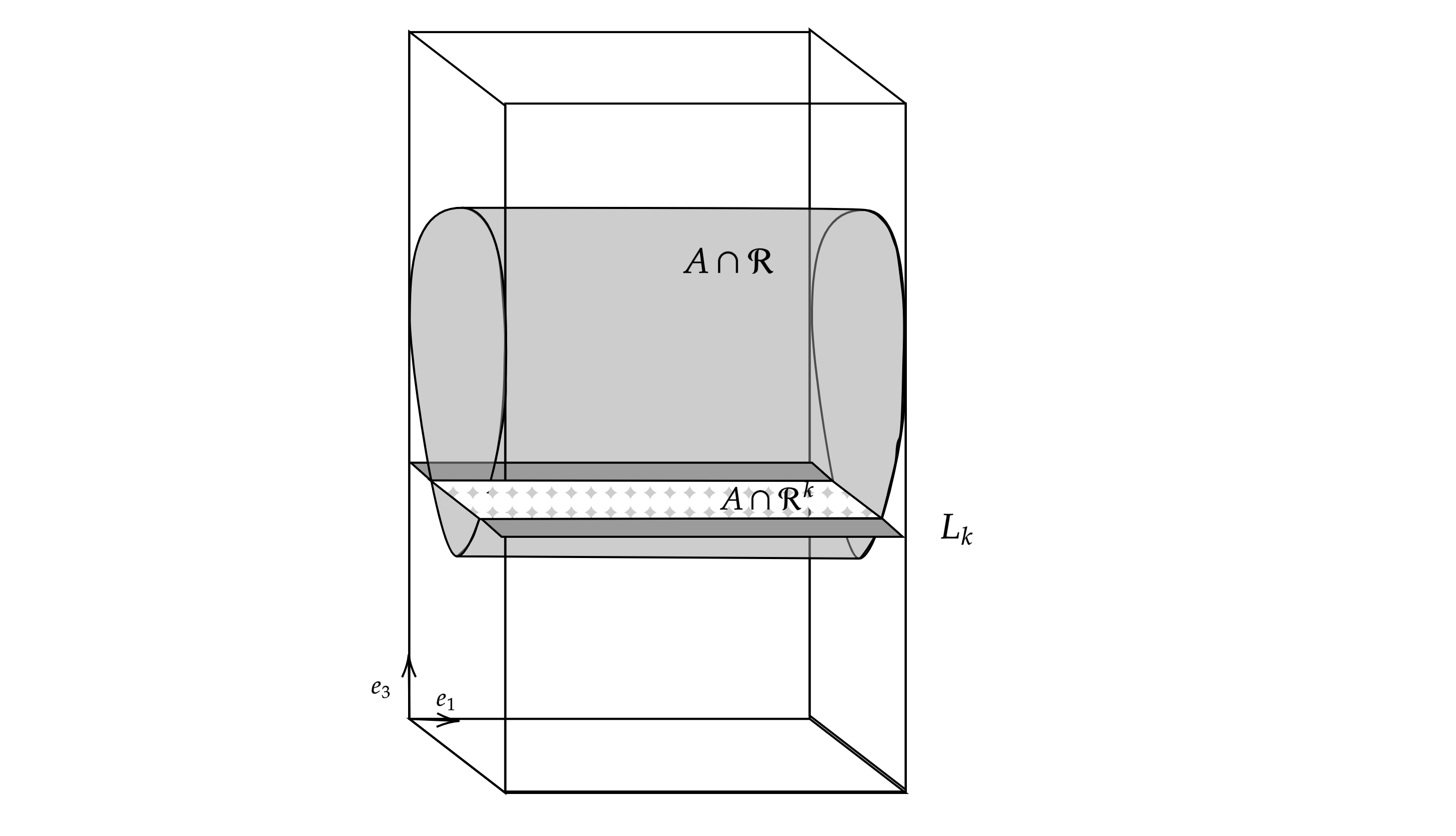}
 
	\caption{Considering $A\cap\R$ the light gray region, the dark gray region is a layer $L_k$  and the dotted region is $A\cap\R^k$, the restriction of $A$ to the layer $L_k$.} \label{fig: Figura6}
\end{figure}

Using the hypothesis \eqref{Eq: hypothesis.lemma.1} we conclude that

\begin{equation}\label{Eq: upper.bound.projection.i.bad.points}
     |\calP_j^B(A\cap \R^k)| \leq \lambda\frac{|\R_d|}{r_j} = \lambda \frac{ \prod_{q\neq d}r_q}{r_j} =  \lambda \prod_{q\neq j,d}r_q = \lambda |(\R^k)_j|.
\end{equation}
We consider two cases:
    \begin{itemize}
        \item[(a)] If $|\calP_i(A\cap \R^k)| \leq \frac{\lambda +1}{2}|(\R^k)_i|$, for all $i\leq d-1$, then we are in the hypothesis of the lemma in $d-1$ and therefore
\begin{equation}\label{Eq: Primeiro.bound.soma.projecoes}
    \sum_{i=1}^{d-1} |\calP_i(A\cap \R^k)| \leq c\left(d-1, \frac{\lambda + 1}{2}\right)|\fext A\cap \R^k|.
\end{equation}
    \item [(b)] If there exists $j\in\{1,\dots,d-1\}$ satisfying $|\calP_j(A\cap \R^k)| > \frac{\lambda +1}{2}|(\R^k)_j|$, by \eqref{Eq: upper.bound.projection.i.bad.points} we have $|\calP_j^G(A\cap \R^k)| = |\calP_j(A\cap \R^k)| - |\calP_j^B(A\cap \R^k)| \geq \frac{1-\lambda}{2}|(\R^k)_j|$, hence
\begin{equation*}
    |(\R^k)_j| \leq  \frac{2}{1 - \lambda}|\fext A \cap \R^k|.
\end{equation*}
    Using that $|(\R^k)_i| \leq (2R)^{d-2} \leq 2^{d-2}|(\R^k)_j|$ for every $i\in\{1,\dots,d\}$, we conclude that 
\begin{equation}\label{Eq: Segundo.bound.soma.projecoes}
    \sum_{i=1}^{d-1} |\calP_i(A\cap\R^k)| \leq \sum_{i=1}^{d-1} |(\R^k)_i| \leq (d-1)2^{d-2}|(\R^k)_j| \leq \frac{(d-1)2^{d-1}}{1 - \lambda}|\fext A \cap \R^k|.
\end{equation}
\end{itemize}

In both cases, we were able to give an upper bound to the sum of projections by a constant times the size of the boundary of $A$ in $\R^k$. Applying \eqref{Eq: Primeiro.bound.soma.projecoes} and \eqref{Eq: Segundo.bound.soma.projecoes} back in \eqref{Eq: Partition.sum.proj.} we get
\begin{align*}
    \sum_{i=1}^d|\calP_i(A\cap \R)| &\leq
    \sum_{k=1}^{r_d}\left[c\left(d-1, \frac{\lambda + 1}{2}\right)+ \frac{(d-1)2^{d-1}}{1 - \lambda}\right]|\fext A \cap \R \cap L_k| + |\calP_d(A\cap \R)|\\
    &=\left[c\left(d-1, \frac{\lambda + 1}{2}\right)+ \frac{(d-1)2^{d-1}}{1 - \lambda}\right]|\fext A \cap \R| + |\calP_d(A\cap \R)|.
\end{align*}
We finish the proof by noticing that we can repeat this same argument but now splitting $\R$ into layers $L_k = \{x\in \R : x_j = k\}$. Doing so, we have that  
\begin{equation*}
    \sum_{i=1}^d|\calP_i(A\cap \R)| \leq
     \left[c\left(d-1, \frac{\lambda + 1}{2}\right)+ \frac{(d-1)2^{d-1}}{1 - \lambda}\right]|\fext A \cap \R| + |\calP_j(A\cap \R)|
\end{equation*}
for any $j\in\{1,\dots, d\}$. Summing both sides in $j$ we conclude

\begin{equation}
    \sum_{i=1}^d|\calP_i(A\cap \R)| \leq \frac{d}{d-1}\left[c\left(d-1, \frac{\lambda + 1}{2}\right)+ \frac{(d-1)2^{d-1}}{1 - \lambda}\right]|\fext A \cap \R|,
\end{equation}
which proves our claim if we take $c(d,\lambda) \coloneqq \frac{d}{d-1}\left[c(d-1, \frac{\lambda + 1}{2})+ \frac{(d-1)2^{d-1}}{1 - \lambda}\right] = 2d + \frac{(d-2)d2^{d-1}}{1-\lambda}$.
\end{proof}

\begin{remark}
Previous lemma can be proved when $R\leq r_i \leq \kappa R$ for any $\kappa>1$. When applying the lemma, we will choose $\lambda=\frac{7}{8}$ to simplify the notation. In what will follow, all the proofs work as long as we choose $\lambda> \frac{3}{4}$.
\end{remark}

\begin{lemma}\label{Lemma: Proposicao1.Aux1}
    Given $A\subset \Z^d$, $\ell\geq 0$ and $U= C_{r\ell}\cup C_{r\ell}^\prime$ with $C_{r\ell}$ and $C_{r\ell}^\prime$ being two $r\ell$-cubes sharing a face, there exists a constant $b\coloneqq b(d)$ such that, if 
    
\begin{align}\label{Eq. U.condition}
    \frac{2^{r\ell d}}{2} \leq |C_{r\ell}\cap A| \qquad \text{and} \qquad |C_{r\ell}^\prime\cap A|< \frac{2^{r\ell d}}{2}
\end{align}
then $2^{r\ell(d-1)}\leq b|\fext A\cap U|$.
\end{lemma}

\begin{proof}
For $\ell=0$, \eqref{Eq. U.condition} guarantees that $C_{r\ell} = \{x\} \subset A$ and $C_{r\ell}^\prime = \{y\}\subset A^c$, hence $|\fext A\cap \{x,y\}| = 1$ and it is enough to take $b\geq 1$. For $\ell \geq 1$, \eqref{Eq. U.condition} yields
\begin{equation}\label{Eq. A.cap.U.volume}
    \frac{1}{2}2^{r\ell d} \leq |A\cap U| \leq \frac{3}{2}2^{r\ell d}.
\end{equation}

To simplify the notation, we can assume without loss of generality that ${U=[1,2^{r\ell}]^{d-1}\times [1, 2^{r\ell+1}]}$. As discussed before, for each point $p\in\calP^B_j(A\cap U)$ in the projection, $\ell_p^j\subset A\cap U$ and the lines are disjoint. Moreover, $|\calP_{j}^B(A\cap U)|r_j = \sum_{p\in\calP_{j}^B(A\cap U)} |\ell_p^j| \leq |A\cap U|$, since the size of the lines are constant $r_j\coloneqq |\ell_p^j|$. Together with the upper bound \eqref{Eq. A.cap.U.volume}, this yields 
\begin{equation}\label{Eq: Upper.bound.bad.points}
    |\calP_{j}^B(A\cap U)| \leq \frac{3}{2}2^{r\ell d}r_j^{-1}.
\end{equation}
Using the isometric inequality, the lower bound on \eqref{Eq. A.cap.U.volume} yields $d2^{\frac{1}{d}}2^{r\ell(d-1)}\leq |\fext (A\cap U)|$. As %
\begin{align*}
\frac{1}{2d}|\fext (A\cap U)| &\leq |\fint (A\cap U)| = |\fint(A\cap U) \cap \fint U| + |\fint(A\cap U) \cap(U\setminus \fint U)|\\
&\leq 2\sum_{i=1}^d |\calP_{i}(A\cap U)| + |\fint A\cap U| \leq 2\sum_{i=1}^d |\calP_{i}(A\cap U)| + |\fext A\cap U|,
\end{align*}
we get
\begin{equation}\label{Eq: Lemma.geo.discreta.3}
    2^{\frac{1}{d}-1}2^{r\ell(d-1)}\leq 2\sum_{i=1}^d |\calP_{i}(A\cap U)| + |\fext A\cap U|
\end{equation}

We again consider two cases:
\begin{itemize}
    \item[(a)] If $|\calP_{j}(A\cap U)|> \frac{7}{8}|U_j|$ for some $j=1,\dots, d$, by \eqref{Eq: upper.bound.good.points} and \eqref{Eq: Upper.bound.bad.points}  we get
    \begin{align*}
    \frac{7}{8}|U_j| < |\calP_{j}(A\cap U)| \leq |\fext A \cap U| + \frac{3}{2}2^{r\ell d}r_j^{-1}.
    \end{align*}
    A simple calculation shows that $\frac{1}{8}2^{r\ell(d-1)}\leq \frac{7}{8}|U_j| - \frac{3}{2}2^{r\ell d}r_j^{-1}$, therefore 
    \begin{equation}\label{Eq: upper.bound.big.projections.1}
        \frac{1}{8}2^{r\ell(d-1)} \leq |\fext A \cap U|.
    \end{equation}
    \item[(b)] If $|\calP_{i}(A\cap U)|\leq \frac{7}{8} |U_i|$ for all $i$, by Lemma \ref{Lemma: Geo.discreta.1}, there is a constant $c= c(d)$ such that
    \begin{equation}\label{Eq: Lemma.geo.discreta.2}
        \sum_{i=1}^d |\calP_{i}(A\cap U)|\leq c|\fext A\cap U|.
    \end{equation}
    Together with \eqref{Eq: Lemma.geo.discreta.3}, this yields
    \begin{equation}\label{Eq: upper.bound.big.projections.2}
        2^{r\ell(d-1)}\leq \frac{2c+1}{2^{\frac{1}{d}-1}}|\fext A\cap U|.
    \end{equation}
\end{itemize}

Equations \eqref{Eq: upper.bound.big.projections.1} and \eqref{Eq: upper.bound.big.projections.2} shows the desired results taking $b\coloneqq \max \{8, {(2c+1)}{2^{1-\frac{1}{d}}}\}$.
\end{proof}

\begin{proposition}\label{Proposition1}For the functions $(B_\ell)_{\ell\geq 0}$ defined in the beginning of Subsection \ref{Subsection: Coarse-graining}, there exists constants $b_1,b_2$ depending only on $d$ and $r$ such that 
\begin{equation}\label{Eq: Prop.1.FFS.i}
    |\fint\mathfrak{C}_\ell(\gamma)| \leq b_1\frac{|\fext \I_-(\gamma)|}{2^{r\ell(d-1)}} \leq b_1 \frac{|\gamma|}{2^{r\ell(d-1)}},
\end{equation}
    and 
\begin{equation}\label{Eq: Prop.1.FFS.ii}
    |B_\ell(\gamma)\Delta B_{\ell+1}(\gamma)| \leq b_2 2^{r\ell} |\gamma|,
\end{equation}
for every $\ell\geq 0$ and $\gamma\in\mathcal{C}_0(n)$.
\end{proposition}

    \begin{proof} Fix $\ell\geq 0$. To each cube $C_{r\ell}\in  \fint \mathfrak{C}_\ell(\gamma)$ there is an $r\ell$-cube $C_{r\ell}^\prime \not\in  \mathfrak{C}_\ell(\gamma)$ not admissible, sharing a face with  $C_{r\ell}$. We denote this relation by $C_{r\ell} \sim C_{r\ell}^\prime$.  Considering the collection of $r\ell$-cubes ${\mathscr{C}}_{r\ell}^\prime \coloneqq \{C_{r\ell} : C_{r\ell}\in \fint \mathfrak{C}_\ell(\gamma) \text{ or } C_{r\ell}\notin  \mathfrak{C}_\ell(\gamma)\}$ and $A\Subset\Z^d$, 
    
    \begin{align*}
        \sum_{\substack{C_{r\ell}\in \fint \mathfrak{C}_\ell(\gamma)}}\sum_{\substack{C_{r\ell}^\prime\notin  \mathfrak{C}_\ell(\gamma)\\ C_{r\ell} \sim C_{r\ell}^\prime}} |A \cap \left(C_{r\ell} \cup C_{r\ell}^\prime\right)| &\leq \sum_{\substack{C_{r\ell}\in \fint \mathfrak{C}_\ell(\gamma)}}\sum_{\substack{C_{r\ell}^\prime\notin  \mathfrak{C}_\ell(\gamma)\\ C_{r\ell} \sim C_{r\ell}^\prime}}  \left(|A \cap C_{r\ell}| + |A \cap C_{r\ell}^\prime|\right)\\
        &\leq \sum_{\substack{C_{r\ell}\in \fint \mathfrak{C}_\ell(\gamma)}} 2d |A \cap C_{r\ell}| + \sum_{C_{r\ell}^\prime\notin \mathfrak{C}_\ell(\gamma)} 2d|A \cap C_{r\ell}^\prime| \\
        &= 2d \sum_{C\in{\mathscr{C}}_{r\ell}^\prime} |A\cap C| = 2d|A\cap B_{{\mathscr{C}}_{r\ell}^\prime}| \leq 2d |A|
    \end{align*}
    
 As any pair of cubes $C_{r\ell}\sim C_{r\ell}^\prime$ are in the hypothesis of Lemma \ref{Lemma: Proposicao1.Aux1}, $ {b 2^{r\ell(d-1)} \leq |\fext \I_-(\gamma) \cap \left( C_{r\ell} \cup C_{r\ell}^\prime \right)|}$. Applying equation above for $A=\fext\I_-(\gamma)$ we get that
    \begin{equation*}
       \frac{b}{2d} 2^{r\ell(d-1)}| \fint \mathfrak{C}_\ell(\gamma)| \leq \frac{1}{2d}\sum_{\substack{C_{r\ell}\in \fint \mathfrak{C}_\ell(\gamma)}}\sum_{\substack{C_{r\ell}^\prime\notin  \mathfrak{C}_\ell(\gamma)\\ C_{r\ell} \sim C_{r\ell}^\prime}}  |\fext \I_-(\gamma) \cap \{C_{r\ell} \cup C_{r\ell}^\prime\}| \leq |\fext \I_-(\gamma)|,
    \end{equation*}
that concludes \eqref{Eq: Prop.1.FFS.i} for $b_1\coloneqq 2d/b$.

Given $C_{r(\ell+1)}\in \mathscr{C}_{r(\ell+1)}(B_{\ell+1}(\gamma)\setminus B_{\ell}(\gamma))$, there is a $r\ell$-cube $C_{r\ell}^\prime\subset C_{r(\ell + 1)}$ with $C_{r\ell}^\prime\notin \mathfrak{C}_\ell(\gamma)$, otherwise  $(B_{\ell+1}(\gamma)\setminus B_{\ell}(\gamma))\cap C_{r(\ell+1)} = \emptyset$. There is also a $r\ell$-cube $C_{r\ell}\subset C_{r(\ell + 1)}$ with $C_{r\ell}\in \mathfrak{C}_\ell(\gamma)$, otherwise we would have 
\begin{align*}
    |\I_-(\gamma)\cap C_{r(\ell+1)}| &= \sum_{C_{r\ell}\subset C_{r(\ell+1)}} |\I_-(\gamma)\cap C_{r\ell}| \leq \frac{1}{2} |C_{r(\ell+1)}|.
\end{align*}

Moreover, we can assume that $C_{r\ell}$ and $C_{r\ell}^\prime$ share a face. Again, we use Lemma \ref{Lemma: Proposicao1.Aux1} to get,
\begin{align}\label{Eq: bound.on.c.bar}
    |B_{\ell+1}(\gamma)\setminus B_\ell(\gamma)\cap C_{r(\ell+1)}| &\leq |C_{r(\ell+1)}| =2^{rd}2^{r\ell}2^{r\ell(d-1)} \nonumber\\
                &\leq 2^{rd}2^{r\ell} b|\fext \I_-(\gamma) \cap \{C_{r\ell} \cup C_{r\ell}^\prime\}| \nonumber\\
                &\leq 2^{rd}2^{r\ell}b|\fext \I_-(\gamma) \cap C_{r(\ell+1)}|.
\end{align}
Therefore, 
\begin{align*}
    |B_{\ell+1}(\gamma)\setminus B_\ell(\gamma)| &= \sum_{C_{r(\ell+1)}\in \mathscr{C}_{r(\ell+1)}(B_{\ell+1}(\gamma)\setminus B_\ell(\gamma))} |B_{\ell+1}(\gamma)\setminus B_\ell(\gamma)\cap C_{r(\ell+1)}| \\
    &\leq \sum_{C_{r(\ell+1)}\in \mathscr{C}_{r(\ell+1)}(B_{\ell+1}(\gamma)\setminus B_\ell(\gamma))} 2^{rd}2^{r\ell}b|\fext \I_-(\gamma) \cap C_{r(\ell+1)}| \leq  \frac{b_2}{2}2^{r\ell}|\fext \I_-(\gamma)|.
\end{align*}
with $b_2=b2^{rd+1}$. To get the same bound for $|B_{\ell}(\gamma)\setminus B_{\ell+1}(\gamma)|$ we repeat a similar argument, covering $B_{\ell}(\gamma)\setminus B_{\ell+1}(\gamma)$ with $r(\ell+1)$-cubes. 
\end{proof}

\begin{remark}\label{rmk: Upper_bound_on_ell}
    This proposition shows that when $\frac{b_1|\gamma|}{2^{r\ell(d-1)}}<1$ there are no admissible cubes. Therefore, in some propositions we assume $\ell\leq \frac{\log_{2^r}(b_1 |\gamma|)}{d-1}$, since the relevant bounds on the complementary case follow trivially.
\end{remark}

The next Corollary estimates the difference between the minus interior of a contour and its approximation, see Figure \ref{Fig: Figura7}.
\begin{corollary}\label{Cor: Bound_diam_B_ell}
     There exists a constant $b_3>0$ such that, for any $\ell>0$ and any two contours $\gamma_1,\gamma_2 \in \mathcal{C}_0(n)$ with $B_\ell(\gamma_1)=B_{\ell}(\gamma_2)$, 
    \begin{equation*}
        \d_2(\I_-(\gamma_1),\I_-(\gamma_2))\leq 4 \varepsilon b_3 2^{\frac{r\ell}{2}} n^{\frac{1}{2}}. 
    \end{equation*} 
\end{corollary}

\begin{proof}
    This is a simple application of the triangular inequality, since $\d_2(\I_-(\gamma_1),\I_-(\gamma_2)) \leq \d_2(\I_-(\gamma_1),B_\ell(\gamma_1)) + \d_2(\I_-(\gamma_2),B_\ell(\gamma_2))$ and 
    \begin{align*}
        \d_2(\I_-(\gamma_1),B_\ell(\gamma_1)) &\leq \sum_{i=1}^\ell \d_2(B_i(\gamma_1),B_{i-1}(\gamma_1)) = \sum_{i=1}^\ell 2\varepsilon\sqrt{|B_i(\gamma_1)\Delta B_{i-1}(\gamma_1)|} \\
        & \leq 2\varepsilon\sqrt{b_2}\sqrt{n} \sum_{i=1}^\ell  2^{\frac{ir}{2}}   \leq 4\varepsilon\sqrt{b_2}2^{\frac{r\ell}{2}} \sqrt{n} 
    \end{align*}
    where in the second to last equation used \eqref{Eq: Prop.1.FFS.ii}. As the same bound holds for $d_2(\I_-(\gamma_2),B_\ell(\gamma_2))$, the corollary is proved by taking $b_3 = 2\sqrt{b_2}$.
\end{proof}

\begin{remark}\label{Rmk: Bounding_N_by_B_ell}
    Corollary \ref{Cor: Bound_diam_B_ell} shows that we can create a cover of $\I_-(n)$, indexed by $B_\ell(\mathcal{C}_0(n))$, of balls with radius $4 \varepsilon b_3 2^{\frac{r\ell}{2}} n^{\frac{1}{2}}$. Therefore $N(\I_-(n), \d_2, 4\varepsilon b_3 2^{\frac{r\ell}{2}} n^{\frac{1}{2}}) \leq |B_\ell(\mathcal{C}_0(n))|$. 
\end{remark}
In the next section we bound $|B_\ell(\mathcal{C}_0(n))|$, using a method similar to the one used in \cite{Affonso.2021} to count $|\mathcal{C}_0(n)|$.

    \subsection{Entropy Bounds}

    As we discussed before, in the definition of admissibility at the beginning of Subsection \ref{Subsection: Coarse-graining}, $|B_\ell(\mathcal{C}_0(n))| = |\partial \mathfrak{C}_{r\ell}(\mathcal{C}_0(n))|$. In the short-range case, a key ingredient to count the admissible cubes is that, despite $B_\ell(\gamma)$ not being connected, all cubes are close to a connected region with size $|\gamma|$. As the contours now may not be connected, we need to change the strategy: we choose a suitable scale $L(\ell)$ and count how many $rL(\ell)$-coverings of contours there are. That is, we first control $|\C_{rL(\ell)}(\mathcal{C}_0(n))|$. Once a $rL(\ell)$-covering is fixed, we choose which $r\ell$-cubes inside this covering will be admissible. At last, we choose the scale $L(\ell)$ in a suitable way.

The first step is to bound $|\C_{rL}(\mathcal{C}_0(n))|$, for $L>0$. For $n,m\geq 0$, we say that $\mathscr{C}_n$ is \textit{subordinated} to $\C_m$, denoted by $\C_n\preceq \C_m$, if for all $C_n\in\C_n$, there exists $C_m\in \C_m$ such that $C_n\subset C_m$. Moreover, define 
\begin{equation*}
    N(\C_m, n, V) \coloneqq |\{\C_n : \C_n\preceq \C_m, |\C_n|=V\}|,
\end{equation*}
the number of collections of $n$-cubes $\C_n$ subordinated to a fixed collection $\C_m$ and with $|\C_n|=V$. Notice that every $m$-cube contains $2^d$ $(m-1)$-cubes, all of them being disjoint. Therefore, the number of $n$-cubes inside a $m$-cube is $2^{(m-n)d}$ and we have $N(\C_m, n, V) = \binom{2^{(m-n)d} |\C_{m}|}{V}$.
In particular, the bound on the binomial $\binom{n}{k}\leq \left(\frac{en}{k}\right)^k$ yields
\begin{equation}\label{Eq: Bound.on.N}
    N(\C_{r(\ell+1)}, r\ell, V) = \binom{2^{rd}|\C_{r(\ell+1)}|}{V} \leq \left(\frac{2^{rd}e|\C_{r(\ell+1)}|}{V}\right)^{V}.
\end{equation}
For any subset $\Lambda \Subset \Z^d$, define
\begin{equation*}
    V_r^\ell(\Lambda)\coloneqq \sum_{n=\ell}^{n_r(\Lambda)} |\C_{rn}(\Lambda)|,
\end{equation*}
where $n_r(\Lambda)\coloneqq \ceil{\log_{2^r}(\diam (\Lambda))}$. To control $V_r^\ell(\Lambda)$ we bound the number of coverings at a fixed step $L>0$.

\begin{proposition}\label{Prop. partition.a.graph}
Let $k\geq 1$ and $G$ be a finite, non-empty, connected simple graph with vertex set $v(G)$. Then, $G$ can be covered by $\ceil*{|v(G)|/k}$ connected sub-graphs of size at most $2k$.
\end{proposition}
We omit the proof since it is the same as in \cite{Affonso.2021}. Remember that, given ${G = (V,E) \in \mathscr{G}_n(\Lambda)}$, $\Lambda^G \coloneqq \Lambda \cap B_V$ denotes the area of $\Lambda$ covered by $G$. Remember also that, for $A\Subset \Z^d$ and $j\geq 1$, $\Gamma^r_j(A)$ are the partition elements removed at step $j$, in the construction presented in Section 2.  Using this construction we can prove the following lemma.

\begin{lemma}\label{Lemma: Big.clusters_2}
    Let $A\Subset \Z^d$, $\gamma\in\Gamma^r(A)$ and $j \geq 1$ be such that $\gamma\in \Gamma^r_j(A)$. Then, for any $\ell < j$ and $G_{r\ell}\in \mathscr{G}_{r\ell}(\gamma)$,
    \begin{equation}\label{Eq: Lower_bound_on_the_covering_of_gamma_G}
        2^{r(1-\frac{1}{d})\ell} \leq |\C_{r\ell}(\gamma^{G_{r\ell}})| 
    \end{equation}
\end{lemma}
\begin{proof}
        Given $G_{r\ell}\in \mathscr{G}_{r\ell}(\gamma)$, by our construction of the contour, $2^{r(d+1)\ell} < |V(\gamma^{G_{r\ell}})|$. A trivial bound gives us $|V(\gamma^{G_{r\ell}})| \leq 2^{r\ell d}|\C_{r\ell}(V(\gamma^{G_{r\ell}}))|$. Associating each cube $C_m(x)$ to $x$, we get a one-to-one correspondence between $m$-cubes and lattice points that preserves neighbors, that is, two m-cubes $C_m(x)$ and $C_m(y)$ share a face if and only if $|x-y|=1$. We can therefore apply the isoperimetric inequality to get $|\C_{r\ell}(V(\gamma^{G_{r\ell}}))| \leq |\fint \C_{r\ell}(V(\gamma^{G_{r\ell}}))|^{\frac{d}{d-1}}\leq |\C_{r\ell}(\gamma^{G_{r\ell}})|^{\frac{d}{d-1}}$, where in the last equation we are using that every cube in the boundary of cubes must cover at least one point of $\gamma^{G_{r\ell}}$. We conclude that $2^{r(d+1)\ell} \leq 2^{r\ell d}|\C_{r\ell}(\gamma^{G_{r\ell}})|^{\frac{d}{d-1}}$, and \eqref{Eq: Lower_bound_on_the_covering_of_gamma_G} follows.
\end{proof}

As a corollary, we can recuperate a key lemma of \cite{Affonso.2021}, which is the following.
\begin{lemma}\label{Lemma: Big.clusters_1}
    Given $A\Subset \Z^d$, $n> 1$ and $\gamma\in\Gamma^r(A)$, if $|\mathscr{G}_{rn}(\gamma)|\geq 2$ then $|v(G_{rn}(\gamma))| \geq 2^r$ for every $G_{rn}(\gamma)\in \mathscr{G}_{rn}(\gamma)$ 
\end{lemma}

   The next proposition bounds the partial volume.
\begin{proposition}\label{Prop. Bound.on.V_r^l(gamma)}
    There exists a constant $b_3 \coloneqq b_3(d, M, r)$ such that, for any $A\Subset \Z^d$, $\gamma\in\Gamma^r(A)$ and $\ell \geq 0$,
    
     \begin{equation*}
        V_r^\ell(\gamma)\leq b_3 (\ell\vee 1)|\mathscr{C}_{r\ell}(\gamma)|. 
    \end{equation*}

\end{proposition}

\begin{proof}
Start by noticing that $\gamma \in \Gamma^r(A)$ implies that $\Gamma^r(\gamma) = \{\gamma\}$. Let's assume first that $\ell\geq 2$. Define $g : \mathbb{N} \xrightarrow{} \Z$ by
\begin{equation}
    g(n)\coloneqq \floor*{\frac{n - 2 - \log_{2^r}(2M)}{a}}.
\end{equation}
It was proved in \cite[Proposition 3.13]{Affonso.2021} that 
\begin{equation}\label{Eq: Bound_c_n_by_C_g(n)}
    |\C_{rn}(\gamma)| \leq \frac{1}{2^{r-d-1}}|\C_{rg(n)}(\gamma)|,
\end{equation}
whenever $g(n)>0$, and every connected component of $G_{rg(n)}(\gamma)$ has more than $2^r -1$ vertices. This is equivalent,  by Lemma \ref{Lemma: Big.clusters_1}, to $|\mathscr{G}_{rg(n)}(\gamma)|\geq 2$ or $|\mathscr{G}_{rg(n)}(\gamma)|=1$ with $|v(G_{rg(n)}(\gamma))| \geq 2^r$. Consider then the auxiliary quantities
\begin{align*}
    &l_1(n)\coloneqq\max\{m : g^m(n)\geq \ell\} &\text{and} &&l_2(n)\coloneqq\max\{ m : |\mathscr{G}_{rg^m(n)}(\gamma)| = 1 \text{ and } |v(G_{rg^m(n)})|\leq 2^r-1\}.
\end{align*}

We first show that $l_2(n)$ is not zero for only a constant number of scales $n$. For any $m\leq l_2(n)$, as $\Sp(\gamma)\subset B_{\C_{rg^m(n)}(\gamma)}$, $\diam(\gamma)\leq \diam(B_{\C_{rg^m(n)}(\gamma)})$. For any $\Lambda,\Lambda^\prime\Subset \Z^d$, 
    \begin{equation*}
        \diam(\Lambda\cup \Lambda) \leq \diam(\Lambda) + \diam(\Lambda^\prime) + \dis(\Lambda,\Lambda^\prime),
    \end{equation*}
    and we can always extract a vertex from a connected graph in a way that the induced sub-graph is still connected, by removing a leaf of a spanning tree. Using this we can bound 
\begin{align}\label{Eq: Bound_diam_removing_trees}
    \diam(\gamma)\leq \diam(B_{\C_{rg^m(n)(\gamma)}}) & \leq \sum_{C_{rg^m(n)}\in v(G_{rg^m(n)})} \diam(C_{rg^{m}(n)}) + |v(G_i)|M2^{arg^m(n)} \nonumber\\
    &\leq(d2^{rg^m(n)} + Md^a2^{arg^m(n)})|\C_{rg^m(n)(\gamma)}|\leq 2Md^a2^{arg^m(n)+r},
\end{align}
since the graph $G_{rg^m(n)}\in \mathscr{G}_{rg^m(n)}(\gamma)$ has $v(G_{rg^m(n)}) = \C_{rg^m(n)}(\gamma)$, and $|\C_{rg^m(n)(\gamma)}|\leq 2^{r}-1$. Applying the logarithm with respect to base $2^{r}$ we get
\begin{equation*}
    \log_{2^r}(\diam(\gamma)) \leq \log_{2^r}(2Md^a) + ag^m(n)+1 \leq \log_{2^r}(2Md^a) + \frac{n}{a^{m-1}} + 1
\end{equation*}
Assuming $\diam(\gamma)>2^{2r + 1}Md^a$, we can isolate the term depending on $m$ in the equation above and take the logarithm on both sides to get
\begin{equation*}
    m \leq 1 + \frac{\log_2(n) - \log_2(\log_{2^r}(\diam(\gamma)) - \log_{2^r}(2Md^a) - 1)}{\log_2(a)}.
\end{equation*}
Equation above holds for any element of $\{m : |\mathscr{G}_{rg^m(n)}(A)| = 1, |v(G_{rg^m(n)})|\leq 2^r-1\}$ thus it also holds for $l_2(n)$. This shows in particular that $l_2(n)=0$ for $n<\log_{2^r}(\diam(\gamma)) - \log_{2^r}(2Md^a) - 1$. Taking $N_0 = n_r(\gamma) - \log_{2^r}(2Md^a) - 2$, as $N_0\leq \log_{2^r}(\diam(\gamma)) - \log_{2^r}(2Md^a) - 1$ we can bound 
\begin{equation}\label{Eq: Bound_last_terms_of_V_r_l}
    \sum_{n=N_0}^{n_r(\gamma)}|\C_{rn}(\gamma)| \leq (\log_{2^r}(2Md^a)+2)|\C_{r\ell}(\gamma)|.
\end{equation}

We consider now $n<N_0$. Knowing that $l_2(n)=0$ and $|\C_{k}(\gamma)|\leq |\C_j(\gamma)|$, for all $j\leq k$, we get 
\begin{equation}\label{Eq: bound.on.rn.covering}
    |\C_{rn}(\gamma)|\leq \frac{1}{2^{(r-d-1)l_1(n)}}|\C_{r\ell}(\gamma)|.
\end{equation}
We claim that
\begin{equation}\label{Eq: lower.bound.on.l1}
    l_1(n) \geq \begin{cases}
                        0, &\text{ if }n\leq \overline{b}+\ell\\ 
                        \left\lfloor\frac{\log_2(n) - \log_2(\overline{b} + \ell)}{\log_2(a)}\right\rfloor, & \text{ if }n > \overline{b} + \ell, 
                \end{cases}
\end{equation}
where $\overline{b} = (a+2 + \log_{2^r}(2M))(a-1)^{-1}$. Given $n > \overline{b} + \ell$, consider
\begin{equation*}
    \Tilde{g}(n) = \frac{n - 2 - \log_{2^r}(2M)}{a} - 1.
\end{equation*}
It is clear that $g(n)\geq \Tilde{g}(n)$ and both functions are increasing, therefore $g^m(n)\geq \Tilde{g}^m(n)$ for every $m\geq 0$. As
\begin{equation*}
    \Tilde{g}^m(n) = \frac{n}{a^m} - b^\prime\frac{a^m - 1}{a^{m-1}(a-1)},
\end{equation*}
with $b^\prime = (a+2 + \log_{2^r}(2M))a^{-1}$, it is sufficient to have
\begin{equation*}
    \frac{n}{a^m} -\frac{a b^\prime}{(a-1)}\geq \ell.
\end{equation*}
We get the desired bound by applying the logarithm with base two in the equation above. The bounds \eqref{Eq: bound.on.rn.covering} and \eqref{Eq: lower.bound.on.l1} yields

\begin{align*}
    V_r^\ell(\gamma) &\leq 
     \overline{b}|\C_{r\ell}(\gamma)| + |\C_{r\ell}(\gamma)|2^{r-d-1}(\overline{b}+\ell)^{\frac{r-d-1}{\log_2(a)}}\sum_{n=\overline{b} + \ell}\frac{1}{n^{\frac{r-d-1}{\log_2(a)}}} + (\log_{2^r}(2Md^a)+2)|\C_{r\ell}(\gamma)|\\
    &\leq (\overline{b} + \log_{2^r}(2Md^a) +2)|\C_{r\ell}(\gamma)| +|\C_{r\ell}(\gamma)| 2^{r-d-1}(\overline{b}+1)^{\frac{r-d-1}{\log_2(a)}}\ell^{\frac{r-d-1}{\log_2(a)}}\sum_{n=\ell + 1}^{\infty}\frac{1}{n^{\frac{r-d-1}{\log_2(a)}}}\\
    &\leq \left(\overline{b} + \log_{2^r}(2Md^a) +2+ 2^{r-d-1}(\overline{b}+1)^{\frac{r-d-1}{\log_2(a)}}\frac{\log_2(a)}{r-d-1+\log_2(a)}\right)\ell|\C_{r\ell}(\gamma)|,
\end{align*}
where in the last inequality we used the integral bound $$\sum_{n=\ell + 1}^{\infty}{n^{-\frac{r-d-1}{\log_2(a)}}}\leq \int_{\ell}^\infty {x^{-\frac{r-d-1}{\log_2(a)}}} dx = \frac{\log_2(a)}{r-d-1+\log_2(a)}\ell^{1 - \frac{r-d-1}{\log_2(a)}}.$$
If $\diam(\gamma)\leq 2^{2r + 1}M$, we have
\begin{equation*}
     V_r^\ell(\gamma) \leq (n_r(\gamma) - \ell+1)|\C_{r\ell}(\gamma)| \leq (3 + \log_{2^r}(2M))|\C_{r\ell}(\gamma)|.
\end{equation*}
Taking $b_3^\prime\coloneqq \max\{{2^{r-d+2}(2 + \frac{a}{d -1})(\overline{b} + \log_{2^r}(2Md^a) +3)^{\frac{r-d-1}{\log_2(a)}}}, 3 + \log_{2^r}(2M)\}$ we get the desired bound when $\ell\geq 2$. For $\ell=0$, a trivial bound yields  $V_r^0(\gamma) = 2|\gamma| + V_r^2(\gamma)\leq (2 + b_3^\prime 2)|\gamma|$. Similarly, for $\ell =1$,  $V_r^1(\gamma) = |\C_{r}(\gamma)| + V_r^2(\gamma)\leq (1+ b_3^\prime 2)|\mathscr{C}_{r}(\gamma)|$ and we conclude the proof by taking $b_3 \coloneqq 2(b_3^\prime +1)$.
\end{proof}

We then need to bound the minimal number of $r\ell$-cubes necessary to cover a contour. Using only Lemma \ref{Lemma: Big.clusters_1}, it is possible to prove the next proposition, in the same steps as in \cite[Proposition 3.13]{Affonso.2021}.

\begin{proposition}\label{Prop. Bound.on.C_rl(gamma)_Lucas}
    There exists a constant $b_4^{\prime\prime}\coloneqq b_4^{\prime\prime}(\alpha, d)$ such that for any $A\Subset \Z^d$, $\gamma\in\Gamma(A)$ and $1 \leq \ell\leq n_r(A)$, 
     \begin{equation*}
       |\C_{r\ell}(\gamma)|\leq b_4^{\prime\prime}\frac{|\gamma|}{\ell^{\frac{r-d-1}{\log_2(a)}}}.
    \end{equation*}
\end{proposition}

Next, we improve this upper bound using our construction. This is the most relevant property of the new contours. 
\begin{proposition}\label{Prop. Bound.on.C_rl(gamma)}
    There exists constants $b_4\coloneqq b_4(\alpha, d)$ and $b_4^\prime\coloneqq b_4^\prime(\alpha, d)$  such that for any $A\Subset \Z^d$, $\gamma\in\Gamma^r_j(A)$ and $0 \leq \ell<j$,
    
     \begin{equation}\label{Eq: Bound.on.C_rl(gamma)_small_l}
       |\C_{r\ell}(\gamma)|\leq b_4\frac{(\ell \vee 1)^{\kappa}}{2^{ra^\prime\ell}}|\gamma|,
    \end{equation}
    with $a^\prime \coloneqq \frac{(1-\frac{1}{d})}{a -\frac{1}{d}}$ and $\kappa \coloneqq \frac{d+1 + r(1-\frac{1}{d}) (a+2 - d^{-1}+ \log_{2^r}(2M))(a-d^{-1})^{-1}}{\log_2(a + 1 -d^{-1})} $.
    Moreover, for $\ell\geq j$
    \begin{equation}\label{Eq: Bound.on.C_rl(gamma)_large_l}
        |\C_{r\ell}(\gamma)|\leq b_4^\prime \ell^{\kappa} \left(\frac{|\gamma|}{2^{r\frac{a^\prime}{a} \ell}} \vee 1\right).
    \end{equation}
\end{proposition}

\begin{proof}
     Lets first consider $\ell < j$. Define $f : \mathbb{N} \xrightarrow{} \Z$ by
\begin{equation}
    f(\ell)\coloneqq \floor*{\frac{\ell - \log_{2^r}(2M) - 1}{a + (1-\frac{1}{d})}}.
\end{equation}
Following the proof of \eqref{Eq: Bound_c_n_by_C_g(n)} in \cite[Proposition 3.13]{Affonso.2021}, we can show that 
\begin{equation}\label{Eq: Bound_C_l_by_C_f(l)}
    |\C_{r\ell}(\gamma)| \leq \frac{2^{d+1}}{2^{r(1-\frac{1}{d})f(\ell)}}|\C_{rf(\ell)}(\gamma)|.
\end{equation}
 By definition, $\mathscr{G}_{rf(\ell)}(\gamma)$ is the set of all connected components of $G_{rf(\ell)}(\gamma)$, hence
    \begin{equation}\label{Eq: 3.16.Lucas}
        |\C_{rf(\ell)}(\gamma)| = 2^{r(1-\frac{1}{d})f(\ell)}\sum_{G\in \mathscr{G}_{rf(\ell)}(\gamma)}\frac{|v(G)|}{2^{r(1-\frac{1}{d})f(\ell)}}.
    \end{equation}
    Proposition \ref{Prop. partition.a.graph} guarantees that we can split $G$ into sub-graphs $G_i$, with $1\leq i\leq \ceil{v(G)/2^{r(1-\frac{1}{d})f(\ell)}}$ and $|v(G_i)|\leq 2^{r(1-\frac{1}{d})f(\ell)+1}$. Proceeding as in \eqref{Eq: Bound_diam_removing_trees}, we can bound 
    \begin{align*}
        \diam(B_{v(G_i)}) &\leq \sum_{C_{rf(\ell)}\in v(G_i)} \diam(C_{rf(\ell)}) + |v(G_i)|M2^{arf(\ell)}\\
        &\leq |v(G_i)|(d2^{rf(\ell)} + M2^{arf(\ell)}) \leq 2M2^{r[f(\ell)(1-\frac{1}{d}) + a] + 1}\\
        &\leq 2^{r\ell}.
    \end{align*}
    
    The last inequality holds since $M,a,r\geq 1$. This shows that every $G_i$ can be covered by a cube with center in $\Z^d$ and side length $2^{r\ell}$. Every such cube can be covered by at most $2^d$ $r\ell$-cubes. Indeed, it is enough to consider the simpler case when the cube is of the form
    \begin{equation}\label{Cube.Q}
        \prod_{i=1}^d[q_i , q_i + 2^{r\ell})\cap\Z^d,
    \end{equation}
    with $q_i\in\{0, 1, \dots, 2^{r\ell}-1\}$, for $1\leq i \leq d$. It is easy to see that \begin{equation*}
        [q_i , q_i + 2^{r\ell }]\subset [0, 2^{r\ell})\cup [2^{r\ell}, 2^{r\ell +1}). 
    \end{equation*} 
    Taking the products for all $1\leq i\leq d$, we get $2^d$ $r\ell$-cubes that covers \eqref{Cube.Q}. 
    We conclude that, to cover a connected component $G\in \mathscr{G}_{rf(\ell)}$, we need at most $2^d\ceil{|v(G)|/2^{r(1-\frac{1}{d})f(\ell)}}$ $rf(\ell)$-cubes, yielding us
    \begin{equation}\label{Eq: 3.18.Lucas}
        |\C_{r\ell}(\gamma)|\leq |\C_{r\ell}( B_{\C_{rf(\ell)}(\gamma)})| \leq \sum_{G\in \mathscr{G}_{rf(\ell)}} |\C_{r\ell}(v(G))| \leq \sum_{G\in \mathscr{G}_{rf(\ell)}} 2^d\left\lceil{\frac{|v(G)|}{2^{r(1-\frac{1}{d})f(\ell)}}}\right\rceil. 
    \end{equation}
    When every connected component of $G_{rf(\ell)}(\gamma)$ has more than $2^{r(1-\frac{1}{d})f(\ell)}$ vertices, we can bound 
    \begin{equation*}
        \frac{1}{2}\left\lceil{\frac{|v(G)|}{2^{r(1-\frac{1}{d})f(\ell)}}}\right\rceil \leq {\frac{|v(G)|}{2^{r(1-\frac{1}{d})f(\ell)}}}.
    \end{equation*}
    Together with Inequalities \eqref{Eq: 3.16.Lucas}  and \eqref{Eq: 3.18.Lucas}, this yields
    \begin{equation}\label{Eq: Bound_C_rl_by_gamma_with_f(l)}
         |\C_{r\ell}(\gamma)| \leq \sum_{G\in \mathscr{G}_{rf(\ell)}} 2^{d+1} {\frac{|v(G)|}{2^{r(1-\frac{1}{d})f(\ell)}}}= \frac{2^{d+1}}{2^{r(1-\frac{1}{d})f(\ell)}}|\C_{rf(\ell)}(\gamma)|.
    \end{equation}

    Equation \eqref{Eq: Bound_C_l_by_C_f(l)} can be iterated as long as $f(\ell)$ is positive. Considering then the auxiliary quantity
    \begin{equation*}
        m(\ell)\coloneqq\max\{m : f^m(\ell)\geq 0\},  
    \end{equation*}
we have 
\begin{equation}\label{Eq: Bound_C_rl_by_gamma_with_f(l)_2}
    |\C_{r\ell}(\gamma)| \leq \frac{2^{(d+1)m(\ell)}}{2^{r\left(1-\frac{1}{d}\right)\left(\sum_{i=1}^{m(\ell)}f^i(\ell)\right)}}|\gamma|,
\end{equation}
so we need upper and lower estimates for $m(\ell)$.
We claim that
\begin{equation}\label{Eq: lower.bound.on.m}
    m(\ell) \geq \begin{cases}
                        0, &\text{ if }\ell\leq \overline{b}\\ 
                        \left\lfloor\frac{\log_2(\ell) - \log_2(\overline{b})}{\log_2(a + (1-\frac{1}{d}))}\right\rfloor, & \text{ if }\ell > \overline{b}, 
                \end{cases}
\end{equation}
where $\overline{b} = (\overline{a}+1 + \log_{2^r}(2M))(\overline{a}-1)^{-1}$ and $\overline{a}\coloneqq a + (1-\frac{1}{d})$. Given $\ell > \overline{b}$, consider
\begin{equation*}
     \overline{f}(\ell) = \frac{\ell - 1 - \log_{2^r}(2M)}{a + (1 - \frac{1}{d})} - 1.
\end{equation*}
It is clear that $f(\ell)\geq \overline{f}(\ell)$ and both functions are increasing, therefore $f^m(\ell)\geq \overline{f}^m(\ell)$ for every $m\geq 0$.  As
\begin{equation*}
       \overline{f}^m(\ell) = \frac{\ell}{\overline{a}^m} - b^\prime\frac{\overline{a}^m - 1}{\overline{a}^{m-1}(\overline{a}-1)},
\end{equation*}
with $b^\prime = (\overline{a}+1 + \log_{2^r}(2M))\overline{a}^{-1}$, it is sufficient to have
\begin{equation*}
    \frac{\ell}{\overline{a}^m} -\frac{\overline{a} b^\prime}{(\overline{a}-1)}\geq 0.
\end{equation*}
We get the desired bound by applying the logarithm with base two in the equation above. Moreover, we can bound
\begin{align*}
    \sum_{i=1}^{m(\ell)}f^{i}(\ell) & \geq  \sum_{i=1}^{m(\ell)}\frac{\ell}{\overline{a}^i} - m(\ell)\frac{\overline{a}b^\prime}{\overline{a} - 1} = \frac{1}{\overline{a}}(\frac{1-\frac{1}{\overline{a}^{m(\ell)}}}{{1-\frac{1}{\overline{a}}}})\ell - m(\ell)\overline{b} \\
    &\geq \frac{1}{\overline{a}-1}(1-\frac{1}{\overline{a}^{m(\ell)}})\ell - m(\ell)\overline{b} \geq \frac{1}{\overline{a}-1}(\ell-{\overline{a}\overline{b}}) - m(\ell)\overline{b}
\end{align*}

For the upper bound on $m(\ell)$, take $\Tilde{f}(\ell) \coloneqq \frac{\ell}{a + (1-\frac{1}{d})}$. As $f(\ell)\leq \Tilde{f}(\ell)$ and $\Tilde{f}$ is increasing, for every $m\geq 0$,  $f^m(\ell)\leq \Tilde{f}^m(\ell)$. Notice that, if $\Tilde{f}^m(\ell)\leq 1$, $f^{m+1}(\ell)<0$, and therefore $m+1>m(\ell)$. As $\Tilde{f}^m(\ell)\leq 1$ if and only if $\ell \leq [a + (1-\frac{1}{d})]^m$, taking $m=\left\lceil\frac{\log_2(\ell)}{\log_2\left(a + (1-\frac{1}{d})\right)}\right\rceil$ we get $\left\lceil\frac{\log_2(\ell)}{\log_2\left(a + (1-\frac{1}{d})\right)}\right\rceil + 1 > m(\ell)$.
Applying this bound on \eqref{Eq: Bound_C_rl_by_gamma_with_f(l)_2} we conclude that 
\begin{equation}\label{Eq: bound_on_C_ell_covering_l_geq_ab}
     |\C_{r\ell}(\gamma)| \leq \frac{2^{d+1 + r(1-\frac{1}{d})\left(\frac{\overline{a}}{\overline{a}-1} + 1\right)\overline{b}}\ell^{\frac{d+1 + r(1-\frac{1}{d})\overline{b}}{\log_2(\overline{a})}}}{2^{r(1-\frac{1}{d})\frac{1}{\overline{a}-1}\ell}}|\gamma|,
\end{equation}
for $\ell>\overline{b}$. When $\ell\leq\overline{b}$, we can take $\overline{b}_4\coloneqq \min\{{ (j\vee 1)^{\frac{d+1 + r(1-\frac{1}{d})\overline{b}}{\log_2(\overline{a})}}{2^{-r(1-\frac{1}{d})\frac{1}{\overline{a}-1}j}}} : 0\leq j \leq \overline{b}\}$ and then 
\begin{equation*}
       |\C_{r\ell}(\gamma)| \leq |\gamma|\leq \frac{1}{\overline{b}_4}\frac{(\ell\vee 1)^{\frac{d+1 + r(1-\frac{1}{d})\overline{b}}{\log_2(\overline{a})}}}{2^{r(1-\frac{1}{d})\frac{1}{\overline{a}-1}\ell}}|\gamma|.
\end{equation*}
This, together with equation \eqref{Eq: bound_on_C_ell_covering_l_geq_ab}, yields inequality \eqref{Eq: Bound.on.C_rl(gamma)_small_l} with $b_4 \coloneqq \max\{ 2^{d+1 + r(1-\frac{1}{d})\left(\frac{\overline{a}}{\overline{a}-1} + 1\right)\overline{b}}, \overline{b}_4^{-1}\}$.

To prove inequality \eqref{Eq: Bound.on.C_rl(gamma)_large_l}, we first notice that for any $\ell\geq j$,
\begin{align}\label{Eq:  Bound.on.C_rl(gamma)_large_l_aux_1}
    |\C_{r\ell}(\gamma)|\leq |\C_{r(j-1)}(\gamma)| \leq  b_42^{ra^\prime}\frac{j^\kappa}{2^{ra^\prime j}}|\gamma|.
\end{align}
When $\ell \leq aj$, this already gives us \eqref{Eq:  Bound.on.C_rl(gamma)_large_l}. For $\ell > aj$, we can give a better bound once we notice that, by the construction of the contour, the graph $G_{rj}(\gamma)$ is connected and its vertices are the covering $\C_{rj}(\gamma)$. In a similar fashion as done previously, by Proposition \ref{Prop. partition.a.graph} we can split $G_{rj}(\gamma)$ into $\ceil{|v(G_{rj}(\gamma))|/k}$ connected sub-graphs $G_1,\dots, G_k$, with $k\coloneqq { \left\lceil \frac{2^{r\ell}}{2^{raj}} \right\rceil}$ and $|v(G_i)|\leq 2 \left\lceil \frac{2^{r\ell}}{2^{raj}} \right\rceil$ for all $i=1,\dots, k$. Assuming $\ell > aj$, we have $|v(G_i)|\leq 2^{r(\ell - aj) + 2}$. As
  \begin{align*}
        \diam(B_{v(G_i)}) &\leq |v(G_i)|(d2^{rj} + M2^{arj}) \leq  2Md2^{raj}|v(G_i)|\\
        &\leq 8Md2^{r\ell},
    \end{align*}
$B_{v(G_i)}$ can be covered by $(8Md)^d$ cubes centered in $\Z^d$ with side length $2^{r\ell}$. As we seen before, every such cube can be covered by at most $2^d$ $r\ell$-cubes, therefore $|\C_{r\ell}(B_{v(G_i)})| \leq (16Md)^d$ and we conclude that 
\begin{align}\label{Eq:  Bound.on.C_rl(gamma)_large_l_aux_2}
    |\C_{r\ell}(\gamma)| \leq \sum_{i=1}^{\ceil{|v(G_{rj}(\gamma))|/k}} |\C_{r\ell}(B_{v(G_i)})| &\leq (16Md)^d \left\lceil \frac{|\C_{rj}(\gamma)|}{\frac{2^{r\ell}}{2^{raj}}} \right\rceil \nonumber \\ 
    &\leq 2(16Md)^d b_42^{ra^\prime} j^\kappa\left(\frac{2^{r(aj - \ell)}}{2^{ra^\prime j}}|\gamma| \vee 1\right).
\end{align}
As we are assuming $\ell> aj$, $\frac{2^{r(aj - \ell)}}{2^{ra^\prime j}} \leq 2^{-r\frac{a^\prime}{a}\ell}$, taking $b_4^\prime \coloneqq 2(16Md)^d b_42^{ra^\prime}$ we conclude the proof.
\end{proof}

For any non-negative $V, M, a, r$, define
\begin{equation*}
    \mathcal{F}_{V}^\ell\coloneqq \{ \C_{r\ell} : V_r^\ell(B_{\C_{r\ell}}) = V, B_{\C_{r\ell}}\subset [-\diam(B_{\C_{r\ell}}), \diam(B_{\C_{r\ell}})]^d\}.
\end{equation*}

Using equation \eqref{Eq: Bound.on.N}, in the same steps as \cite[Proposition 3.11]{Affonso.2021}, we can show that the number of collections in $\mathcal{F}_V$ is exponentially bounded by $V$.

\begin{proposition}\label{Prop. Bound.on.Fv}
    There exists $b_5\coloneqq b_5(d,r)$ such that
    \begin{equation}\label{Eq: Bound.on.F_V}
        |\mathcal{F}^\ell_V| \leq e^{b_5V}.
    \end{equation}
\end{proposition}

\begin{proof}
We start by splitting $\mathcal{F}^\ell_V$ into $\mathcal{F}^\ell_{V,m} \coloneqq \{ \C_{r\ell}\in \mathcal{F}^\ell_V : n_r(B_{\C_{r\ell}})=m \}$. Since $\ell\leq n_r(B_{\C_{r\ell}}) \leq V_r^\ell(B_{\C_{r\ell}}) +\ell$, we get

\begin{equation}
    |\mathcal{F}^\ell_V| \leq \sum_{m=\ell}^{V+\ell} |\mathcal{F}^\ell_{V,m}|. 
\end{equation}
Denoting $(V_{rn})_{n=\ell}^{m}$ an arbitrary family of natural numbers satisfying 
\begin{equation}\label{Eq: sum.of.V_rn}
    \sum_{n=\ell}^{m} V_{rn}\leq  V,
\end{equation}
with $V_{rn}\leq V_{r(n-1)}$, we can bound
\begin{align}\label{Eq: bound.FVL}
    |\mathcal{F}^\ell_{V,m}| \leq \sum_{(V_{rn})_{n=\ell}^{m}} |\{\C_{r\ell} : B_{\C_{r\ell}} \subset [-2^{rm},2^{rm}]^d, |\C_{rn}(B_{\C_{r\ell}})| = V_{rn}, \text{ for every } \ell \leq n\leq m, n_r(B_{\C_{r\ell}}) = m\}|.
\end{align}
 As  $[-2^{rm},2^{rm}]^d$ is a cube centered in $\Z^d$ with side length $2^{rm+1}$, it can be covered by $3^d$ ${(rm+1)}$-cubes, as we showed in Proposition \ref{Prop. Bound.on.C_rl(gamma)}. So, denoting $\C^0_{rm+1} \coloneqq \C_{rm+1}([-2^{rm},2^{rm}]^d)$, we have $|\C^0_{rm+1}| \leq 3^d$.
 
We can give an upper bound to the right-hand side of equation \eqref{Eq: bound.FVL} by counting the number of families $(\C_{rn})_{n=\ell}^m$ such that $\C_{rn}\preceq \C_{r(n+1)}$, for $n<m$, and $\C_{rm}\preceq \C^0_{rm+1}$, yielding us 
\begin{align*}
    |\mathcal{F}^\ell_{V,m}| 
    &\leq \sum_{(V_{rn})_{n=\ell}^{m-1}} |\{ (\C_{rn})_{n=\ell}^{m} : |\C_{rn}|=V_{rn}, \C_{rn}\preceq \C_{r(n+1)},\C_{rm}\preceq \C^0_{rm+1}\}| \\
    & \leq  \sum_{(V_{rn})_{n=\ell}^{m-1}} \sum_{\substack{\C_{rm}\preceq \C^0_{rm+1}\\ |\C_{rm}|=V_{rm}}}\sum_{\substack{\C_{r(m-1)} \\ |\C_{r(m-1)}|=V_{r(m-1)}\\ \C_{r(m-1)\preceq \C_{rm}}}} \cdots \sum_{\substack{\C_{r(\ell+1)} \\ |\C_{r(\ell+1)}|=V_{r(\ell+1)}\\ \C_{r(\ell+1)\preceq \C_{r(\ell+2)}}}} N(\C_{r(\ell+1)}, r\ell, V_{r\ell}).           
\end{align*}
Iterating equation \eqref{Eq: Bound.on.N} we get that
\begin{align*}
    |\mathcal{F}^\ell_{V,m}| &\leq \sum_{(V_{rn})_{n=\ell}^{m}}\left( \frac{2^{d}e |\C^0_{rm+1}|}{V_{rm}}\right)^{V_{rm}}\prod_{n=\ell}^{m-1}\left( \frac{2^{rd}e V_{r(n+1)}}{V_{rn}}\right)^{V_{rn}}\\
    &\leq  \sum_{(V_{rn})_{n=\ell}^{m-1}}\left( \frac{2^{d}e3^d}{V_{rm}}\right)^{V_{rm}}\prod_{n=\ell}^{m-1}e^{(rd\log(2) +1)V_{rn}} \\
    &\leq \sum_{(V_{rn})_{n=\ell}^{m-1}}e^{(d\log(2) + 1 + d\log(3))V_{rm}}\prod_{n=\ell}^{m-1}e^{(rd\log(2) +1)V_{rn}} \leq  \sum_{(V_{rn})_{n=\ell}^{m-1}}e^{(rd\log(2) + 1 + d\log(3))V}
\end{align*}

As the number of solutions of \eqref{Eq: sum.of.V_rn} is bounded by $2^V$, we conclude that 
\begin{equation*}
     |\mathcal{F}^\ell_{V}| \leq \sum_{m=\ell}^{V+\ell}|\mathcal{F}_{V,m}^\ell| \leq  V2^Ve^{(rd\log(2) + 1 + d\log(3))V},
\end{equation*}
therefore equation \eqref{Eq: Bound.on.F_V} holds for $b_5\coloneqq [rd +1]\log(2) + 2 + d\log(3)$.
\end{proof}

With these propositions we can control the number of coverings of contours at a given scale, that its, we can give an upper bound to $\left|\C_{r\ell}\left( \mathcal{C}_0(n) \right) \right| = \left|\{\C_{r\ell} : \C_{r\ell}=\C_{r\ell}(\gamma) \ \text{for some }\gamma\in\mathcal{C}_0(n)\} \right|$.

\begin{proposition}\label{Prop: Bound_on_rl_coverings}
    Let $n\geq 0$, $\Lambda\Subset\Z^d$. There exists a constant $b_6\coloneqq b_6(a,d)>0$ such that,
    \begin{equation*}
        |\C_{r\ell}(\mathcal{C}_0(n))|\leq \exp{\left\{ b_6 (\ell\vee 1)^{\kappa+1}\left(\frac{n}{2^{r\frac{a^\prime}{a}\ell}}\vee 1\right)    \right\}}.
    \end{equation*}
\end{proposition}
\begin{proof}
    Proposition \ref{Prop. Bound.on.V_r^l(gamma)} together with Proposition \ref{Prop. Bound.on.C_rl(gamma)} yields, 
    \begin{equation}\label{Eq: Bound_partial_volume_l_between_1_and_j}
        V_r^\ell(\gamma) = V_r^\ell(B_{\C_{r\ell}(\gamma)}) \leq b_3(b_4+b_4^\prime)(\ell\vee 1)^{\kappa+1}\left(\frac{n}{2^{r\frac{a^\prime}{a}\ell}}\vee 1\right) =: R_{n,\ell}.
    \end{equation}
    Therefore, 
        \begin{equation*}
        \left\{\C_{r\ell} :\C_{r\ell}=\C_{r\ell}(\gamma) \ \textrm{for some }\gamma\in\mathcal{C}_0(n)\right\}\subset \bigcup_{V=1}^{\left\lceil R_{n,\ell}\right\rceil}\mathcal{F}^\ell_V
    \end{equation*}
    and Proposition \ref{Prop. Bound.on.Fv} yields 
    \begin{align*}
        |\left\{\C_{r\ell} :\C_{r\ell}=\C_{r\ell}(\gamma) \ \textrm{for some }\gamma\in\mathcal{C}_0(n)\right\}| &\leq \sum_{V=1}^{\left\lceil R_{n,\ell}\right\rceil}|\mathcal{F}^\ell_V|  
        \leq \exp{\left\{ 2b_5b_3(b_4+b_4^\prime)(\ell\vee 1)^{\kappa+1}\left(\frac{n}{2^{r\frac{a^\prime}{a}\ell}}\vee 1\right)\right\}}.
    \end{align*}
    This concludes the proof for $b_6\coloneqq 2b_5b_3(b_4 + b_4^\prime)$.
\end{proof}

A consequence of Proposition \ref{Prop: Bound_on_rl_coverings} is that we get an exponential bound on the number of contours with a fixed size.

\begin{corollary}\label{Cor: Bound_on_C_0_n}
	Let $d\ge 2$, and $\Lambda\Subset \mathbb{Z}^d$. For all $n\geq 1$, $|\mathcal{C}_0(n)| \leq e^{b_6 n}$.  
\end{corollary}

We are finally ready to upper bound the number of admissible regions $ |B_\ell(\mathcal{C}_0(n))|$ at scale $r\ell$.
\begin{proposition}\label{Prop: Bound_on_boundary_of_admissible_sets}
    Let $n\geq 0$, $\Lambda\Subset\Z^d$ and $1\leq\ell\leq \log_{2^r}(b_1 n)(d-1)^{-1}$. There exists a constant $c_4\coloneqq c_4(\alpha, d)$ such that,
    \begin{equation}\label{Eq: Bound_on_boundary_of_admissible_sets}
        |B_\ell(\mathcal{C}_0(n))|\leq \exp{\left\{c_4 \frac{\ell^{\kappa + 1} n}{2^{r\ell(d-1)}}   \right\}}.
    \end{equation}
\end{proposition}

\begin{proof}
    The upper bound on $\ell$ may seem artificial, but Remark \ref{rmk: Upper_bound_on_ell} shows that this is not the case. Remember that $|B_\ell(\mathcal{C}_0(n))|=|\partial\mathfrak{C}_\ell(\mathcal{C}_0(n))|$. Moreover, given $\{C_{r\ell},C_{r\ell}^\prime\}\in\partial\mathfrak{C}_\ell(\gamma)$, either $C_{r\ell}\in\fint\mathfrak{C}_{\ell}$ or $C_{r\ell}^\prime\in\fint\mathfrak{C}_{\ell}$. Using that $\sum_{k=0}^p\binom{p}{k} = 2^{p}$, we have
    \begin{equation}\label{Eq: Replacing_edge_by_inner_boundary}
       \begin{split} |\partial\mathfrak{C}_\ell(\mathcal{C}_0(n))| &= \sum_{\fint\C_{r\ell}\in \fint\mathfrak{C}_{\ell}(\mathcal{C}_0(n))} |\{\partial\C_{r\ell}^\prime : \fint\C_{r\ell}^\prime = \fint\C_{r\ell}\}| \\
        &\leq \sum_{\fint\C_{r\ell}\in \fint\mathfrak{C}_{\ell}(\mathcal{C}_0(n))} \sum_{k=1}^{2d|\fint\C_{r\ell}|}\binom{2d|\fint\C_{r\ell}|}{k}  \\
        &\leq\sum_{\fint\C_{r\ell}\in \fint\mathfrak{C}_{\ell}(\mathcal{C}_0(n))} 2^{2d|\fint\C_{r\ell}|}\leq |\fint\mathfrak{C}_\ell(\mathcal{C}_0(n))|e^{\log(2)2db_1\frac{n}{2^{r\ell(d-1)}}},
        \end{split}
    \end{equation}
     where in the last inequality we applied Proposition \ref{Proposition1}.  For every $L\geq \ell$ and an arbitrary collection $\C_{rL}$, define $\overline{\C_{rL}} = \C_{rL}\cup \{C_{rL}^\prime : \exists C_{rL}\in\C_{rL} \text{ such that } C_{rL}^\prime \text{ shares a face with } C_{rL}\}$. 
     
     Given $C_{r\ell}\in \fint \mathfrak{C}_\ell(\gamma)$, either $C_{r\ell}$ or one of its neighbouring cubes intersects $\Sp(\gamma)$. Hence, for any $L\geq \ell$, $\fint \mathfrak{C}_{\ell}(\gamma)\preceq \overline{\C_{rL}(\gamma)}$. Moreover, the number of $r\ell$-cubes inside a collection $\overline{\C_{rL}(\gamma)}$ of $rL$-cubes is bound by $|\overline{\C_{rL}(\gamma)}|2^{rd(L-\ell)} \leq 2d|\C_{rL}(\gamma)|2^{rd(L-\ell)}$. Using again Proposition \ref{Proposition1}, we can bound 
    \begin{equation}\label{Eq: Bound_on_internal_boundary}
    \begin{split}
           |\fint \mathfrak{C}_{\ell}(\mathcal{C}_0(n))| &\leq  \sum_{\substack{\C_{rL} \in \C_{rL}(\mathcal{C}_0(n))}}\sum_{k=0}^{\left\lceil{\frac{b_1n}{2^{r\ell(d-1)}}}\right\rceil}\binom{2d|\C_{rL}|2^{rd(L-\ell)}}{k} \\
           &\leq   \sum_{\substack{\C_{rL} \in \C_{rL}(\mathcal{C}_0(n))}}\left(\frac{e2d|\C_{rL}|2^{rdL}}{{b_1n}{2^{r\ell}}}\right)^{\frac{2b_1n}{2^{r\ell(d-1)}}},
           \end{split}
    \end{equation}
    where in the last equation we used that, for any $0<M\leq N$, $\sum_{p=0}^{M}\binom{N}{p}\leq \left(\frac{eN}{M}\right)^{M}$. Moreover, the restriction $\ell\leq \log_{2^r}(b_1 n)(d-1)^{-1}$ gives us $1\leq \frac{b_1 n}{2^{r\ell(d-1)}}$, so we bounded $\left\lceil \frac{b_1n}{2^{r\ell(d-1)}} \right\rceil \leq \frac{2b_1n}{2^{r\ell(d-1)}}$. Given a scale $\ell$, we choose $L(\ell) \coloneqq \left\lfloor \frac{a(d-1)\ell}{a^\prime} \right\rfloor$. The restriction $\ell\leq \log_{2^r}(b_1 n)(d-1)^{-1}$ allow us to bound $\left(\frac{n}{2^{r\frac{a^\prime}{a} L(\ell)}} \vee 1\right)\leq b_1(2^{r\frac{a^\prime}{a}}\vee 1)\frac{n}{2^{(d-1)r\ell}}$, so, for any $ \C_{rL(\ell)} \in \C_{rL(\ell)}(\mathcal{C}_0(n))$, Proposition \ref{Prop. Bound.on.C_rl(gamma)} yields
    \begin{align*}
        |\C_{rL(\ell)}|2^{rdL(\ell)} &\leq (b_4+b_4^\prime)L(\ell)^{\kappa}\left(\frac{n}{2^{r\frac{a^\prime}{a} L(\ell)}}\vee 1\right) 2^{rd\frac{a(d-1)}{a^\prime}\ell} \\
        &\leq b_1 (b_4+b_4^\prime)\left(\frac{(d-1)a}{a^\prime}\right)^{\kappa}\left\lceil{2^{r\frac{a^\prime}{a}}}\right\rceil{\ell^{\kappa}}{2^{(d-1)r(\frac{ad}{a^\prime} -1)\ell}}n,
    \end{align*}
    hence 
    \begin{align*}
        \left(\frac{e2d|\C_{rL(\ell)}|2^{rdL(\ell)}}{{b_1n}{2^{r\ell}}}\right)^{\frac{2b_1n}{2^{r\ell(d-1)}}} &\leq \left(\frac{e2db_1 (b_4+b_4^\prime)\left(\frac{(d-1)a}{a^\prime}\right)^{\kappa}\left\lceil{2^{r\frac{a^\prime}{a}}}\right\rceil{\ell^{\kappa}}{2^{(d-1)r(\frac{ad}{a^\prime} -1)\ell}}n}{{b_1n}{2^{r\ell}}}\right)^{\frac{2b_1n}{2^{r\ell(d-1)}}}\\
        &\leq \left({e2d (b_4+b_4^\prime)\left(\frac{a(d-1)}{a^\prime}\right)^{\kappa}\left\lceil{2^{r\frac{a^\prime}{a}}}\right\rceil{\ell^{\kappa}}{2^{[(d-1)(\frac{ad}{a^\prime} -1) - 1]r\ell}}}\right)^{\frac{2b_1n}{2^{r\ell(d-1)}}}\\
        &\leq \exp\left\{ c_4^\prime \frac{\ell n}{2^{r\ell(d-1)}} \right\},
    \end{align*}
    with $c_4^\prime = [1  + \log(2d (b_4+b_4^\prime)\left(\frac{(d-1)a}{a^\prime}\right)^{\kappa}\left\lceil{2^{r\frac{a^\prime}{a}}}\right\rceil) + \kappa  + ((d-1)(\frac{ad}{a^\prime} -1) -1)\log(2)r]2b_1$. Moreover, by Proposition \ref{Prop: Bound_on_rl_coverings},
    \begin{align*}
         |\C_{rL(\ell)}(\mathcal{C}_0(n))| &\leq  \exp{\left\{ b_6L(\ell)^{\kappa+1} \left(\frac{n}{2^{r\frac{a^\prime}{a} L(\ell)}} \vee 1 \right) \right\}} \leq  \exp{\left\{ b_6 b_1\left(\frac{a(d-1)}{a^\prime}\right)^{\kappa +1}\left\lceil2^{r\frac{a^\prime}{a}}\right\rceil\frac{\ell^{\kappa+1}n}{2^{(d-1)r\ell}}\right\}}   \\
    \end{align*}
    so equations \eqref{Eq: Replacing_edge_by_inner_boundary} and \eqref{Eq: Bound_on_internal_boundary} yield
    \begin{align}\label{Eq: Bound_on_boundary_of_admissible_sets_aux_1}
         |\partial\mathfrak{C}_\ell(\mathcal{C}_0(n))| &\leq  \exp{\left\{b_6 b_1\left(\frac{a(d-1)}{a^\prime}\right)^{\kappa +1}\left\lceil2^{r\frac{a^\prime}{a}}\right\rceil\frac{\ell^{\kappa+1}n}{2^{r\ell(d-1)}} + c_4^\prime \frac{\ell n}{2^{r\ell(d-1)}} + \log(2)2db_1\frac{n}{2^{r\ell(d-1)}} \right\}}.
    \end{align}

  that concludes our proof taking $c_4\coloneqq b_6 b_1\left(\frac{a(d-1)}{a^\prime}\right)^{\kappa +1}\left\lceil2^{r\frac{a^\prime}{a}}\right\rceil + c_4^{\prime} + \log(2)2db_1$.
    
\end{proof}

\begin{remark}\label{Rmk: Adaptation_for_Mar_partition}
    Using the notion of long-range contours of \cite{Affonso.2021}, we can get a worse upper bound on $|B_\ell(\mathcal{C}_0(n))|$ that is still good enough to prove phase transition in $d\geq 3$. Using Proposition \ref{Prop. Bound.on.C_rl(gamma)_Lucas}, we can prove in the same steps as Proposition \ref{Prop: Bound_on_rl_coverings} that $|\C_{r\ell}(\mathcal{C}_0(n))|\leq b_4^{\prime\prime}n \ell^{-\frac{r - d- 1 - \log_2(a)}{\log_2(a)}}$. With this, we can proceed similarly as in the proof of Proposition \ref{Prop: Bound_on_boundary_of_admissible_sets} but now choosing $L(\ell) = 2^{2r\left\lfloor \frac{\log_2(a)\ell}{r - d -1 - \log_2(a)} \right\rfloor}$, which gives us the bound 
    \begin{equation*}
        |B_\ell(\mathcal{C}_0(n))| \leq \exp{\left\{ c_4^{\prime}\frac{n}{2^{r\ell(d-1-\frac{2\log_2(a)}{r - d -1 - \log_2(a)})}} \right\}}.
    \end{equation*}
    For $r$ large enough, $d-1-\frac{2\log_2(a)}{r - d -1 - \log_2(a)}>1$ and the proof of Proposition \ref{Prop: Bound.gamma_2} follows with small adaptations. 
\end{remark}

At last, we prove the main proposition of this section. 

\subsubsection*{Proof of Proposition \ref{Prop: Bound.gamma_2}}

 As $N(\I_-(n), \d_2, \epsilon)$ is decreasing in $\epsilon$, we can use Dudley's entropy bound to get
\begin{multline*}
      {\mathbb{E}\left[\sup_{\gamma\in\mathcal{C}_0(n)}{\Delta_{\I_-(\gamma)}(h)}\right]} \leq 4\varepsilon b_3 \overline{L} n^{\frac{1}{2}} \sqrt{\log{N(\I_-(n), d_2, 0)}} \\ + 4\varepsilon b_3 \overline{L}n^{\frac{1}{2}}\sum_{\ell=0}^\infty (2^{\frac{r(\ell+1)}{2}} - 2^{\frac{r\ell}{2}})\sqrt{\log N(\I_-(n), \d_2,4\varepsilon b_3 2^{\frac{r\ell}{2}}n^{\frac{1}{2}})}.
\end{multline*}
We can bound the first term by noticing that $N(\I_-(n), d_2, 0) = |\I_-(n)|\leq 2^n|\mathcal{C}_0(n)|$, $2^n$ being an upper bound on the number of labels given a fixed support. By Corollary \ref{Cor: Bound_on_C_0_n}, $|\mathcal{C}_0(n)| \leq e^{c_1 n}$, and hence
\begin{equation*}
    4\varepsilon b_3 \overline{L} n^{\frac{1}{2}}\sqrt{\log{N(\I_-(n), d_2, 0)}} \leq 4\varepsilon b_3 \overline{L} (c_1 + \log 2)^{\frac{1}{2}} n.
\end{equation*}

Since $\d_2(\I_-(\gamma_1),\I_-(\gamma_2))\leq 2\varepsilon\sqrt{|\I_-(\gamma_1)| + |\I_-(\gamma_2)|}\leq 2\sqrt{2}\varepsilon n^{\frac{1}{2} + \frac{1}{2(d-1)}}$ for any $\gamma_1,\gamma_2\in\mathcal{C}_0(n)$, when $4\varepsilon b_3 \overline{L} 2^{\frac{r\ell}{2}}n^{\frac{1}{2}}\geq 2\sqrt{2}\varepsilon n^{\frac{1}{2} + \frac{1}{2(d-1)}}$, only one ball covers all interiors, hence all the terms in the sum above with $\ell > k(n)\coloneqq \floor{\frac{\log_{2^r}(n)}{(d-1)}}$ are zero. As $N(\I_-(n), \d_2,\varepsilon b_3 2^{\frac{r\ell}{2}}n^{\frac{1}{2}})\leq |B_{\ell}(\mathcal{C}_0(n))|$, see Remark \ref{Rmk: Bounding_N_by_B_ell}, using Proposition \ref{Prop: Bound_on_boundary_of_admissible_sets} we get
\begin{align*}
        {\mathbb{E}\left[\sup_{\gamma\in\mathcal{C}_0(n)}{\Delta_{\I(\gamma)}(h)}\right]} &\leq 4\varepsilon b_3 \overline{L} 2^{\frac{r}{2}}\sqrt{c_4} n^{\frac{1}{2}}\sum_{\ell=1}^{k(n)}2^{\frac{r\ell}{2}}\sqrt{\frac{\ell^{\kappa + 1} n }{2^{r(d-1)\ell}}} +  4\varepsilon b_3 \overline{L} (c_1 + \log 2)^{\frac{1}{2}} n\nonumber\\
        &\leq  4\varepsilon b_3 \overline{L} 2^{\frac{r}{2}} \sqrt{c_4}\left[ (c_1 + \log 2)^{\frac{1}{2}} + \sum_{\ell=1}^{\infty}\left(\frac{\ell^\frac{\kappa+1}{2}}{2^{\frac{r\ell(d-2)}{2}}} \right)\right]n.
\end{align*}

The series above converges for any $d\geq 3$, and we conclude that 
\begin{equation*}
       {\mathbb{E}\left[\sup_{\gamma\in\mathcal{C}_0(n)}{\Delta_{\I_-(\gamma)}(h)}\right]} \leq \varepsilon L_1^\prime n,
\end{equation*}
with $L_1^\prime\coloneqq   4 b_3 \overline{L} 2^{\frac{r}{2}}\sqrt{c_4}\left[ (c_1 + \log 2)^{\frac{1}{2}} + \sum_{\ell=1}^{\infty}\left(\frac{\ell^\frac{\kappa+1}{2}}{2^{\frac{r\ell(d-2)}{2}}} \right)\right]$. The desired result follows from Theorem \ref{MMT} taking the constant $L_1 \coloneqq L L_1^\prime$.

\qed

\section{Phase transition}   

\begin{theorem}\label{Theo: Transicao_de_fase}
For $d\geq 3$ and $\alpha>d$, there exists a constant $C\coloneqq C(d,\alpha)$ such that, for all $\beta>0$ and $e\leq C$, the event 
    \begin{equation}\label{Eq: PTLR}
        \nu_{\Lambda; \beta, \varepsilon h}^+(\sigma_0 = -1) \leq e^{-C\beta} + e^{-C/\varepsilon^2} 
    \end{equation}
    has $\mathbb{P}$-probability bigger then $1 - e^{-C\beta} - e^{-C/\varepsilon^2}$.\\
    
In particular, for $\beta>\beta_c$ and $\varepsilon$ small enough, there is phase transition for the long-range Ising model.  
\end{theorem}

\begin{proof}
        The proof is an application of the Peierls' argument, but now on the joint measure $\mathbb{Q}$. Define $\mathcal{E} = \mathcal{E}_0 \cap \mathcal{E}_1$. By Proposition \ref{Prop: Bound.bad.event.0} and Proposition \ref{Prop: Bound.bad.event.1}, we have
        \begin{align}\label{Eq: Upper.bound.on.Q.1}
            \mathbb{Q}_{\Lambda; \beta, \varepsilon}^+(\sigma_0 = -1) &=  \mathbb{Q}_{\Lambda; \beta, \varepsilon}^+(\{\sigma_0 = -1\} \cap \mathcal{E}_0) + \mathbb{Q}_{\Lambda; \beta, \varepsilon}^+(\{\sigma_0 = -1\}\cap \mathcal{E}_0^c) \nonumber \\
            & \leq \mathbb{Q}_{\Lambda; \beta, \varepsilon}^+(\{\sigma_0 = -1\} \cap \mathcal{E}_0) +  e^{-C_0/\varepsilon^2} \nonumber \\
            & \leq \mathbb{Q}_{\Lambda; \beta, \varepsilon}^+(\{\sigma_0 = -1\} \cap \mathcal{E}) + \mathbb{Q}_{\Lambda; \beta, \varepsilon}^+(\{\sigma_0 = -1\}\cap \mathcal{E}_0 \cap \mathcal{E}_{1}^c)  + e^{-C_0/\varepsilon^2} \nonumber \\
            & \leq \mathbb{Q}_{\Lambda; \beta, \varepsilon}^+(\{\sigma_0 = -1\} \cap \mathcal{E}) + e^{-C_1/\varepsilon^2}  + e^{-C_0/\varepsilon^2},
        \end{align}
since $\mathbb{Q}_{\Lambda; \beta, \varepsilon}^+(\{\sigma_0 = -1\}\cap \mathcal{E}_0^c) \leq \mathbb{Q}_{\Lambda; \beta, \varepsilon}^+(\mathcal{E}_0^c) = \mathbb{P}(\mathcal{E}_0^c)$ and, analogously, ${\mathbb{Q}_{\Lambda; \beta, \varepsilon}^+(\{\sigma_0 = -1\} \cap \mathcal{E}_0 \cap \mathcal{E}_{1}^c) \leq \mathbb{P}(\mathcal{E}_1^c)}$.  When $\sigma_0 = -1$, there must exist a contour $\gamma$ with $0\in V(\gamma)$, hence
\begin{equation*}
    \nu_{\Lambda; \beta, \varepsilon h}^+(\sigma_0 = -1) \leq \sum_{\gamma \in \mathcal{C}_0}\nu_{\Lambda; \beta, \varepsilon h}^+(\Omega(\gamma)),
\end{equation*}
where $\Omega(\gamma) \coloneqq \{\sigma\in\Omega : \Sp(\gamma) \subset \Gamma(\sigma)\}$. So we can write

\begin{align}\label{Eq: Upper.bound.on.Q.2}
    \mathbb{Q}_{\Lambda; \beta, \varepsilon}^+(\{\sigma_0 = -1\} \cap \mathcal{E}) &= \int_{\mathcal{E}}\sum_{\sigma : \sigma_0 = -1}g_{\Lambda; \beta, \varepsilon}^+(\sigma, h)dh \nonumber \\
    &\leq  \sum_{\gamma\in\mathcal{C}_0} \int_{\mathcal{E}}\sum_{\sigma\in\Omega(\gamma)}g_{\Lambda; \beta, \varepsilon}^+(\sigma, h)dh \nonumber \\
    &\leq  \sum_{\gamma \in \mathcal{C}_0} \frac{2^{|\gamma|}\int_{\mathcal{E}}\sum_{\sigma\in\Omega(\gamma)}g_{\Lambda; \beta, \varepsilon}^+(\sigma, h)dh}{\int_{\mathcal{E}}\sum_{\sigma\in\Omega(\gamma)}g_{\Lambda; \beta, \varepsilon}^+(\tau_{\gamma}(\sigma), \tau_{\I_-(\gamma)}(h))dh} \nonumber \\
    & \leq \sum_{\gamma\in\mathcal{C}_0}2^{|\gamma|} \sup_{\substack{h\in\mathcal{E}\\ \sigma\in\Omega(\gamma)}}\frac{g_{\Lambda; \beta, \varepsilon}^+(\sigma, h)}{g_{\Lambda; \beta, \varepsilon}^+(\tau_{\gamma}(\sigma), \tau_{\I_-(\gamma)}(h))}. 
\end{align}

In the third equation, we used that $\int_{\mathcal{E}}\sum_{\sigma\in\Omega(\gamma)}g_{\Lambda; \beta, \varepsilon}^+(\tau_{\gamma, \sigma}(\sigma), \tau_{\I_-(\gamma)}(h))dh \leq 2^{|\gamma|}$, since the number of configurations that are incorrect in $\Sp(\gamma)$ are bounded by $2^{|\gamma|}$. By \eqref{Eq: quotient.of.gs} and the definition of the event $\mathcal{E}$, 
\begin{align}\label{Eq: Upper.bound.on.Q.3}
    \sup_{\substack{h\in\mathcal{E}\\ \sigma\in\Omega(\gamma)}}\frac{g_{\Lambda; \beta, \varepsilon}^+(\sigma, h)}{g_{\Lambda; \beta, \varepsilon}^+(\tau_{\gamma, \sigma}(\sigma), \tau_{\I_-(\gamma)}(h))} &\leq \sup_{\substack{h\in\mathcal{E}\\ \sigma\in\Omega(\gamma)}}  \exp{\{{- \beta c_2 |\gamma| -2\beta\sum_{x\in \Sp^-(\gamma)}\varepsilon h_x}\}}\frac{Z_{\Lambda; \beta, \varepsilon}^{+}(\tau_{\I_-(\gamma)}(h))}{Z_{\Lambda; \beta, \varepsilon}^{+}(h)} \nonumber\\
    &= \sup_{\substack{h\in\mathcal{E}\\ \sigma\in\Omega(\gamma)}}  \exp{\{{- \beta c_2 |\gamma| -2\beta\sum_{x\in \Sp^-(\gamma)}\varepsilon h_x + \beta \Delta_{\gamma}(h)}\}} \nonumber\\
    &\leq  \exp{\{{- \beta \frac{c_2}{2} |\gamma| }\}},
\end{align}
since $\Delta_{\gamma}(h) -2\beta\sum_{x\in \Sp^-(\gamma)}\varepsilon h_x \leq \frac{c_2}{2}|\gamma|$, for all $h\in\mathcal{E}$. Equations \eqref{Eq: Upper.bound.on.Q.1}, \eqref{Eq: Upper.bound.on.Q.2} and \eqref{Eq: Upper.bound.on.Q.3} yields
\begin{align*}
     \mathbb{Q}_{\Lambda; \beta, \varepsilon}^+(\sigma_0 = -1) &\leq  \sum_{\substack{\gamma\in \mathcal{E}_\Lambda^+\\ 0\in V(\gamma)}} 2^{|\gamma|}\exp{\{{- \beta \frac{c_2}{2} |\gamma| }\}} + e^{-C_1/\varepsilon^2} + e^{-C_0/\varepsilon^2}\\
     &\leq \sum_{n\geq 1}\sum_{\substack{\gamma\in \mathcal{E}_\Lambda^+, |\gamma|=n \\ 0\in V(\gamma)}} \exp{\{{(-\beta \frac{c_2}{2} + \log2)n}\}} + e^{-C_1/\varepsilon^2} + e^{-C_0/\varepsilon^2}\\
     &\leq \sum_{n\geq 1}|\mathcal{C}_0(n)| \exp{\{{(-\beta \frac{c_2}{2} +\log2)n}\}} + + e^{-C_1/\varepsilon^2} + e^{-C_0/\varepsilon^2}\\
    &\leq \sum_{n\geq 1} e^{(c_1 -\beta \frac{c_2}{2} +\log2)n} + e^{-C_1/\varepsilon^2} + e^{-C_0/\varepsilon^2}.
\end{align*}
When $\beta$ is large enough, the sum above converges and there exists a constant $C$ such that   
\begin{equation*}
    \mathbb{Q}_{\Lambda; \beta, \varepsilon}^+(\sigma_0 = -1) \leq e^{-\beta 2C} + e^{-2C / \varepsilon^2}.
\end{equation*}
The Markov Inequality finally yields
\begin{align*}
    \mathbb{P}\left( \nu_{\Lambda; \beta, \varepsilon h}^+(\sigma_0 = -1) \geq e^{-C\beta} + e^{-C/\varepsilon^2}\right) &\leq \frac{\mathbb{Q}_{\Lambda; \beta, \varepsilon}^+(\sigma_0 = -1)}{e^{-C\beta} - e^{-C/\varepsilon^2}} \\
    &\leq \frac{e^{-\beta 2C} + e^{-2C / \varepsilon^2}}{e^{-C\beta} + e^{-C/\varepsilon^2}} \leq e^{-C\beta} + e^{-C/\varepsilon^2},
\end{align*}
what proves our claim.
\end{proof}

\section{Concluding Remarks}

In this paper, we proved phase transition for the long-range random field Ising model in $d\geq 3$ and $\alpha >d$, by following a new method of proving phase transition introduced by Ding and Zhuang \cite{Ding2021}, and using a modification of multidimensional contours defined in \cite{Affonso.2021}. The key part of the argument was to extend the results of \cite{FFS84} to contours that are not necessarily connected. This proof can be extended to other models with a contour system, as long as the probability of the event $\mathcal{E}^c$ decreases to zero for large $\varepsilon$.

The results presented by Bricmont and Kupiainen \cite{Bricmont.Kupiainen.88} are more general than ours since they only need the external field to be symmetric around zero and have a sub-Gaussian tail. In \cite{Ding2021}, Ding and Zhuang claim that it should be possible, with more care, to extend their results to an external field in the same generality.

Bricmont and Kupiainen also study the decay of correlation for the random field Ising model, and so far as we know, the only result for the long-range RFIM is \cite{Klein.Masooman.97}, for high temperature.   


\section*{Acknowledgements}

This study was financed, in part, by the S\~{a}o Paulo Research Foundation (FAPESP), Brazil. Process Numbers 2016/25053-8, 2017/18152-2, 2018/26698-8, 2020/14563-0, and 2023/00854-1. RB is supported by CNPq grants 311658/2025-3, 312294/2018-2 and 408851/2018-0, and by the University Center of Excellence \textquotedblleft Dynamics, Mathematical Analysis and Artificial Intelligence\textquotedblright, at the Nicolaus Copernicus University.  The authors thank Kelvyn Welsch and Gilberto Ara\'ujo for carefully reading the previous versions of the paper and, in particular, for Kelvyn pointing out the condition \textbf{(A1)}. We thank Amnon Aharony for pointing out the reference \cite{Aharony_Imry_Ma_76}. We thank Abel Klein for discussions about earlier results in the literature concerning multiscale analysis. JM and LA are very grateful to Eric Endo for the support and hospitality during their first visit to China and NYU-Shanghai; they also thank Weijun Xu, as well as Jian Ding, for the support and hospitality during their visit to Peking University, especially to Professor Ding for fruitful discussions.
	
\bibliographystyle{habbrv} 
\bibliography{bib} 

\end{document}